%% file: main.tex
\documentclass[11pt]{article}
\usepackage[numbers]{natbib}
\usepackage{amsmath}
\usepackage{amssymb}
\usepackage{amsthm}
\usepackage{xcolor}
\usepackage{graphicx}
\usepackage{algorithm}
\usepackage{algpseudocode}
\usepackage{enumitem}
\usepackage{array}
\usepackage{booktabs}
\usepackage{longtable}
\usepackage[margin=1in]{geometry}
\usepackage{hyperref}
\hypersetup{colorlinks=true,citecolor=blue,linkcolor=blue,urlcolor=blue}
\allowdisplaybreaks
\input{_0_vars}

\input{_1_maths}
\theoremstyle{plain}
\newtheorem{theorem}{Theorem}[section]
\newtheorem{lemma}[theorem]{Lemma}

\theoremstyle{definition}
\newtheorem{definition}[theorem]{Definition}

\begin{document}
\date{September 11, 2026}
\title{\paperTitle}
\author{\paperAuthor}
\maketitle
\begin{abstract}
\input{00_abstract}
\end{abstract}
\newpage
\input{_2_body}

\section*{Acknowledgements and AI Disclosure}
In the August 29, 2026 version of this draft (\url{10.5281/zenodo.22239302}), the author used Gemini Pro 3.1, Codex 5.6 Sol, and Claude Code Fable 5 to assist with language editing and grammar checking. In that version, the author was only able to obtain results for $\epsilon$ around $10^{-30}$. In the current version, the author used Codex 6, ChatGPT Pro 6, and Fable 5.1 to further improve $\epsilon$ to around $10^{-18}$. The main ideas for improving the paper came from AI tools. All of the proofs have been completely rewritten by the author. The author takes full responsibility for the correctness of the proofs.
\bibliographystyle{alpha}
\bibliography{ref}
\end{document}

%% file: _0_vars.tex
\newcommand{\paperTitle}{A Sharper Explicit Bound on the Subtour-LP Integrality Gap for Metric TSP}
\newcommand{\paperAuthor}{Zhao Song\thanks{\texttt{magic.linuxkde@gmail.com}.  The author would lilke to thank Victor Reis and Lichen Zhang for very helpful discussions.  This is version 2 of this draft. We provide a detailed comparison between the two versions in the Acknowledgments section of this paper.
} }

%% file: _1_maths.tex
\DeclareMathOperator{\E}{\mathbb{E}}

%% file: 00_abstract.tex
Karlin, Klein, and Oveis Gharan introduced a randomized better-than-$3/2$ approximation algorithm for metric TSP \cite{kko21} and subsequently established the corresponding improvement in the integrality gap of the subtour-elimination LP \cite{kko22}, with an explicit constant $\varepsilon>1.00000\cdot10^{-36}$. Gurvits, Klein, and Leake subsequently improved the certified saving to $2.18000\cdot10^{-34}$ \cite{gkl24}. In this paper, we obtain a randomized polynomial-time $(3/2-\varepsilon)$-approximation for every fixed $0<\varepsilon<\varepsilon_\star$, where $\varepsilon_\star>2.78621\cdot10^{-18}$, and consequently the subtour-elimination LP has integrality gap at most $3/2-\varepsilon_\star$.
The classical worst-case integrality-gap lower bound is $4/3$~\cite{williamson90}.

%% file: _2_body.tex
\input{01_intro}
\input{02_tech}
\input{10_preli}
\input{40_proof}
\input{60_improve}

%% file: 01_intro.tex
\section{Introduction}

The traveling salesperson problem is among the oldest and most extensively
studied problems in combinatorial optimization; formulations of the routing
question go back to the nineteenth century~\cite[Chapter~1]{applegate2007}.
The modern polyhedral theory began with the cutting-plane work of Dantzig,
Fulkerson, and Johnson~\cite{dfj54}.  Held and Karp subsequently developed the
linear-programming lower bound now called the subtour-elimination or Held--Karp
relaxation~\cite{heldkarp70}.  Metric TSP remains
APX-hard~\cite{karpinski15}.  In particular, it is NP-hard to approximate
within $123/122$.  Thus the
central algorithmic question has long been how closely a polynomial-time
algorithm can approach the optimum.

The classical $3/2$-approximation was discovered independently by Christofides
and Serdyukov in the 1970s~\cite{christofides76,serdyukov78}.  Their
method combines a minimum spanning tree with a minimum-cost parity correction
on its odd-degree vertices.  Wolsey later showed that the same $3/2$ guarantee
holds relative to the subtour LP~\cite{wolsey80}, which also established the
long-standing upper bound $3/2$ on its integrality gap.  The widely studied
$4/3$ conjecture predicts that the true integrality gap is $4/3$.  A related
conjecture of Schalekamp, Williamson, and Zuylen asserts that the largest
gap is attained by a half-integral vertex of the subtour polytope
~\cite{schalekamp14}; this made half-integral instances a particularly
important testing ground for progress below $3/2$.

Substantial improvements were obtained earlier for restricted metrics.
Euclidean TSP admits polynomial-time approximation schemes due independently
to Arora~\cite{arora98} and Mitchell~\cite{mitchell99}.  For graph metrics,
Oveis Gharan, Saberi, and Singh introduced a randomized rounding method based
on maximum-entropy spanning trees and obtained the first
$(3/2-\epsilon_0)$-approximation~\cite{oss11}.  M\"omke and Svensson then gave
a combinatorial $1.461$-approximation~\cite{momke_svensson11}; Mucha improved
the factor to $13/9$~\cite{mucha12}, and Seb\H{o} and Vygen reached
$7/5$~\cite{sebo_vygen14}.  This graphic-TSP line demonstrated that parity
correction could be made cheaper than half the optimum, while the
maximum-entropy approach suggested a possible route for general metrics.

For instances with a half-integral optimum subtour-LP solution,
\cite{kko20} gives a $1.49993$-approximation.  For general metric TSP,
\cite{kko21} proves the first randomized approximation with ratio
$3/2-\epsilon$ and gives the explicit scale
$\epsilon>1.00000\cdot10^{-36}$.  The subtour-LP
analysis in \cite{kko22} proves that its integrality gap is strictly below
$3/2$ and also gives a better-than-$3/2$ result for the 2-edge-connected
multi-subgraph problem.  The conditional-expectation method in \cite{kko23}
derandomizes the max-entropy algorithm.

Gurvits, Klein, and Leake later revisited the probabilistic part of the KKO
analysis using new capacity bounds for real stable polynomials~\cite{gkl24}.
Their estimates raised the controlling common-event probability from
$2\cdot10^{-10}$ to $1.5\cdot10^{-9}$~\cite[Section~4.2]{gkl24} and yielded
the explicit approximation constant
$\epsilon=2.18000\cdot10^{-34}$~\cite[Corollary~4.6]{gkl24}.  This refinement
preserved the max-entropy algorithm and the combinatorial framework while
sharpening the quantitative probability bounds used in its analysis.

\subsection{Our result}
We state our result as follows.

\begin{theorem}[Informal version of Theorem~\ref{thm:our_main_formal}]
\label{thm:our_main_informal}
For some absolute constant
$\epsilon>2.78621\cdot10^{-18}$, there exists a randomized polynomial-time
algorithm that, on every metric TSP instance, returns a tour whose expected
cost is at most
$(3/2-\epsilon)\operatorname{OPT}_{\mathrm{LP}}$.
Consequently, the integrality gap of the subtour-elimination LP is at most
$3/2-\epsilon$.
\end{theorem}

\begin{table}[ht]
\centering
\caption{Upper and lower bounds on the integrality gap of the subtour LP
for metric TSP. Part (a) compares the certified savings $\epsilon$ in the
upper bound $3/2-\epsilon$. Part (b) records the classical lower bound,
approached by a family of graphic metric instances.}
\label{tab:main_theorems}
{\bf (a) Upper bounds: certified savings below $3/2$.}\par\smallskip
\begin{tabular}{|l|l|l|l|}
\toprule
Year & Authors & Reference & Certified lower bound on $\epsilon$ \\
\midrule
2022
& Karlin, Klein, and Oveis Gharan
& \cite{kko22}
& $>1.00000\cdot10^{-36}$ \\
2024
& Gurvits, Klein, and Leake
& \cite{gkl24}
& $\geq2.18000\cdot10^{-34}$ \\
2026
& This work
& Theorem~\ref{thm:our_main_informal}
& $>2.78621\cdot10^{-18}$ \\
\bottomrule
\end{tabular}
\par\medskip
{\bf (b) Lower bound on the integrality gap.}\par\smallskip
\begin{tabular}{|l|l|l|l|}
\toprule
Year & Authors & Reference & Integrality-gap lower bound \\
\midrule
1990
& Williamson
& \cite[p.~35, Figure~2.3]{williamson90}
& $4/3\approx1.33333$ \\
\bottomrule
\end{tabular}
\end{table}

{\bf Organization.} Section~\ref{sec:technique_overview}
reviews the KKO and GKL frameworks and summarizes the refinements that yield
our improved constant. Section~\ref{sec:preliminaries} introduces the notation,
records the max-entropy reduction and its finite-precision guarantee, and gives
the conversion from slack vectors to tours. Section~\ref{sec:proof_improved_theorem}
establishes the probability, payment, and parity-correction estimates, constructs the
layered slack vector, and proves the main result of this paper.

%% file: 02_tech.tex
\section{Technique overview}
\label{sec:technique_overview}

Section~\ref{sec:tech_kko21} reviews the KKO21 framework, including
max-entropy spanning trees, happiness events, and the payment and repair
arguments.
Section~\ref{sec:tech_gkl24} explains how GKL24 sharpens the probability
estimates while retaining this framework.
Section~\ref{sec:tech_ours} summarizes our refined probability, payment, and
repair estimates and explains how they combine to improve the constant.

\subsection{The KKO21 framework}
\label{sec:tech_kko21}

{\bf Step 1: Max-entropy tree and parity correction.}
Here KKO21 refers to the max-entropy framework of Karlin, Klein, and
Oveis Gharan~\cite{kko21}; we also use its subtour-LP extension
~\cite{kko22}.  Start with an optimum subtour-LP solution $x^0$ and perform
the standard vertex-splitting reduction, introducing a zero-cost edge
$e_0=\{u_0,v_0\}$ of value one.  After deleting $e_0$, the remaining vector
$x$ lies in the spanning-tree polytope.
The algorithm samples a maximum-entropy spanning tree $T$ with marginals $x$,
lets $O$ be the set of odd-degree vertices of $T$, and augments $T$ by a
minimum-cost $O$-join.  The zero-cost edge $e_0$ is used only in the LP
reduction and the fractional-join bookkeeping.  Thus the tree has expected
cost $c(x)$, and the
problem is to construct a feasible fractional $O$-join of expected cost
strictly below $c(x)/2$; see \cite[Sections~2 and~6]{kko22}.

{\bf Step 2: Structure of near-minimum cuts.}
After a small coordinatewise decrease from $x/2$, only cuts whose $x$-value
is extremely close to two can violate an odd-cut constraint.  KKO therefore focus on the
$\eta$-near-minimum cuts.  Crossing components of these cuts admit polygon
representations, while the remaining relevant cuts can be organized into a
laminar hierarchy.  This approximate-uncrossing description is what prevents
one edge from being charged by an unbounded number of near-minimum cuts.  It
also separates cuts crossed on both sides, cuts crossed on at most one side,
and the degree, triangle, and outer-polygon cuts appearing in the hierarchy.
The extension in \cite{kko22} replaces charging to edges of an optimum tour by
repair sets extracted from the polygon representation.  Consequently the
entire comparison is against $c(x)$, which is the key step needed for an
integrality-gap bound rather than only an approximation guarantee.

{\bf Step 3: Happiness events and payment.}
The quantitative heart of the method is a random signed slack vector.  A
hierarchy edge bundle is declared \emph{good} when one of several local tree
events has a uniform positive probability.  These events say, for example,
that prescribed boundary sets contain exactly one or two tree edges; they are
called happiness events.  On a happiness event, KKO decrease the slack of the
corresponding bundle.  A max-flow matching distributes the burden of an
upward boundary among good horizontal bundles, and the payment theorem proves
that every good edge receives a negative expected contribution.  The common
probability $p$ is the minimum of six local happiness estimates, so the
payment coefficient $\epsilon_P$ is proportional to $p$.

{\bf Step 4: Feasibility repair and parameter balance.}
Negative slack cannot be used without repairing the odd cuts on which it may
create a deficit.  KKO construct nonnegative repair vectors for cuts crossed
on both sides and for the one-sided polygon configurations.  Their support is
chosen so that each relevant odd cut receives enough positive slack and every
edge has bounded expected repair.  In the notation of
\cite[Theorem~6.1]{kko22}, the signed vector contributes a decrease of order
$\epsilon_P\beta x_e$, whereas the repair costs order $\eta\beta x_e$.
Choosing $\eta$ proportional to $\epsilon_P$ makes the net saving quadratic in
$\epsilon_P$, hence quadratic in $p$.  This architecture---max-entropy tree,
near-minimum-cut hierarchy, happiness events, matching and payment, and
bounded-congestion repair---is the part retained by all subsequent work.

\subsection{The GKL24 refinement}
\label{sec:tech_gkl24}

{\bf Steps 1, 2, and 4 in Section~\ref{sec:tech_kko21} are reused directly.}
GKL use the same max-entropy spanning-tree algorithm and parity correction
from Step~1, the same hierarchy and polygon structure from Step~2, and the
same repair vectors and parameter balancing from Step~4.  Within Step~3 they
also retain the definitions of good bundles and happiness events, as well as
the KKO matching and payment arguments.  Thus GKL do not change the sampled
tree, the combinatorial architecture, or the final $O$-join construction.

{\bf Only the probability certification in Step 3 is modified.}
A maximum-entropy spanning-tree distribution is strongly Rayleigh.  Therefore
its multivariate generating polynomial is real stable.  If
$A_1,\ldots,A_m$ count tree edges in disjoint edge sets, then the probability
of an exact pattern $A_i=\kappa_i$ is a coefficient of this polynomial, while
the expectations of the $A_i$ are entries of its gradient at the all-ones
vector.  GKL prove new capacity bounds that convert control of these
expectations, including the expectations of partial sums, into a lower bound
for the desired coefficient \cite[Theorem~4.1 and Corollary~4.2]{gkl24}.
This replaces several case-specific Bernoulli-sum estimates in KKO by one
multivariate stable-polynomial inequality.  For the harder configurations,
GKL first use capacity to reduce to a structured special case and then combine
it with the original log-concavity and stochastic-dominance arguments.

{\bf How the modified probability argument works.}
The contrast with the KKO probability lemma is especially useful here.  That
lemma turns lower-tail and upper-tail estimates for every partial sum into a
joint exact-count estimate, but its dependence on the tail parameter is
doubly exponential in the number of counted sets.  KKO observed that suitable
expectation bounds imply the required tail estimates, yet the resulting
generic constant was too small for their final analysis, so they proved the
important happiness estimates by ad hoc arguments.  GKL instead show, under
the slightly stronger condition that every partial-sum expectation stays a
fixed distance from the neighboring integer, that the exact-count probability
has a simply exponential lower bound, with the correct dependence on that
distance \cite[Theorems~2.5 and~2.6]{gkl24}.  Their proof productizes a
stable polynomial into affine factors, reduces the associated matrices to
extreme points supported on bipartitioned forests, and inducts on forest
leaves \cite[Section~3]{gkl24}.  Removing the dependence on the total degree
is what makes this capacity bound effective for the TSP events.

{\bf Quantitative effect of the modification.}
For events~2--5, the KKO coefficients of $\epsilon_{1/2}^2$ are
$0.005$, $0.006$, $0.005$, and $0.020$, respectively.  GKL replace them by
$0.039$, $0.038$, $0.0498$, and $0.0485$
\cite[Section~4.2]{gkl24}; the first event is unchanged.  At
$\epsilon_{1/2}=0.0002$, the KKO and GKL lower bounds for the first five
events are, respectively,
\[
\begin{array}{c|ccccc}
 & 1 & 2 & 3 & 4 & 5 \\
\hline
\text{KKO} & 1.50 & 0.20 & 0.24 & 0.20 & 0.80 \\
\text{GKL} & 1.50 & 1.56 & 1.52 & \geq 1.90 & \geq 1.90
\end{array}
\quad\text{times }10^{-9}.
\]
The sixth event is far from the bottleneck; GKL note that its threshold changes
slightly when $p$ is raised and record that it remains above the new common
value.  Thus the minimum rises from $p_{\mathrm{KKO}}=2\cdot10^{-10}$ to
$p_{\mathrm{GKL}}=1.5\cdot10^{-9}$, an improvement by a factor of $7.5$
\cite[Section~4.2]{gkl24}.

{\bf How the modification improves $\epsilon$.}
Nothing after this substitution requires a new TSP construction.  The KKO
payment theorem still makes $\epsilon_P$ linear in $p$, and the final choice
of $\eta$ still makes the approximation saving quadratic in $p$.  Consequently,
the $7.5$-fold improvement in $p$ becomes a factor of $7.5^2=56.25$: the same
parameterized formula changes from $3.88\cdot10^{-36}$ to
$2.1825\cdot10^{-34}$.  GKL package this dependence as
$9.7p^2\cdot10^{-17}$ and report $2.18\cdot10^{-34}$
\cite[Lemma~4.5 and Corollary~4.6]{gkl24}.  Comparing the displayed theorem
guarantees $1.00\cdot10^{-36}$ and $2.18\cdot10^{-34}$ gives a factor of
$218>10^2$ \cite[Theorem~1.1]{kko21}.  This ratio is larger than $56.25$
because the KKO theorem states its final constant conservatively: against the
sharper pre-GKL bound $4.11\cdot10^{-36}$ quoted by GKL, the ratio is about
$53$ \cite[Sections~2 and~4.2]{gkl24}.  In short, KKO supply the combinatorial
mechanism that turns local happy events into a cheap $O$-join, whereas GKL
supply a stronger real-stable-polynomial method for proving that those same
events occur often enough.

\subsection{Summary of our approach}
\label{sec:tech_ours}

{\bf Steps 1 and 2 retain the structural framework of
Section~\ref{sec:tech_kko21}.}
We use the same subtour-LP reduction, max-entropy spanning-tree algorithm,
hierarchy and polygon representation of near-minimum cuts, and conversion
of a feasible slack vector into a cheap tour. The new work strengthens the
quantitative slack certificate under the exact tree law. The finite-precision
comparison is applied only after all layers have been combined, to the
actual tree-plus-minimum-join cost.

{\bf Section~\ref{sec:proof_improved_theorem} uses thirty-four parameter choices.}
The probability estimates, localized matching, and conditional parity bound
feed the same ancestor and payment argument for each row. The separated repairs and
aggregate-credit lemma then give a centered slack vector, and one finite-layer
lemma converts its threshold-dependent guarantees into all-cut feasibility.
The hierarchy domain and all scalar conditions are checked for each complete tuple.

{\bf Step 3 strengthens the local probability and parity estimates.}
For a fixed row, a half bundle is good at conditional probability
$\gamma_{\rm good}h$, with that row's parameters.
This class is not assumed to be contained in the source good-half class.
An electrical-energy argument bounds the smaller remote-vertex mean
increase by $0.22$ throughout $h\leq0.0055$. The competing-half
calculation then gives probability greater than $0.01330$ for one of
two incident half bundles, so bad half bundles form a matching.
Atom-dependent Newton inequalities and a normalized concave rank-balance expression show
that upward mass at least $1/2+k_{\rm good}h$ forces goodness
(Lemmas~\ref{lem:final23_electrical_remote} and~\ref{lem:candidate_good_threshold}).

For the window event, we condition on the bundle being selected before
estimating its residual lower tail. A pointwise inequality in the binary
covering coupling gives
\[
L_*:=\frac{(\gamma_{\rm good}-2)h-12d_0}{3/2-h-3d_0}>0.
\]
Factoring the parallel-edge choice and applying full-rank monotonicity
to deletion preserve this total-rank lower tail. The distinct window-premise
subcount still retains its additive deletion allowance. The conditional
rank-three atom first exceeds $0.248$. Four mean updates and an
atom-dependent Newton coefficient strengthen this estimate separately
on three rank ranges. A root bound gives central mass $m_*$, and a
mean-sensitive quadratic split proves the small and middle window ranges.
The large range retains the quadratic tail estimate. The direct mixed
event uses the same cutoff-dependent products and keeps both its actual
conditioning probability and independent endpoint-selection factor.
The low-mean CDF domain extends through $1.512$.

For the non-half events, full projected-rank conditioning changes each
subunion mean by its rank-excess budget. In the small case, a sharper
excess-tail estimate and two concavity arguments reduce small rank-three
probability to two outer endpoint checks; the complementary range uses up/down tails.
In the large case, proper-cut support and Newton bounds control the
rank-three means, and a homogeneous cubic split gives the stronger event
probability. All dimension reductions and shared interval endpoints are
explicit.
The paired proof keeps crossing-dependent lower and upper excess moments.
Nested-event mean comparisons and an asymmetric binary separator sharpen both branches.
For upper subset means, deleting the residual block before the final binary
selection gives the same final law with stronger bounds. Both complete
subset tables are retained on a closed crossing-mass cover, and a vertex
calculation retains the additive relations between subset means.
Proper-cut exclusions give coefficient factors
$e^{-2}/2$ and $e^{-1}/2$ in the low and high branches, respectively.
The two opposite selected-bundle events are disjoint, so their probabilities
are added within the same separator branch.

For polygons, a small-excess joint-atom estimate and a complete max-flow cut
proof preserve both separate marginal budgets. Each row's top and polygon
rates are $p_j$ and $P_j=s_jp_j$. Independent thinning preserves the conditional
laws; paired reductions share one coin. The crossing cutoff $\omega_j$ is
separate from the half-bandwidth $h_j$. Event classification is repeated for
each complete parameter tuple in Definition~\ref{def:final24_parameters}.

At a bad endpoint, both bundle states contribute to the same badness budget.
An outward-rounded bootstrap and closed-interval certificate bound its covering
mean difference by the row's $\Delta_{*,j}$. For a polygon child, its boundary
count reduces outside an explicit exceptional event to a Bernoulli sum of
degree at most two, independent of the internal child-tree count. The overlap
thresholds $0.0025$ and $0.25$ give stronger bottom savings. Neither argument
divides by the rare polygon-event probability.

The ancestor proof retains actual bottom savings, including terminal root-bad
edges and all large-, small-, and fractional-mass cases. Only a fractional
atom incident to an internal bad bundle receives the matching discount.
The top and bottom credits are $a_j,a_{{\rm bot},j}$. The two reduction
rates remain separate. Exact cut-excess accounting and fractional endpoint
allocation give one-sided repair coefficient $12$; the two-sided coefficient
is $5$. The one-sided repair remains supported on bottom edges.

{\bf Step 4 mixes threshold-dependent cutwise credit.}
Fix a hierarchy at threshold $u$ before choosing a parameter row.
In row $j$, put $W_j=s_j^{\rm pay}+r_j^{(1)}$ and retain the edgewise
credit $m_j(u)=\min\{a_j,a_{{\rm bot},j}-R_1(u)u\}>0$.
At a child incident to a bad bundle, the horizontal credit is the sum
of a baseline term $a_j^{(4)}I$ and a nonnegative term proportional to
the upward bad mass; neither term is discarded. Together with the
upward good-edge credit, this gives the degree-parent bound $A_j(u)$
in Lemma~\ref{lem:final24_credit}. The remaining one-sided cuts retain
bottom credit $B_j(u)=(1-d_0)(a_{{\rm bot},j}-R_1(u)u)$.

Mix the row witnesses with weights $w_j$. Their good sets may differ,
but the two cut classes agree because the hierarchy is fixed. Hence
the guaranteed cut credit is
\[
M_w:=\min\{\sum_jw_jA_j(u),\sum_jw_jB_j(u)\},\quad
\bar m:=\sum_jw_jm_j(u).
\]
The minimum is taken after mixing. Define
\[
C_w:=\frac{M_w}{(1-\bar m)(2+u)+M_w},\quad
\lambda:=\frac{2+u}{(1-\bar m)(2+u)+M_w}.
\]
For $W=\sum_jw_jW_j$, the centered vector
$\lambda(W-\E[W])-C_w\beta x$ preserves the coordinate and one-sided
cut bounds: $\lambda(1-\bar m)+C_w=1$ and
$\lambda M_w=C_w(2+u)$. Adding one two-sided repair leaves saving
$C_w-R_2(u)u$.

Use $10^6$ equally spaced thresholds up to $H=138217/10^{14}$, all on the
same sampled tree. An explicit thirty-three-interval schedule uses two rows
whose mixed cut credits agree. Coordinate losses
telescope, and cut excess pays for layers below that cut's threshold.
The rational ledger and finite-layer argument give
\[
B_*>2.78621079751\cdot10^{-18}.
\]
The definition and finite-sum proof are in Lemma~\ref{lem:final24_layered}.
A one-time implementation loss $10^{-28}$ still leaves
$\varepsilon>2.78621\cdot10^{-18}$. The centered expectations and auxiliary
layers are analysis witnesses, not computations required of the algorithm.

%% file: 10_preli.tex
\section{Preliminaries}
\label{sec:preliminaries}

Section~\ref{sec:prelim_notation} introduces the notation for the subtour-LP
solution, cuts, the KKO hierarchy, and the exact and approximate max-entropy
distributions. Section~\ref{sec:prelim_setup} records the standard reduction
to the spanning-tree polytope and the finite-precision max-entropy guarantee.
Finally, Section~\ref{sec:prelim_slack_to_tours} bounds the tour-cost loss
when passing to the approximate tree law.

\subsection{Notation}
\label{sec:prelim_notation}

Throughout the paper, references to numbered statements in \cite{kko21},
\cite{kko22}, and \cite{gkl24} use, respectively, the full versions
arXiv:2007.01409v6, arXiv:2105.10043v3, and arXiv:2311.09072v2.

Let $x^0$ denote a feasible subtour-LP solution after adjoining the zero-cost edge $e_0=\{u_0,v_0\}$, and set $E:=\operatorname{supp}(x^0)\setminus\{e_0\}$.
Let $x$ be the restriction of $x^0$ to $E$, let $G=(V,E)$, and write $n:=|V|$. We use $\operatorname{OPT}_{\mathrm{LP}}$ for the optimum subtour-LP value.

For a vector $z\in\mathbb R^E$ and an edge set $F\subseteq E$, write $z(F):=\sum_{e\in F}z_e$ and $\operatorname{supp}(z):=\{e\in E:z_e>0\}$. For a nonnegative edge-cost vector $c$, write $c(z):=\sum_{e\in E}c_ez_e$.

For a vertex set $S\subseteq V$, let $E(S):=\{\{u,v\}\in E:u,v\in S\}$ and $\delta(S):=\{\{u,v\}\in E:|\{u,v\}\cap S|=1\}$.
The notation $e_0\in\delta(S)$ is shorthand for $|\{u_0,v_0\}\cap S|=1$ in the augmented graph; sums over $\delta(S)$ remain over $E$ unless the $e_0$ coordinate is explicitly included.
If $T\subseteq E$ is a spanning tree and $F\subseteq E$, write $F_T:=|F\cap T|$. In particular, $\delta(S)_T$ is the number of tree edges crossing the cut $S$.

When invoking the KKO hierarchy and polygon classification, we identify $u_0$ and $v_0$. Thus every classified cut satisfies $e_0\notin\delta(S)$; cuts with $e_0\in\delta(S)$ are handled separately by setting $y_{e_0}=1$. A cut $S$ with $e_0\notin\delta(S)$ is \emph{$\eta$-near-minimum} if $x(\delta(S))<2+\eta$.

We use the hierarchy and polygon terminology of~\cite[Appendix~B]{kko22}.
The construction in the proof of \cite[Theorem~B.3]{kko22}, together with
\cite[Theorem~6.2 and Fact~B.4]{kko22}, produces for every $x$ and
$\eta\leq10^{-12}$ a valid hierarchy whose root is
$V\setminus\{u_0,v_0\}$, whose one-sided and two-sided classes exhaust the
relevant $\eta$-near-minimum cuts, and whose added hierarchy cuts are
$7\eta$-near-minimum. In the payment estimates we use the conservative
uniform envelope $d:=14\eta$, consistent with
\cite[Fact~4.34]{kko21}; the one-sided repair estimate retains the sharper
$7\eta$ bound. For a non-root hierarchy cut $S$, write $\mathsf p(S)$ for
its parent, $\delta^\uparrow(S):=\delta(S)\cap\delta(\mathsf p(S))$, and
$\delta^\rightarrow(S):=\delta(S)\setminus\delta(\mathsf p(S))$.
The children of a hierarchy cut are called its atoms. For disjoint atoms $u,v$, the set $E(u,v)$ is the edge bundle with one endpoint in each. Hierarchy cuts are classified as degree cuts or near-cycle cuts. Following \cite[Definition~B.1]{kko22}, a \emph{triangle cut} is a hierarchy cut with exactly two children. Separately, the three-atom degree configuration in \cite[Case~1 in the proof of Lemma~7.6]{kko21} will always be called a \emph{three-atom degree configuration}. The KKO polygon classification partitions the relevant near-minimum cuts into those crossed on both sides and those crossed on at most one side. When applying the KKO payment theorem, we follow \cite[Definition~4.31]{kko21} and regard a triangle cut as a degenerate polygon cut with $C=\varnothing$.

Whenever a KKO hierarchy is used, write $d$ for the uniform error envelope just described. The one-sided repair theorem retains its sharper error $7\eta$.
For the final parameter choice, Lemma~\ref{lem:final24_hierarchy}
extends this hierarchy interface to $\eta\leq1.4\cdot10^{-9}$ directly
from the same construction and its unrounded geometric bounds.
As in \cite[Definition~5.13]{kko21}, edges incident to $u_0$ or $v_0$
are declared bad and carry zero payment; they are neither top nor bottom
edges.

We use $\mu$ for the exact max-entropy distribution on spanning trees with marginals $x$, and $\mu_\lambda$ for a computable $\lambda$-uniform approximation. Unless a subscript is displayed, probabilities and expectations are taken over the current random spanning tree and any auxiliary randomness in the slack vectors.

\subsection{Setup and max-entropy approximation}
\label{sec:prelim_setup}

We first record the standard reduction that places the subtour-LP solution in the spanning-tree polytope, together with the finite-precision max-entropy guarantee.

\begin{lemma}[Subtour-LP reduction and max-entropy approximation]
\label{lem:prelim_reduction_approximation}
An optimal extreme point $x^0$ may be written in the form above so that $x$ belongs to the spanning-tree polytope and
\[
c(x)=\operatorname{OPT}_{\mathrm{LP}}.
\]
Every positive coordinate of $x^0$ is at least $1/n!$. For every rational $\delta>0$, one can compute, in time polynomial in $n$ and $\log(1/\delta)$, a $\lambda$-uniform spanning-tree distribution $\mu_\lambda$ satisfying
\[
\Pr_{T\sim\mu_\lambda}[e\in T]\leq(1+\delta)x_e
\quad\text{for every }e\in E.
\]
Moreover, there is an explicitly computable universal constant
$C_{\mathrm{stab}}>0$, independent of the instance, such that, for
$0<\delta\leq1$, the exact and approximate distributions satisfy
\[
q:=\|\mu-\mu_\lambda\|_1
\leq C_{\mathrm{stab}}n^3
\sqrt{\delta(1+\log(n/\delta))}.
\]
\end{lemma}

\begin{proof}
The reduction and the spanning-tree-polytope assertion are
from~\cite[Section~2.1 and Fact~2.2]{kko22}. The approximation algorithm is
\cite[Theorem~2.1]{kko22}, and the coordinate bound is recorded in
\cite[Section~6.2]{kko22}. Let $x^\lambda$ be the marginal vector of
$\mu_\lambda$, and define $E_+:=\{e:\lambda_e>0\}$. Since
$x^\lambda_e=0$ for $e\notin E_+$, every spanning-tree law with marginals
$x^\lambda$ is supported on trees contained in $E_+$. On this support,
\[
\log\mu_\lambda(T)=\sum_{e\in T}\log\lambda_e-\log Z_\lambda.
\]
Thus, for every spanning-tree law $\nu$ with marginals $x^\lambda$,
\[
0\leq D(\nu\Vert\mu_\lambda)=H(\mu_\lambda)-H(\nu),
\]
because $\nu$ and $\mu_\lambda$ have the same expectation of the preceding
display. Hence $\mu_\lambda$ is the maximum-entropy law with marginals
$x^\lambda$, also when some weights vanish. Since
$x^\lambda_e\leq(1+\delta)x_e$ and both marginal
vectors have coordinate sum $n-1$,
\[
\|x-x^\lambda\|_1\leq2(n-1)\delta.
\]
Apply the quantitative form \cite[Theorem~10]{straszak2019maximum} to the
spanning-tree polytope. Its ambient dimension is at most $n^2$, its unary
facet complexity is one, the diameter of its vertices is smaller than
$\sqrt{2n}$, and the number of spanning trees is at most $n^{n-2}$.
Substitution in that theorem gives
\[
\|\mu-\mu_\lambda\|_1
\leq C_{\mathrm{stab}}n^3
\sqrt{\delta(1+\log(n/\delta))}
\]
after increasing one universal constant. The proof of the cited theorem is
effective, so we fix any computable value supplied by that proof; no
optimization of this value is needed here.
\end{proof}
\subsection{Change of tree law}
\label{sec:prelim_slack_to_tours}

We record the output-cost comparison used after the slack construction.

\begin{lemma}[Change of tree law]
\label{lem:prelim_slack_to_tour}
Let $\mathsf J(T)$ denote the cost of a minimum join for the odd vertices of
$T$. For the law $\mu_\lambda$ in
Lemma~\ref{lem:prelim_reduction_approximation},
\[
\E_{\mu_\lambda}[c(T)+\mathsf J(T)]
\leq
\E_\mu[c(T)+\mathsf J(T)]
+(\delta+\frac q2)c(x).
\]
\end{lemma}

\begin{proof}
For every spanning tree $T$, the
vector that equals $x/2$ on $E$ and one on the zero-cost edge $e_0$ is a
feasible fractional join for the odd vertices of $T$. Therefore
\[
0\leq\mathsf J(T)\leq c(x)/2.
\]
Moreover,
\[
\E_{\mu_\lambda}[c(T)]=c(x^\lambda)\leq(1+\delta)c(x),
\quad \E_\mu[c(T)]=c(x).
\]
Using $q=\|\mu-\mu_\lambda\|_1$ gives
\[
\E_{\mu_\lambda}[\mathsf J(T)]-\E_\mu[\mathsf J(T)]
\leq\sum_T|\mu_\lambda(T)-\mu(T)|\,\mathsf J(T)
\leq\frac q2c(x).
\]
Adding the tree- and join-cost comparisons proves the claim.
\end{proof}

%% file: 40_proof.tex
\section{Proof of upper bound}
\label{sec:proof_improved_theorem}

Section~\ref{sec:proof_main_theorem} derives the main theorem from the
layered slack vector, and Section~\ref{sec:probability_primitives}
records the parameter choices.
Sections~\ref{sec:atom_tail_rank_estimates}--\ref{sec:tree_law_estimates}
establish the atom, tail, conditioning, and structured tree-law estimates.
Sections~\ref{sec:good_common_events}--\ref{sec:paired_matching}
certify bundle and polygon events and construct the localized matching.
Section~\ref{sec:reparam_payment_estimates} proves the conditional parity
bounds and synchronizes the reduction events.
Section~\ref{sec:ancestor_polygon_charges} controls ancestor credit and
polygon charges, and Section~\ref{sec:payment_theorem} assembles the
payment guarantee.
Section~\ref{sec:repair_combined_estimates} constructs the separated
repairs, and Section~\ref{sec:final_cutwise_credit} combines them using
aggregate expected credit.
Section~\ref{sec:layered_certificate} proves the finite-layer lemma and
evaluates the final certificate.
Section~\ref{sec:arithmetic_details} collects the exact arithmetic details.

\subsection{Proof of the main theorem}
\label{sec:proof_main_theorem}

The layered slack-vector construction supplies the improvement under the exact
max-entropy law. The finite-precision comparison then transfers this bound to
the polynomial-time implementation and yields the main result.

\begin{theorem}[Formal version of Theorem~\ref{thm:our_main_informal}]
\label{thm:our_main_formal}
There exists an explicitly computable constant
$\varepsilon_\star>2.78621\cdot10^{-18}$
such that, for every fixed $0<\varepsilon<\varepsilon_\star$ and every metric TSP instance, there is a randomized polynomial-time implementation of the max-entropy algorithm whose expected tour cost is at most
$(\frac32-\varepsilon)\operatorname{OPT}_{\mathrm{LP}}$.
Consequently, the integrality gap of the subtour-elimination LP is at most $3/2-\varepsilon_\star$.
\end{theorem}

\begin{proof}
Solve the subtour LP and choose an optimal extreme point $x^0$. Apply Lemma~\ref{lem:prelim_reduction_approximation}, and let $x$ be the resulting vector in the spanning-tree polytope. Then
\[
c(x)=\operatorname{OPT}_{\mathrm{LP}}.
\]

Apply Lemma~\ref{lem:final24_layered} and define
$\varepsilon_\star:=B_*>2.78621\cdot10^{-18}$, where $N=10^6$.
The construction uses thirty-four certified parameter choices and couples all
finitely many analysis layers to the same sampled tree. Each layer has a uniform edgewise
certificate; the rows mixed within it share one geometric hierarchy,
but their good-edge sets may differ. Different layers need not use
the same hierarchy.
For the sampled tree
$T$, define $y_e:=x_e/2+Z_e$ for $e\in E$ and $y_{e_0}:=1$.
Since $\beta(H)<1/2$, every coordinate of $y$ is nonnegative.  A cut
crossed by $e_0$ has $y$-value at least one, and every other cut containing
an odd number of tree edges has $y$-value at least one by the same lemma.
Thus $y$ is a feasible fractional join for the odd vertices of $T$.
The minimum join is no more expensive than $y$, so under the exact
max-entropy law the expected tour cost is at most
\[
c(x)+\frac{c(x)}2+\E[c(Z)]
\leq(\frac32-\varepsilon_\star)c(x).
\]

It remains to control the computable approximation. Fix $0<\varepsilon<\varepsilon_\star$ and define
\[
\rho:=\frac{\varepsilon_\star-\varepsilon}{2}.
\]
Define
\[
B:=\min\{
\frac1{n^2},\frac\rho2,
(\frac{\rho}{\sqrt2C_{\mathrm{stab}}n^3})^4
\}
\]
and choose the largest number of the form $2^{-m}$, $m\in\mathbb N$, satisfying
$\delta\leq B$.
Then $B/2<\delta\leq B$. Since $\delta\leq n^{-2}$ and
$1+\log n\leq n$,
\[
1+\log(n/\delta)
\leq n+\log(1/\delta)
\leq2\delta^{-1/2}.
\]
Lemma~\ref{lem:prelim_reduction_approximation} therefore gives
\[
q\leq\sqrt2C_{\mathrm{stab}}n^3\delta^{1/4}\leq\rho,
\quad
q/2+\delta\leq\rho.
\]
The running time is polynomial because
\[
\log(1/\delta)\leq\log(2/B)
=O(\log n+\log(1/\rho)).
\]
The comparison in Lemma~\ref{lem:prelim_slack_to_tour} is applied once to
the output-cost function $T\mapsto c(T)+\mathsf J(T)$; the layers incur no separate change-of-law loss. Combining that lemma with the
exact-law estimate above gives
\[
\E_{\mu_\lambda}[c(T)+\mathsf J(T)]
\leq(\frac32-\varepsilon_\star+\delta+\frac q2)c(x)
\leq(\frac32-\varepsilon_\star+\rho)c(x)
<
(\frac32-\varepsilon)\operatorname{OPT}_{\mathrm{LP}}.
\]

The integrality-gap conclusion follows because the expectation is an average over tours. Taking $\varepsilon\uparrow\varepsilon_\star$ proves the final claim.
\end{proof}

%% file: 60_improve.tex
\subsection{Parameters}
\label{sec:probability_primitives}

The construction balances local probability gains against repair costs using
several parameter choices. We collect these choices here so that every later
estimate can be checked against the same data.

\begin{definition}[Parameters for the final bound]
\label{def:final24_parameters}
Set $d_0=2\cdot10^{-8}$,
$\epsilon_F=0.1$, $t=1$, $H=138217/10^{14}$, and $N=10^6$.
At each threshold $0<u\leq H$, fix the geometric hierarchy before
choosing any of the thirty-four complete parameter tuples below. In all local
estimates fix one row and suppress its index. The final construction
mixes the resulting payment vectors on this common hierarchy.
\begin{center}
\begin{tabular}{c r r r}
\toprule
$j$&$h_j$&$r_j$&$\omega_j$\\
\midrule
1&0.00474337168262&0.00202678149813&0.01318549812643\\
2&0.00467675511181&0.00177671136598&0.01363267584577\\
3&0.00471839213351&0.00194637284476&0.01333199308802\\
4&0.00467695582122&0.00167963135242&0.01379715841829\\
5&0.00476419702921&0.00209078256027&0.01306816751213\\
6&0.00478554874980&0.00214507537291&0.01296586991455\\
7&0.00469744222034&0.00186122615296&0.01348297583738\\
8&0.00477485672223&0.00211791900739&0.01301704112751\\
9&0.00468704811896&0.00181885335662&0.01355805311931\\
10&0.00479624453261&0.00217224497700&0.01291466357717\\
11&0.00478093276463&0.00213336218116&0.01298794337151\\
12&0.00475369051062&0.00205853505742&0.01312729181588\\
13&0.00472189722769&0.00196060306704&0.01330674550642\\
14&0.00469246359492&0.00184095832576&0.01351889261962\\
15&0.00467685541771&0.00172816616502&0.01371493768943\\
16&0.00479007494455&0.00215658452102&0.01294418104241\\
17&0.00474803683734&0.00204115404231&0.01315915568507\\
18&0.00473284373636&0.00199437365229&0.01324488540165\\
19&0.00467690654287&0.00170372751009&0.01375634072287\\
20&0.00477774383791&0.00212525639700&0.01300321614393\\
21&0.00474803674645&0.00204115375515&0.01315915619965\\
22&0.00473284630570&0.00199438156806&0.01324487091759\\
23&0.00469546411802&0.00185319478263&0.01349721042547\\
24&0.00467690653116&0.00170373303383&0.01375633135951\\
25&0.00477593005096&0.00212064513522&0.01301190413386\\
26&0.00475087683756&0.00204988662224&0.01314314732933\\
27&0.00472735768247&0.00197745550883&0.01327587851396\\
28&0.00469395131816&0.00184702620783&0.01350814109403\\
29&0.00467693271912&0.00169129246971&0.01377740530048\\
30&0.00479318352889&0.00216447610735&0.01292930715794\\
31&0.00475087633852&0.00204988508855&0.01314315014649\\
32&0.00472735784162&0.00197745599867&0.01327587760705\\
33&0.00469395061194&0.00184702332724&0.01350814619543\\
34&0.00479318337235&0.00216447570999&0.01292930790696\\
\bottomrule
\end{tabular}

\begin{tabular}{c r r r r}
\toprule
$j$&$\gamma_{\rm good,j}$&$k_{\rm good,j}$&$K_j$&$p_j$\\
\midrule
1&2.67354783909477&5.35418823079187&4.96008779244273&0.00001003268213\\
2&2.71408821368829&5.43538975894109&5.07030735453861&0.00001088046515\\
3&2.68540726863074&5.37792621474281&4.99978369762098&0.00001028003047\\
4&2.71398974378901&5.43519265469476&5.08250650488441&0.00001102426051\\
5&2.66443186986062&5.33594598225196&4.92760756885928&0.00000984296697\\
6&2.65791872884318&5.32293146871799&4.89565657947595&0.00000970589056\\
7&2.69977465171858&5.40671188815691&5.03508369990248&0.00001057963385\\
8&2.66117751560732&5.32944317680430&4.91161532968433&0.00000977442922\\
9&2.70695058757152&5.42108919157615&5.05273102911632&0.00001073010768\\
10&2.65467532679840&5.31644856301738&4.87976413806608&0.00000963760377\\
11&2.65932635906375&5.32574221806823&4.90253569452708&0.00000973547056\\
12&2.66902627330834&5.34513803439368&4.94395202649790&0.00000993845706\\
13&2.68301696864247&5.37313517401852&4.99391013733910&0.00001023037233\\
14&2.70320889816446&5.41359055417716&5.04352301654999&0.00001065154649\\
15&2.71404012047016&5.43529148470283&5.07640547300905&0.00001095235831\\
16&2.65654680620950&5.32018780771587&4.88892068041669&0.00000967698045\\
17&2.67150320791336&5.35009437762919&4.95278160889408&0.00000999002972\\
18&2.67818285132013&5.36346118548895&4.97662534451924&0.00001012934226\\
19&2.71401508196225&5.43524103236718&5.07947623224218&0.00001098856174\\
20&2.66029793293741&5.32768327848892&4.90729806989373&0.00000975591695\\
21&2.67150325018243&5.35009445612074&4.95278174863043&0.00000999003057\\
22&2.67818171890139&5.36345891749554&4.97662130437009&0.00001012931866\\
23&2.70113907251115&5.40944323344689&5.03843180808855&0.00001060815651\\
24&2.71401508996976&5.43524104365391&5.07947553635127&0.00001098855356\\
25&2.66085091227518&5.32878823967807&4.91000983148454&0.00000976755463\\
26&2.67025844606347&5.34760344099774&4.94834336049095&0.00000996410629\\
27&2.68060335247908&5.36830490594996&4.98527688086611&0.00001017990775\\
28&2.70218239257611&5.41153357075155&5.04099776124092&0.00001063002161\\
29&2.71400213313538&5.43521510762811&5.08103892950298&0.00001100698583\\
30&2.65560345763351&5.31830281031666&4.88430394459967&0.00000965713143\\
31&2.67025866429681&5.34760387842857&4.94834414148665&0.00000996411085\\
32&2.68060328295096&5.36830476542260&4.98527662674262&0.00001017990628\\
33&2.70218287999458&5.41153454692878&5.04099895874939&0.00001063003182\\
34&2.65560350511665&5.31830290519690&4.88430417693565&0.00000965713243\\
\bottomrule
\end{tabular}

\begin{tabular}{c r r r r}
\toprule
$j$&$s_j$&$\epsilon_{M,j}$&$s_{0,j}$&$\Delta_{*,j}$\\
\midrule
1&1.75695012458503&0.00333673995714&0.00207689444253&0.06197219683386\\
2&1.75771062070818&0.00347571662414&0.00179283742356&0.06199620810775\\
3&1.75743198589334&0.00337811461027&0.00196337167292&0.06190838025713\\
4&1.75787708346393&0.00349879158233&0.00169520561035&0.06199670760157\\
5&1.75650704739913&0.00330460186055&0.00212391896109&0.06204089288943\\
6&1.75588145196533&0.00328090997277&0.00221443543505&0.06217465972923\\
7&1.75757177175423&0.00342715930297&0.00187780153496&0.06195256910820\\
8&1.75619431297988&0.00329277525335&0.00218345205937&0.06210777941927\\
9&1.75764137076370&0.00345153172861&0.00183520707661&0.06197449432340\\
10&1.75556848699708&0.00326904852370&0.00223457989067&0.06224158615535\\
11&1.75601638832737&0.00328603546358&0.00219199931186&0.06214581977932\\
12&1.75673019112439&0.00332081496804&0.00210365970105&0.06200633339280\\
13&1.75740863150922&0.00336991743431&0.00197766980151&0.06190098665974\\
14&1.75760505778367&0.00343882836210&0.00185742858879&0.06196306994659\\
15&1.75779384507810&0.00348727187106&0.00174402106003&0.06199645819419\\
16&1.75574887078568&0.00327589292239&0.00223421522581&0.06220301792610\\
17&1.75685056947241&0.00332954019361&0.00209378827695&0.06198765883662\\
18&1.75717463128261&0.00335300109505&0.00201198991729&0.06193734887532\\
19&1.75783575393709&0.00349307675298&0.00171944228772&0.06199658229672\\
20&1.75610977474780&0.00328957422079&0.00219199931186&0.06212585195524\\
21&1.75685057146048&0.00332954033973&0.00209378827695&0.06198765853056\\
22&1.75717457643687&0.00335299713066&0.00201198991729&0.06193735738935\\
23&1.75758495934691&0.00343179226736&0.00186972870900&0.06195674052212\\
24&1.75783574446115&0.00349307544276&0.00171944784346&0.06199658227262\\
25&1.75616290306021&0.00329158685523&0.00219199931186&0.06211449230334\\
26&1.75679008718679&0.00332515712082&0.00209847179819&0.06199704234655\\
27&1.75729185259445&0.00336147786956&0.00199489911154&0.06191913674650\\
28&1.75759509097792&0.00343533967974&0.00186352811633&0.06195993160999\\
29&1.75785708136135&0.00349602729773&0.00170693502434&0.06199664498911\\
30&1.75565796989904&0.00327244443398&0.00223530670442&0.06222245346856\\
31&1.75679009780906&0.00332515789087&0.00209847179819&0.06199704069726\\
32&1.75729184919600&0.00336147762397&0.00199489911154&0.06191913727617\\
33&1.75759509570897&0.00343534133597&0.00186352522077&0.06195993309965\\
34&1.75565797447607&0.00327244460765&0.00223530670442&0.06222245248995\\
\bottomrule
\end{tabular}
\end{center}

Define $P=sp$ and retain
\begin{align*}
\epsilon_B&=\frac{2((k_{\rm good}+1)h+2d_0h)}{1+2k_{\rm good}h}+10^{-11},&
\vartheta&=\frac{(1-r)(1-2r-2d_0)}{1+d_0}-10^{-9},\\
\chi&=\vartheta r,&\sigma&=\chi+2\chi^2+d_0,\\
\rho_{\rm sm}&=1-\frac{\sigma+5d_0+2d_0\sigma+6d_0^2}{1-\sigma-2d_0},&
\zeta&=\vartheta\rho_{\rm sm}-2d_0/r-10^{-9},\\
a&=\zeta rp,&
a_{\rm bot}&=p(s-(1+d_0)(7/4+h/2+3d_0/2)-2d_0s).
\end{align*}
Also set
\[
T_{18}=\sum_{i=0}^{18}\frac{(2(1-\epsilon_M/2-3d_0))^i}{i!},
q_0=(1+1/T_{18})/2+10^{-12},\quad
\bar q=\epsilon_M/2,\quad q_{\rm unh}=q_0+4d_0/(1-\bar q).
\]
Each half-bundle class is defined directly by conditional
$2$--$2$ probability at least $\gamma_{\rm good,j}h_j$. All thirty-four $\gamma_{\rm good,j}$ are below three. No inclusion in a source good-half class is
used. The crossing cutoff is $\omega$, not $h$; it does not replace $h$ in the
half-bandwidth or the full-pair saving. All events are reclassified at the
row's own $p$. Put $w_0:=s_0(1-s_0)$.
All terminating decimals in the parameter tables are exact rationals.
The local error satisfies $0\leq d\leq d_0$, and the crossing cutoff
has the separate domain $h\leq\omega<0.02$. Its individual probability,
structural, and ancestor inequalities are checked for every row; no
uniform assertion over this larger box is used.

\end{definition}

All probability statements refer to the exact max-entropy tree law.
The current hierarchy has error $0\leq d\leq d_0$.
For an arbitrary payment amplitude $\beta>0$, put $\tau:=t\beta=\beta$.
Top events are thinned to probability $p$ and polygon events to $P$.
Whenever $x(\delta^\uparrow(S))=0$, set the degree-cut increase
$I_{e,S}:=0$. Matching conservation then forces the corresponding
$m_{e,S}=0$, so no quotient by zero occurs.

Scalar comparisons are verified over the rationals. For $x\geq0$ we use
\[
\sum_{j=0}^{n}\frac{x^j}{j!}
\leq e^x\leq
\sum_{j=0}^{n}\frac{x^j}{j!}
+\frac{x^{n+1}}{(n+1)!}\frac1{1-x/(n+2)},\quad x<n+2,
\]
with $n=40$, or the shorter Taylor certificates specified below.
The arithmetic rules and finite evaluations are detailed in
Section~\ref{sec:arithmetic_details}.

\subsection{Atom, tail, and rank estimates}
\label{sec:atom_tail_rank_estimates}

Our local event estimates repeatedly require lower bounds on a few small counts.
We begin with dimension-independent atom and tail inequalities for sums of
independent Bernoulli variables.

\begin{lemma}[Bernoulli atom and tail estimates]
\label{lem:bernoulli_estimates}
For a Poisson-binomial variable $Z$, the following bounds hold in every
finite dimension:
\[
\begin{array}{c|c}
\text{hypothesis}&\text{conclusion}\\ \hline
\E[Z]\leq1.01&\Pr[Z\leq1]>0.70\\
\E[Z]\leq1.5&\Pr[Z\leq1]\geq7/16\\
\E[Z]\leq1.504&\Pr[Z\leq1]>0.434\\
\E[Z]\leq1.512&\Pr[Z\leq1]>0.4284\\
\E[Z]\leq1.53&\Pr[Z\leq1]>0.4147\\
\E[Z]\leq2.01&\Pr[Z\leq2]>0.66\\
\E[Z]\leq2.5&\Pr[Z\leq2]>0.42\\
0.99\leq\E[Z]\leq1.512&\Pr[Z=1]>0.33\\
1\leq\E[Z]\leq1.2&\Pr[Z=1]>0.36\\
1.2\leq\E[Z]\leq2.5&\Pr[Z=2]>0.1637\\
1.49\leq\E[Z]\leq2.02&\Pr[Z=2]>1/4\\
1.47\leq\E[Z]\leq2.02&\Pr[Z=2]>c_3:=0.248.
\end{array}
\]
For $0<\delta\leq0.02$, one also has
\begin{align}
0.99\leq\E[Z]\leq2-\delta&\implies
 \Pr[Z=1]\geq\delta(1-\delta/2),\label{eq:near_two_atom}\\
\E[Z]\leq3-\delta&\implies
 \Pr[Z\leq2]\geq\delta(1-\delta/3).\label{eq:near_three_cdf}
\end{align}
Moreover, $\Pr[Z=1]\geq\mu e^{-\mu}$ for $0\leq\mu:=\E[Z]\leq1$.
For the same $\delta$, if $\E[Z]\geq1+\delta$, then
\[
\Pr[Z\geq2]\geq\delta e^{-\delta}\geq\delta c(\delta),
\quad c(x):=1-x+x^2/2-x^3/6.
\]
\end{lemma}

\begin{proof}
At fixed mean and dimension, choose a minimizer of the desired atom or
tail with the fewest fractional Bernoulli parameters. For two parameters
of fixed sum, the objective is affine in their product. A positive
coefficient forces an endpoint, a zero coefficient permits an endpoint,
and a negative coefficient forces equality. Thus all remaining fractional
parameters are equal: a minimizer is a forced integer plus a binomial,
as in \cite[Theorem~A.9]{gkl24}. We use the inequality
\[
\log(1-z)\geq-z-\frac{z^2}{2(1-z)},\quad 0\leq z<1.
\]

{\bf CDFs at one.}
Monotonicity permits the upper mean endpoint $M$. One forced success
gives $2-M$; an unforced binomial with $n\geq2$ gives
\[
F_n(M)=(1-M/n)^{n-1}(1+M-M/n).
\]
The logarithmic inequality gives
$(n-1)\log(1-M/n)\geq-M$ when $M\leq2n/(n+1)$.
The following finite checks and dimension-uniform lower bounds suffice:
\[
\begin{array}{c|c|c}
M&\text{dimensions checked exactly}&\text{bound in all larger dimensions}\\ \hline
1.01&2\leq n\leq20&e^{-M}(1+20M/21)>0.70\\
1.5&n=2,3&e^{-M}(1+3M/4)>7/16\\
1.504&n=2,3&e^{-M}(1+3M/4)>0.434\\
1.512&n=2,3&e^{-M}(1+3M/4)>0.4284\\
1.53&n=2,3&e^{-M}(1+3M/4)>0.4147.
\end{array}
\]
Dimensions at most one are immediate; equality $7/16$ occurs for
$M=1.5,n=2$.

{\bf CDFs at two.}
Two forced successes give $3-M$, and one forced success reduces to
the preceding CDF at mean $M-1$. In the unforced case, for
$M\leq4n/(n+2)$ the logarithmic inequality gives
\[
\Pr[\operatorname{Bin}(n,M/n)\leq2]
\geq e^{-M}(1+M(1-2/n)+M^2(1/2-3/(2n)+1/n^2)).
\]
The parenthesized expression increases with $n$ on the required ranges.
At $M=2.01$, check $3\leq n\leq99$ and use the displayed bound at
$n=100$ thereafter. At $M=2.5$, check $3\leq n\leq9$ and use the
bound at $n=10$ thereafter. These exceed $0.66$ and $0.42$,
respectively. Dimensions at most two have CDF one.

{\bf Atoms at one.}
A forced success gives $2-\mu$. Without one, the logarithmic bound gives
\[
\Pr[\operatorname{Bin}(n,\mu/n)=1]\geq\mu e^{-D(n,\mu)},
\quad
D(n,\mu):=\frac{\mu(n-1)(2n-\mu)}{2n(n-\mu)}.
\]
For $\mu\in[0.99,1.512]$, check $n=2,3$ at the mean endpoints.
For every $n\geq4$, $\mu<2n/(n+1)$ gives $D(n,\mu)\leq\mu$;
both $0.99e^{-0.99}$ and $1.512e^{-1.512}$ exceed $0.33$.
On $[1,1.2]$ this logarithmic condition holds already for $n\geq2$,
and $1.2e^{-1.2}>0.36$. The same argument gives
$\Pr[Z=1]\geq\mu e^{-\mu}$ on $[0,1]$.
The one-variable and integer endpoints follow directly.

{\bf Atoms at two.}
For $\mu\in[1.2,2.5]$, two forced successes leave a zero atom at
least $1/2$. One forced success leaves an atom at one with residual
mean in $[0.2,1.5]$. Check one and two fractional variables directly;
for $n\geq3$, the logarithmic bound gives
\[
\min\{0.2e^{-0.2},1.5e^{-1.5}\}>0.1637.
\]
Without a forced success, check $n=2,3,4$ at the admissible mean
endpoints. The atom is unimodal in its mean. For $n\geq5$,
$\mu<4n/(n+2)$ gives
\[
\Pr[Z=2]\geq0.4\mu^2e^{-\mu}>0.1637.
\]

For the sharper target-two bounds below mean two, the zero- and
one-forced-success branches of \cite[Lemma~2.21]{kko21} are at least
\[
1.49^2e^{-1.49}/2>1/4,\quad 0.49e^{-0.49}>1/4.
\]
Above mean two, the correction exponent is at most $0.02$.
For zero or one forced success its base is at least $1/4$.
Since $0.97^{50}<1/4$, the correction exceeds $0.97$.
The two branches are consequently at least
\[
0.97(2.02)^2e^{-2.02}/2>1/4,\quad
0.97(1.02)e^{-1.02}>1/4.
\]
The Poisson factors decrease above their target. With two forced
successes the bound is $0.98e^{-0.02}>1/4$.
On $[1.47,2.02]$, only the lower endpoints change:
$1.47^2e^{-1.47}/2>0.248$ and $0.47e^{-0.47}>0.248$.

{\bf Near-integer bounds.}
For Eq.~\eqref{eq:near_two_atom}, a forced success gives at least
$\delta$. With two unforced variables, concavity of
$\mu(1-\mu/2)$ gives $\delta-\delta^2/2$ at the upper endpoint;
the other endpoint is larger. For $n\geq3$,
\[
D(n,\mu)\leq D(n,2)=2+\frac2{n(n-2)}\leq8/3,
\]
so $0.99e^{-8/3}>0.02$ suffices. One variable is immediate.

For Eq.~\eqref{eq:near_three_cdf}, monotonicity permits mean
$3-\delta$ whenever attainable. Two forced successes give at least
$\delta$. One forced success leaves a CDF at one of mean $2-\delta$.
With two fractional variables this is $\delta-\delta^2/4$; with
$n\geq3$, its atom at one is at least $1.98e^{-8/3}>0.02$.
With no forced success, $n=3$ gives
$\delta-\delta^2/3+\delta^3/27$. Check $n=4,5$ at mean three.
For $n\geq6$, the logarithmic condition $3\leq4n/(n+2)$ bounds
the atom at two below by $(5/12)(2.98)^2e^{-3}>0.02$.
Smaller dimensions only improve the CDF.

Finally, decreasing Bernoulli parameters to total mean $1+\delta$
can only decrease the event of at least two successes.
At that mean, \cite[Lemma~2.21]{kko21} has zero correction exponent
for target two. Its one-forced branch is $\delta e^{-\delta}$,
and its zero-forced branch is at least $e^{-1.2}/2>\delta$.
The cubic Taylor lower bound proves the last assertion.
All finite comparisons use the rational exponential enclosures in
Section~\ref{sec:arithmetic_details}; the logarithmic arguments cover
every larger dimension.
\end{proof}

When a conditional count has degree at most two, bounds on its mean or its two
tails give explicit lower bounds on its middle atom.

\begin{lemma}[Quadratic mean and tail bounds]
\label{lem:quadratic_bounds}
Let $J\in\{0,1,2\}$ have a real-rooted probability generating polynomial,
and put $m=\Pr[J=1]$ and $\psi(v)=v(1-v/2)$.
If $v\leq\E[J]\leq2-v$ for $0\leq v\leq1$, then $m\geq\psi(v)$.

If both tails at one have probability at least $\alpha\in(0,1)$, then
\[
m\geq\phi(\alpha):=2\sqrt{1-\alpha}-2(1-\alpha)
>\alpha(1-\alpha).
\]
In particular, for $0<b<\alpha$, the rational comparison
\[
4(\alpha-b)(1-\alpha)-b^2>0
\]
implies $m>b$.
\end{lemma}

\begin{proof}
Factor the count law into two Bernoulli factors, padding by a zero
parameter if necessary. If their sum is $\mu$, their product is at
most $\mu^2/4$, so $m\geq\mu-\mu^2/2\geq\psi(v)$ on
$[v,2-v]$.

For the tail assertion, write the other probabilities as $p_0,p_2$.
If $m<\alpha$, then $p_0,p_2\geq\alpha-m$ and
$p_0+p_2=1-m$. Hence real-rootedness gives
\[
m^2\geq4p_0p_2\geq4(\alpha-m)(1-\alpha).
\]
Solving this quadratic gives $m\geq\phi(\alpha)$.
If $m\geq\alpha$, the same conclusion follows from
$\phi(\alpha)\leq\alpha$.
For $t=\sqrt{1-\alpha}\in(0,1)$,
$\phi(\alpha)-\alpha(1-\alpha)=t(1-t)^2(t+2)>0$.
Finally $4(\alpha-m)(1-\alpha)-m^2$ decreases with $m\geq0$,
which proves the rational test.
The factorization describes the count distribution, not independence
of the actual sampled edges.
\end{proof}

To control a joint event, we also need to compare how probability is distributed
between two coordinate blocks at consecutive total ranks. The following stability
inequality bounds the mixed rank-two coefficient from their truncated means.

\begin{lemma}[Adjacent-rank balance]
\label{lem:candidate_adjacent_rank_balance}
Let a bivariate stable probability generating polynomial have
degree-one and degree-two parts
\[
a x+b y,\quad c x^2+mxy+e y^2,
\]
where $a,b>0$. Put $P:=a+b$ and $q:=c+m+e>0$.
If both degree-at-most-two contributions to its coordinate means are at
least $L$, and $W:=P+2q-2L$, then
\[
m\geq\frac q2-\frac{W^2}{8(P+q)}.
\]
\end{lemma}

\begin{proof}
The adjacent-rank homogenization in \cite[Lemma~4.16]{bbl09}, with a
positive rescaling of the homogenizing variable, shows that
\[
Q(x,y,z):=(ax+by)z+cx^2+mxy+ey^2
\]
is stable. In particular,
\[
mab\geq cb^2+ea^2.
\]
Indeed, if $N:=cb^2-mab+ea^2>0$, take $x=b+i\tau$ and $y=-a+i\tau$.
Solving $Q=0$ for $z$ gives
\[
\operatorname{Im}z=\frac{N-\tau^2q}{\tau(a+b)}>0
\]
for sufficiently small $\tau>0$, contradicting stability.

Set $f:=a-b$, $g:=c-e$, and $s:=f+2g$.
The two truncated means have sum $P+2q$ and difference $s$, so
$|s|\leq W$. Rewriting the preceding coefficient inequality gives
\[
m\geq\frac q2+\frac{(P+q)f^2}{2P^2}-\frac{fs}{2P}.
\]
Completing the square in $f$ bounds the last two terms below by
$-s^2/(8(P+q))\geq-W^2/(8(P+q))$.
\end{proof}

The preceding comparison uses only low ranks, so we must account for the
probability and mean carried by higher ranks. Newton's inequalities give the
required tail bounds from the first two atoms.

\begin{lemma}[Newton tail and excess bounds]
\label{lem:newton_tails}
Let $(p_j)$ be a Poisson-binomial rank distribution with
$p_1\geq c>0$ and $p_2=q<2c$. Put $z=q/(2c)$ and
\[
T_c(q):=\frac{2q^2}{3c}\frac{3-2z}{(1-z)^2},
\quad
T_c^{\rm ex}(q):=\frac{2q^2}{3c}(1-z)^{-2}.
\]
Then
\[
\sum_{j\geq3}jp_j\leq T_c(q),
\quad
\sum_{j\geq3}(j-2)p_j\leq T_c^{\rm ex}(q).
\]
\end{lemma}

\begin{proof}
Newton's inequalities, after weakening their dimension factors, give
\[
p_3\leq\frac{2q^2}{3c},
\quad
\frac{p_4}{p_3}\leq\frac34\frac{p_3}{p_2}\leq z.
\]
Log-concavity therefore gives
$p_{3+i}\leq(2q^2/(3c))z^i$ for every $i\geq0$.
Sum $(3+i)z^i$ and $(1+i)z^i$ to obtain the two bounds.
If $q=0$ or $p_3=0$, gapless rank support makes the tail zero,
so no division by zero is needed.
\end{proof}

Combining adjacent-rank balance with the tail bounds gives a joint-event
certificate that is concave in the rank-two probability. This lets us verify the
certificate on an interval by checking its endpoints.

\begin{lemma}[Concave adjacent-rank certificate]
\label{lem:concave_rank_certificate}
Let $X,Y$ have a stable joint probability generating polynomial, and
put $Z:=X+Y$. Suppose
\[
\E[X],\E[Y]\geq L_0,\quad \E[Z]\geq1,\quad
\Pr[Z=1]\geq c>0,\quad \Pr[Z=2]=q<2c.
\]
Define
\[
\Phi(q):=1-q-T_c(q),\quad W(q):=q+1-2L_0+2T_c(q).
\]
If $\Phi(q)>0$, $W(q)\geq0$, and $L_0-2q-T_c(q)>0$, then
\[
\Pr[X=Y=1]\geq F(q):=q/2-W(q)^2/(8\Phi(q)).
\]
On every interval where these conditions hold, $F$ is concave.
\end{lemma}

\begin{proof}
Write $P:=\Pr[Z=1]$. Lemma~\ref{lem:newton_tails} bounds the
mean contribution from ranks at least three by $T_c(q)$. Both
degree-one coefficients are at least $L_0-2q-T_c(q)>0$, and
both truncated means are at least $L_0-T_c(q)$. Also
\[
P+q\geq1-q-T_c(q)=\Phi(q),\quad P\leq1-q.
\]
Lemma~\ref{lem:candidate_adjacent_rank_balance} therefore gives
the displayed bound. This applies after projection to the counted
union; no outside coordinates are retained in a non-extremal rank
condition. Compulsory monomial factors can be removed and the quotient
polarized, since these operations preserve stability.

The power series of $T_c$ has nonnegative coefficients, so
$T_c',T_c''\geq0$, $W''\geq0$, and $\Phi''\leq0$. Thus
\[
(W^2/\Phi)''=
\frac{2(W'\Phi-W\Phi')^2}{\Phi^3}
+\frac{2WW''}{\Phi}-\frac{W^2\Phi''}{\Phi^2}\geq0.
\]
Hence $F''\leq0$, and its minimum on a compact interval is attained
at an endpoint.
\end{proof}

A large central atom contains information that the usual Newton inequalities do
not exploit. We retain this information to strengthen the coefficient ratios used
in the later tail estimates.

\begin{lemma}[Atom-dependent Newton inequalities]
\label{lem:atom_newton}
Let $a_j=\Pr[Z=j]$ for a Poisson-binomial variable $Z$.
For $1/2\leq c\leq1$, put $\delta=(1-c)/c$.
For $(L,\nu)=(1/2,2/3)$ or $(1/4,1)$ and $0<x\leq1/3$, define
\[
q_L(x)=x-(1-L)x^2,\quad w_L(x)=Lx^2(1-\nu x),\quad
d_L(x)=\min_{S>0}\{1/S+q_L(x)S+w_L(x)S^2\}.
\]
Set $d_L(0)=0$, and let $X_L(\delta)\in[0,1/3)$ be the unique
solution of $d_L(X_L(\delta))=\delta$. Put $x_c=X_{1/2}(\delta)$,
$y_c=X_{1/4}(\delta)$, and
\[
J(c)=\frac{2(1-x_c/2)^2}{1-2x_c/3},\quad
K(c)=\frac{4(1-3y_c/4)^2}{1-y_c}.
\]

If $a_1\geq c$, then
$a_2^2\geq\max\{3/2,J(c)\}a_1a_3$.
If in addition $Z$ has degree at most three, then
$a_2^2\geq\max\{3,K(c)\}a_1a_3$.

For $0<c\leq1$, put $\delta=(1-c)/c$ and, for $r>0$, define
\[
Q_r(\delta)=\max_{t\geq0}\{\delta t-t^2-t^3/r\},\quad
F_\delta(r)=r(4-r)-9Q_r(\delta).
\]
If $\delta=0$, set $r_*(c)=4$. Otherwise set $r_*(c)=2$ when
$F_\delta(2)\leq0$, and let $r_*(c)$ be the unique zero of
$F_\delta$ in $(2,4)$ when $F_\delta(2)>0$. Define
\[
k_0(c)=\max\{2,4-3\delta^2/4\},\quad
k(c)=\max\{k_0(c),\lfloor10^{20}r_*(c)\rfloor/10^{20}\},\quad
D(c)=ck(c).
\]

If $a_2\geq c$, then $a_1^2\geq k(c)a_0a_2$.
The functions $k,D$ are increasing, $2\leq k\leq4$, and
$D(c')-D(c)\geq2(c'-c)$ for $c'\geq c$.
\end{lemma}

\begin{proof}
For $0<x\leq1/3$, the minimizing $S$ in $d_L(x)$ is unique and
satisfies $1=q_L(x)S^2+2w_L(x)S^3$. Since $q_L',w_L'>0$,
$d_L$ is strictly increasing. Testing $S=x^{-1/2}$ proves continuity
at zero, and
$d_L(1/3)>2\sqrt{q_L(1/3)}\geq1$. This proves the asserted
existence and uniqueness of $X_L$.

The fixed coefficients $3/2$ and $3$ follow from Newton's inequalities.
When $a_3=0$ the first assertions are immediate. When $a_0=0<a_1$,
exactly one success is compulsory; Newton applied after removing it gives
coefficients $2$ and $4$, respectively. In the remaining case $\delta>0$.
Let $t_i$ be the finite
Bernoulli odds, $S=\sum_i t_i$, and $e_j$ their elementary symmetric sums.
Writing $T=\max_i t_i$, $U=S-T$, and $x=U/S$, normalization gives
\[
\frac1S+\frac{e_2}S\leq\delta,\quad
U\leq S-\frac{\sum_i t_i^2}{S}=\frac{2e_2}{S},\quad
x\leq\frac{2\delta}S-\frac2{S^2}\leq\frac{\delta^2}2\leq\frac12.
\]
Let $V,W$ be the second and third elementary symmetric sums of the odds
remaining after $T$. Then $e_2=TU+V$, $e_3=TV+W$, $V\leq U^2/2$,
and $3W\leq UV$. Put $\lambda=V/U^2>0$ and $R=e_2^2/(Se_3)$.
Use $(L,\nu)=(1/2,2/3)$ in general and $(1/4,1)$ in degree
at most three; in the latter case $W=0$ and $V\leq U^2/4$. Thus
\[
0<\lambda\leq L,\quad
R\geq F(\lambda,x):=
\frac{(1-(1-\lambda)x)^2}{\lambda(1-\nu x)}.
\]
Write $X=X_L(\delta)$ and $r=F(L,X)=q_L(X)^2/w_L(X)$.
Assume $R<r$, since otherwise the desired bound holds.
Set $v=e_2/S^2=x-(1-\lambda)x^2$ and $P=e_2/S>0$.
Retaining the third atom in normalization yields
\[
\delta\geq\frac{1+e_2+e_3}{S}
=\frac vP+P+\frac{P^2}R
\geq\frac vP+P+\frac{P^2}r,
\quad v\leq Q_r(\delta)=:D.
\]
Let $S_*$ attain $d_L(X)$ and put $P_*=q_L(X)S_*$. The stationary
equation and $d_L(X)=\delta$ give
\[
q_L(X)=P_*^2+2P_*^3/r,\quad \delta=2P_*+3P_*^2/r.
\]
Hence $P_*$ maximizes $Q_r$, so $0<D=q_L(X)\leq\delta^2/4\leq1/4$.

For either chosen pair $(L,\nu)$, $F(\lambda,x)$ decreases on
$0\leq x\leq1/2$: its logarithmic derivative has the sign of
$\nu-2(1-\lambda)+\nu(1-\lambda)x$, which is negative there.
Let $z_\lambda\leq1/2$ solve
$z_\lambda-(1-\lambda)z_\lambda^2=D$. This solution exists, is at
least $x$, and decreases with $\lambda$. At equality,
\[
F(\lambda,z_\lambda)=\frac{D^2}{G(z_\lambda)},\quad
G(z)=(z^2-z+D)(1-\nu z).
\]
For $z\in[z_L,z_\lambda]$, we have $D\geq z-z^2$, so
\[
G'(z)=-3\nu z^2+2(1+\nu)z-1-\nu D
\leq(2z-1)(1-\nu z)\leq0.
\]
All denominators are positive, and $z_L=X$. Hence
$R\geq F(\lambda,z_\lambda)\geq F(L,X)=r$, a contradiction.
This gives $J(c)$ and $K(c)$ without asserting that $x_c$ or $y_c$
separately bounds the actual $x$. In degree three the coefficient is
sharp for $c<1$: the odds $((1-y_c)S_*,y_cS_*/2,y_cS_*/2)$ have
first atom $c$ and ratio $K(c)$.

For the second assertion, $a_0=0$ is harmless. Otherwise put
$r=S^2/e_2\geq2$. The case $r\geq4$ is immediate, so suppose $r<4$.
Cauchy--Schwarz gives $\sum_i t_i^3\geq(\sum_i t_i^2)^2/S$, whence
\[
e_3\geq\frac{(4-r)S^3}{3r^2},\quad
\delta\geq\frac{1-a_2}{a_2}
\geq\frac rS+\frac{(4-r)S}{3r}\geq2\sqrt{(4-r)/3}.
\]
This proves $r\geq k_0(c)$. We also have
\[
S^3+9e_3-4Se_2\geq0.
\]
Here is a direct proof in arbitrary dimension. Merging the two smallest
positive odds $u,v$ changes the left side by
$uv(9(u+v)-5S)\leq-Suv/2$ when at least four odds remain.
At three odds $x\geq y\geq z\geq0$, it equals
$(x-y)^2(x+y-z)+z(x-z)(y-z)\geq0$.
Iterating the merge proves the inequality.
Retain the constant term in normalization and set $t=r/S$. Then
\[
\delta\geq r/S^2+r/S+(4-r)S/9,\quad
r(4-r)\leq9(\delta t-t^2-t^3/r)\leq9Q_r(\delta).
\]
For $\delta>0$, the maximizer $t_r$ of $Q_r$ is positive and unique,
with $\delta=2t_r+3t_r^2/r$. Therefore
\[
\frac{\partial F_\delta(r)}{\partial r}
=4-2r-9t_r^3/r^2<0\quad(2\leq r\leq4).
\]
Also $F_\delta(4)<0$. Thus the defining root exists when needed,
and $F_\delta(r)\leq0$ implies $r\geq r_*(c)$.
Since $Q_r$ increases with $\delta$, $r_*(c)$ increases with $c$.
The endpoint $\delta=0$ follows by continuity, or from $a_2=1$.
Downward rounding preserves the bound and monotonicity. Retaining
$k_0$ is necessary, since it can exceed $r_*$ for small $c$.
Finally $k\geq2$ gives the asserted increment bound for $D$.
\end{proof}

We propagate the improved coefficient ratios to the higher ranks while retaining a
relation between the omitted probability and omitted mean. This relation
strengthens the normalized adjacent-rank certificate.

\begin{lemma}[Factorial tails and normalized adjacent ranks]
\label{lem:normalized_rank_tail}
Suppose $a_1\geq c>0$, $a_2=q$, and
$a_2^2\geq k a_1a_3$. Put $b=3q/(kc)<4$, $z=b/6$, and
\begin{align*}
S(b)&=1+b/4+\frac{b^2}{20(1-z)},\\
U(b)&=3+b+\frac{b^2(5-4z)}{20(1-z)^2},&
V(b)&=1+b/2+\frac{b^2(3-2z)}{20(1-z)^2},\\
\mathcal T(q,c,k)&=\frac{q^2}{kc}U(b),&
\mathcal E(q,c,k)&=\frac{q^2}{kc}V(b).
\end{align*}
Writing $M=\sum_{j\geq3}ja_j$ and $B=\sum_{j\geq3}a_j$, we have
\[
M\leq\mathcal T(q,c,k),\quad
\sum_{j\geq3}(j-2)a_j\leq\mathcal E(q,c,k),\quad
B\geq\frac{S(b)}{U(b)}M.
\]

More generally, let $X,Y$ have a stable joint generating polynomial,
$\Pr[X+Y=2]=q$, and $\E[X],\E[Y]\geq L_0$.
Suppose their omitted probability and mean satisfy $B\geq\lambda M$,
$M\leq T(q)$, where $0\leq\lambda\leq1/3$ is fixed. Define
\[
\Phi(q)=1-\lambda T(q),\quad
W(q)=1+q-2L_0+(2-\lambda)T(q),\quad
F(q)=q/2-W(q)^2/(8\Phi(q)).
\]
If $\Phi>0$, $W\geq0$, and $L_0-2q-T>0$, then
$\Pr[X=Y=1]\geq F(q)$.
If $T,T',T''\geq0$ on an interval and these conditions hold there,
then $F$ is concave on that interval.
\end{lemma}

\begin{proof}
Newton gives $(j+1)a_{j+1}/a_j\leq3a_3/q\leq b$ for $j\geq3$.
Thus $a_{3+i}\leq(q^2/(kc))w_i$, where
$w_0=1,w_1=b/4$, and $w_i=(b^2/20)z^{i-2}$ for $i\geq2$.
Summing $w_i,(3+i)w_i,(1+i)w_i$ gives $S,U,V$.
Moreover $a_{3+i}/w_i$ is decreasing: the first two ratio bounds are
$b/4,b/5$, and subsequent ratios are at most $b/6$.
Pairwise expansion of the covariance of an increasing and a decreasing
sequence shows that the mean of $3+i$ under the normalized $a_{3+i}$
is at most its mean under the normalized $w_i$. This proves $M/B\leq U/S$.
Vanishing tails follow by continuity or gapless support. The ratio $U/S$
increases with $b$, again by the covariance identity applied to the powers
of $b$ in $w_i$. A uniform upper bound on $b$ therefore supplies one
fixed $\lambda=S(b)/U(b)$ for an entire interval.

For the adjacent-rank assertion, put $P=\Pr[X+Y=1]$ and $L=L_0-M$.
Both degree-one coefficients are positive, and
Lemma~\ref{lem:candidate_adjacent_rank_balance} gives the loss
$(P+2q-2L)^2/(8(P+q))$. The numerator's unsquared quantity is nonnegative.
For $t\geq P$,
\[
\frac{d}{dt}\frac{(t+2q-2L)^2}{t+q}
=\frac{(t+2q-2L)(t+2L)}{(t+q)^2}\geq0.
\]
Since $P+q\leq1-B\leq1-\lambda M$, the loss is at most
\[
\frac{(1+q-2L_0+(2-\lambda)M)^2}{8(1-\lambda M)}.
\]
Its numerator before squaring is nonnegative at the actual $M$ and
increases thereafter. Replacing $M$ by $T$ therefore increases the loss.
This argument uses $B\geq\lambda M$, not the generally false inequality
$B\geq\lambda T$. The second-derivative identity in
Lemma~\ref{lem:concave_rank_certificate}, now with this $W,\Phi$,
proves concavity. For a law of degree at most three, $B=M/3$, so
$\lambda=1/3$ is available exactly.
\end{proof}

We now package two rounds of atom and tail improvement into explicit functions.
These profiles keep the subsequent probability estimates finite and reproducible.

\begin{definition}[Two-step tail profiles]
\label{def:tail_profiles}
For $q\geq0$, $\alpha\geq0$, and $f\in\{0,1\}$, define
\begin{align*}
c_1(q)&=1+\alpha-2q-\mathcal T(q,0.36,3/2),\\
j_1(q)&=J(c_1(q))\quad(f=0),\quad j_1(q)=3/2\quad(f=1),\\
c_2(q)&=1+\alpha-2q-\mathcal T(q,c_1(q),j_1(q)),\\
T_{\alpha,f}(q)&=\mathcal T(q,c_2(q),J(c_2(q))),&
E_{\alpha,f}(q)&=\mathcal E(q,c_2(q),J(c_2(q))).
\end{align*}
For the degree-three profile, put
\[
g_1(q)=1-2q-q^2/0.36,\quad
g_2(q)=1-2q-\frac{3q^2}{g_1(q)K(g_1(q))},\quad
T_3(q)=\frac{3q^2}{g_2(q)K(g_2(q))}.
\]
The ordinary profile $T_c^{\rm ord}$ is defined by two updates
$c_0=c$, $c_{j+1}(q)=1-2q-\mathcal T(q,c_j(q),3/2)$, $j=0,1$,
followed by $T_c^{\rm ord}(q)=\mathcal T(q,c_2(q),3/2)$.
\end{definition}

Before applying these profiles, we establish both their probabilistic validity and
the curvature properties needed for endpoint checks.

\begin{lemma}[Validity and concavity of the tail profiles]
\label{lem:tail_profiles}
Suppose $a_1\geq0.36$, $a_2=q$, and $\E[Z]\geq1+\alpha$.
On an interval where $2q\geq\alpha$, $c_1,c_2\geq1/2$, and each
used $J(c_j)\geq3/2$, and every factorial ratio parameter is below four,
the quantities $T_{\alpha,f},E_{\alpha,f}$
bound respectively the omitted mean and excess in
Lemma~\ref{lem:normalized_rank_tail}. They are nonnegative, increasing,
and convex. The degree-three profile has the same properties on
$[0,0.1]$, with $T_3'(q)\geq3q/2$; it bounds the omitted mean when
$\E[Z]\geq1$, $a_1\geq0.36$, and the degree is at most three.
The ordinary profile is valid, increasing, and convex on $[0,0.11]$
for $c=0.33$, and on $[0,0.1]$ for $c=0.36$.
\end{lemma}

\begin{proof}
The identity $a_1=\E[Z]-2q-\sum_{j\geq3}ja_j$ justifies each
update from the preceding proved tail; only two updates are used.
Put $u=1-c$ and
\[
f(x)=\frac{1-2x/3}{2(1-x/2)^2},\quad
g(x)=\frac{1-x}{4(1-3x/4)^2}.
\]
Then $1/(cJ(c))=f(x_c)/c$ and $1/(cK(c))=g(y_c)/c$.
Direct differentiation gives
\begin{align*}
f'(x)&=\frac{1-x}{6(1-x/2)^3},&
f''(x)&=\frac{1-2x}{12(1-x/2)^4},\\
g'(x)&=\frac{2-3x}{16(1-3x/4)^3},&
g''(x)&=\frac{3-9x}{32(1-3x/4)^4}.
\end{align*}
Thus $f$ is nonnegative, increasing, and convex on $[0,1/2]$,
and $g$ has these properties on $[0,1/3]$.
For either $(L,\nu)$, abbreviate $q=q_L(x)$, $w=w_L(x)$, and let
$S$ minimize $d_L(x)$. Differentiation of the minimizing value gives
\[
d_L''(x)=q''S+w''S^2-\frac{(q'+2w'S)^2}{2/S^3+2w}.
\]
Using $1/S^2=q+2wS$, multiplication by its positive denominator yields
\begin{align*}
(2/S^3+2w)d_L''(x)
&=2q''q-(q')^2+(6q''w+2qw''-4q'w')S\\
&\phantom{=}+(6ww''-4(w')^2)S^2
=-(1+2LxS)^2<0.
\end{align*}
Consequently $X_L(\delta)$ is increasing and convex, including at
zero by continuity, and $x_c,y_c<1/3$.
Since $\delta=u/(1-u)$ and $1/c=1/(1-u)$ are nonnegative,
increasing, and convex, both reciprocal products have these properties
as functions of $u\in[0,1/2]$.

Now $u(q)=2q+T_{\rm previous}(q)-\alpha$ is nonnegative,
increasing, and convex. Composition preserves these properties,
as do products, by $(ab)''=a''b+2a'b'+ab''\geq0$.
The formulas for $\mathcal T,\mathcal E$ are sums of positive multiples
of $q^2/(ck)$ times powers of $3q/(ck)$, proving the assertion inductively.
The fixed coefficient $3/2$ is treated in the same way.
The right endpoint bounds on the atoms and coefficients suffice,
since the atoms decrease. The condition $2q\geq\alpha$ ensures $c_j\leq1$.
All ratio denominators are checked positive at the right endpoint.
For the ordinary profiles the two updated atoms increase from their
initial value and stay positive at the stated endpoint.

For degree three, the omitted mean is $3a_3$. Ordinary degree-three
Newton gives the first update, and Lemma~\ref{lem:atom_newton} gives the
second and final tail. On $[0,0.1]$, we have
$g_1\geq139/180>3/4$ and $g_1K(g_1)>9/4>3(0.36)$, so
$g_2\geq g_1>3/4$; also $K(g_j)>3$. Since $1/(g_2K(g_2))\geq1/4$ is nondecreasing in $q$,
differentiating $T_3(q)=3q^2/(g_2K(g_2))$ gives $T_3'(q)\geq3q/2$.
Similarly, $T_{\alpha,f}(q)\geq3q^2/2$,
since $c_2\leq1$, $J(c_2)\leq2$, and $U\geq3$.

For finite evaluation only, round each intermediate tail upward and each
updated atom downward to multiples of $10^{-20}$. Evaluate $J,K$ with
the directed root bounds in Section~\ref{sec:arithmetic_details}.
Since $cJ(c)$ and
$cK(c)$ increase, this gives upper enclosures for the final tail and
excess. Choose the interval coupling at its right endpoint using the
downward atom enclosure. Concavity concerns the unrounded analytic
profile, not its rounded evaluations. In particular, certify $W\geq0$
using the lower bound $3q^2/2$ (or the displayed exact numerator),
never merely an upper enclosure for the tail. Endpoint losses may then
use the upper enclosures conservatively.
\end{proof}

For some rank-two events, a bound on one coordinate mean is more informative than
a uniform tail estimate. The next lemma combines this bound with nonzero
probabilities for the other two coordinates.

\begin{lemma}[A mean-sensitive rank-two split]
\label{lem:mean_sensitive_split}
Let $(X,Y,V)$ have a stable homogeneous quadratic probability generating
polynomial. Suppose
\[
\E[X]\leq U,\quad 1/3\leq U<1,\quad
\Pr[Y\geq1]\geq b,\quad 0<b\leq1/5,\quad \Pr[V\geq1]\geq1/4.
\]
Then $\Pr[X=0,Y=V=1]\geq(1-U)b$.
\end{lemma}

\begin{proof}
Write the generating polynomial as $z^{\mathsf T}Qz$ and put
$\mu=Q\mathbf1$, so $\mathbf1^{\mathsf T}Q\mathbf1=1$.
Real-rootedness on every real line parallel to $\mathbf1$ implies
$\mu\mu^{\mathsf T}-Q\succeq0$ by the discriminant test.
Let $x=\E[X]$, $y=\E[Y]$, $a_X=\Pr[X=2]$, $a_Y=\Pr[Y=2]$,
and $d=\Pr[X=Y=1]$. A principal minor therefore gives
\[
d/2\leq xy/4+\sqrt{(x^2/4-a_X)(y^2/4-a_Y)}
\leq x(y+\sqrt{y^2-4a_Y})/4.
\]
The real-rooted marginal of $Y$ is a sum of two Bernoulli variables,
with parameters $s\geq t$. Hence $d\leq xs\leq Us$.
Put $T=\Pr[Y\geq1]=s+t-st$ and $F=\Pr[X=0,Y=V=1]$.
Since $T=st+d+F$, if $T\leq1/4$, then $s\leq T$ and
\[
F-(1-U)T\geq t(U(1-s)-s)\geq0.
\]
If $T>1/4$, then $\Pr[X=0]\geq1-U>0$. Deleting $X$ increases
the two other nonzero probabilities by negative association. Under this
stable rank-two law both are at least $1/4$, so
Lemma~\ref{lem:quadratic_bounds} gives a middle atom greater than $1/5$:
$4(1/4-1/5)(1-1/4)-(1/5)^2=11/100>0$.
Multiplying by $\Pr[X=0]$ proves the claim.
\end{proof}

\subsection{Conditioning and structured tree-law estimates}
\label{sec:tree_law_estimates}

Conditioning on the presence or absence of a bundle changes the means of the
remaining counts. We record bounds for binary selection and coordinate deletion
that will control these changes.

\begin{lemma}[Binary covering and deleted mass]
\label{lem:binary_covering}
Let a homogeneous strongly Rayleigh law have a binary coordinate $E$,
with $0<q:=\Pr[E=1]<1$, and let $F$ be a disjoint coordinate set.
Then
\[
0\leq\E[F_T]-\E[F_T\mid E=1]\leq1-q,\quad
0\leq\E[F_T\mid E=0]-\E[F_T]\leq q.
\]
Deleting a coordinate set $D$ with positive deletion probability
increases the mean of any disjoint coordinate set by at most $x(D)$.
\end{lemma}

\begin{proof}
Project away $E$. Its two conditional laws have adjacent total ranks.
The coupling in \cite[Theorem~2.10]{kko21} adds exactly one remaining
coordinate when passing from $E=1$ to $E=0$. Their $F$-mean difference
therefore lies in $[0,1]$. Averaging with weights $q,1-q$ proves
both inequalities. For deletion of $D$, negative association makes
all surviving-coordinate changes nonnegative. Homogeneity makes their
total change exactly $x(D)$, which proves the last assertion.
We use the inclusion loss $1-q$ in this form, rather than the $q$
printed in the last inequality of \cite[Lemma~A.1]{gkl24}.
\end{proof}

A proper cut must contain a tree edge, which rules out the constant term in its
count polynomial. The next estimate exploits this support restriction to improve
the conversion from capacity to a target coefficient.

\begin{lemma}[Compulsory-cut coefficient bound]
\label{lem:candidate_compulsory_capacity}
For a polynomial with nonnegative coefficients, write
\[
\operatorname{Cap}_{\boldsymbol\kappa}(f)
:=\inf_{\boldsymbol x>0}
\frac{f(\boldsymbol x)}{\prod_i x_i^{\kappa_i}}.
\]
If $q(x,y)$ is real stable with nonnegative coefficients and $q(0,0)=0$, then
\[
[xy]q\geq\frac{e^{-1}}2\operatorname{Cap}_{(1,1)}(q).
\]
Consequently, if $f(x,y,z)$ is real stable with nonnegative coefficients
and $f(0,0,z)=0$ identically,
then
\[
[xyz]f\geq\frac{e^{-2}}2\operatorname{Cap}_{(1,1,1)}(f).
\]
\end{lemma}

\begin{proof}
We first prove the needed univariate contraction.
Let $R(t)$ have nonnegative coefficients, only nonpositive real roots,
and degree at most $D\geq2$. Then
\[
[t]R\geq(\frac{D-1}{D})^{D-1}\operatorname{Cap}_{1}(R)
\geq e^{-1}\operatorname{Cap}_{1}(R).
\]
If $R(0)=0$, letting $t$ decrease to zero in $R(t)/t$ proves the
stronger factor one. Otherwise write
\[
R(t)=R(0)\prod_{i=1}^{D}(1+a_it),\quad a_i\geq0,
\]
padding with zero parameters if necessary, and set $A:=\sum_i a_i$.
For $A=0$ the capacity is zero. For $A>0$, the arithmetic-geometric
mean inequality gives
\[
R(t)\leq R(0)(1+At/D)^D.
\]
The infimum of the right side divided by $t$ is attained at
$t=D/(A(D-1))$. Hence
\[
\operatorname{Cap}_1(R)
\leq R(0)A(\frac{D}{D-1})^{D-1}
=[t]R(\frac{D}{D-1})^{D-1}.
\]
The bound $(1+1/(D-1))^{D-1}\leq e$ proves the assertion.
For degree one the capacity is exactly $[t]R$, and for degree zero
it is zero, so the $e^{-1}$ contraction holds in every degree.

Now let $N:=\deg q$. The zero polynomial is immediate.
If $N<2$, scaling $x=y$ to infinity gives
$\operatorname{Cap}_{(1,1)}(q)=0$.
For $N\geq2$, \cite[Theorem~4.5]{bbl09} shows that the homogenization
\[
Q(x,y,w):=w^Nq(x/w,y/w)
\]
is stable. This closure theorem applies to all nonnegative stable
polynomials, without a multiaffinity assumption.
Since $q(0,0)=0$, $\deg_wQ\leq N-1$, and homogeneity gives
\[
\operatorname{Cap}_{(1,1,N-2)}(Q)
=\operatorname{Cap}_{(1,1)}(q).
\]
Let $H(x,y):=[w^{N-2}]Q$, the degree-two part of $q$.
Differentiation followed by real specialization shows that $H$ is stable
or zero.

For $N\geq3$, fix $x,y>0$, put $D:=N-1$, and reverse the polynomial in $w$:
\[
R(t):=t^D Q(x,y,1/t).
\]
It has nonnegative coefficients and only nonpositive real roots, with
\[
[t]R=H(x,y),\quad
\operatorname{Cap}_{1}(R)
=\inf_{w>0}\frac{Q(x,y,w)}{w^{N-2}}.
\]
The univariate contraction, followed by division by $xy$ and
infimization over $x,y$, yields
\[
\operatorname{Cap}_{(1,1)}(H)
\geq e^{-1}\operatorname{Cap}_{(1,1)}(q).
\]
For $N=2$ this holds with factor one, since taking $w$ down to zero
in the capacity of $Q$ leaves exactly $H$.
Write $H=ax^2+bxy+cy^2$. Stability gives $b^2\geq4ac$, and therefore
\[
\operatorname{Cap}_{(1,1)}(H)=b+2\sqrt{ac}\leq2b.
\]
This proves the bivariate assertion, including degenerate quadratics.

Finally set $q(x,y):=[z]f(x,y,z)$.
It is stable or zero and has $q(0,0)=0$.
Apply the univariate contraction in $z$ for each $x,y>0$ and take
the infimum after division by $xy$:
\[
\operatorname{Cap}_{(1,1)}(q)
\geq e^{-1}\operatorname{Cap}_{(1,1,1)}(f).
\]
The bivariate assertion gives the claimed $e^{-2}/2$ factor.
If $q$ is zero, the displayed inequality forces the latter capacity
to be zero, so the conclusion still holds.

\end{proof}

The competing-half-bundle argument needs a bound on how deleting one edge changes
the count at a remote endpoint. An electrical-network calculation provides the
sharper estimate used below.

\begin{lemma}[Electrical remote influence]
\label{lem:final23_electrical_remote}
In the setting of \cite[Lemma~2.27]{kko21}, let $u,v,w$ be distinct
vertices and let $f=uv$, $g=vw$ have marginals in $[1/2-e,1/2+e]$,
where $0\leq e<1/6$. The smaller increase in the mean of
$\delta(w)\setminus\{g\}$ on deleting $f$, or of
$\delta(u)\setminus\{f\}$ on deleting $g$, is at most
\[
\frac{1/2+e}{1/2-e}\max_{0\leq z\leq1}\{z/(1+z)-(1/2-e)z^2\}<0.22
\quad\text{if }0\leq e\leq0.00550006.
\]
Consequently the bound $0.22$ applies after endpoint-tree conditioning
whenever $h\leq0.0055$ and $d_0\leq2\cdot10^{-8}$.
\end{lemma}

\begin{proof}

The exact max-entropy law may lie on a face of the marginal polytope.
To justify the weighted-tree calculation there, let $\bar\mu$ be
uniform on all spanning trees of the support graph, with marginals
$\bar x$, and let $\mu_t$ maximize entropy at
$x_t:=(1-t)x+t\bar x$, $0<t<1$. This is a relative-interior point,
so $\mu_t$ is a positive weighted-tree law. Every subsequential
limit $\nu$ has marginals $x$, and
\[
H(\mu_t)\geq H((1-t)\mu+t\bar\mu)\longrightarrow H(\mu).
\]
Continuity of entropy on the finite probability simplex and uniqueness
of the maximizer imply $\nu=\mu$. Hence $\mu_t\to\mu$.
Stability and the weighted-tree factorization identities pass to this
limit. Conditional identities pass whenever their limiting conditioning
probabilities are positive, as they are at every use below.

Let $f=uv$ and $g=vw$ be the two distinct incident edges, with
conductances $c_f,c_g$ and marginals $a,b$. Put
$U:=\delta(u)\setminus\{f\}$ and $W:=\delta(w)\setminus\{g\}$.
For a set $F$ disjoint from $f$, write
\[
I_f(F):=\E[F_T\mid f\notin T]-\E[F_T].
\]
Let $L^+$ be the graph Laplacian pseudoinverse, and let $b_j$ be an
oriented incidence vector for edge $j$. Differentiating the reduced
Laplacian determinant gives
\[
\Pr[j\in T]=c_jb_j^TL^+b_j,\quad
\operatorname{Cov}(\mathbf1_f,\mathbf1_j)
=-c_fc_j(b_f^TL^+b_j)^2\quad(j\ne f).
\]
Thus, for the unit-current voltage $V:=L^+b_f$,
\[
I_f(F)=\frac{c_f}{1-a}\sum_{j\in F}c_j(b_j^TV)^2.
\]

Eliminate all vertices other than $u,v,w$ by harmonic extension.
Write the reduced triangle conductances as $A$ on $uv$, $B$ on $vw$,
and $C$ on $uw$, and set $D:=AB+AC+BC>0$.
Then $A\geq c_f$ and $B\geq c_g$.
For an interior vertex $j$, let $H_{jt}$ be its harmonic measure on
the boundary. These numbers are nonnegative and sum to one, so
\[
(V_w-V_j)^2\leq
\sum_{t\in\{u,v,w\}}H_{jt}(V_w-V_t)^2.
\]
Multiplying by $c_{wj}$ and summing, including direct boundary edges,
gives the boundary-star inequality
\[
\sum_{j\in\delta(w)}c_j(b_j^TV)^2
\leq C(V_w-V_u)^2+B(V_w-V_v)^2.
\]
Indeed, the coefficient at boundary vertex $t$ is
$c_{wt}+\sum_jc_{wj}H_{jt}$, exactly the corresponding
Schur-complement conductance. In particular, reduction bounds the
original star energy; it does not identify that energy with the total
energy of the reduced network.

Take $V_v=0$. The unit-current equations and the two marginal
identities give
\[
V_u=(B+C)/D,\quad V_w=C/D,\quad
c_f=aD/(B+C),\quad c_g=bD/(A+C).
\]
Subtract the actual $g$-energy from the star bound. Substitution and
the symmetric calculation yield, with
$z:=C/\sqrt{(A+C)(B+C)}\in[0,1)$,
\[
I_f(W)\leq\frac a{1-a}(BC/D-bz^2),\quad
I_g(U)\leq\frac b{1-b}(AC/D-az^2).
\]
Both parentheses are nonnegative: before simplification their energy
coefficients are $C,B-c_g$ and $C,A-c_f$, respectively.
Also
\[
\min\{AC/D,BC/D\}\leq z/(1+z).
\]
For example, if $A\leq B$, put $S:=\sqrt{(A+C)(B+C)}$.
Then $A\leq S-C$, so $A(S+C)\leq D$, proving the assertion.
It follows that
\[
\min\{I_f(W),I_g(U)\}
\leq\frac{1/2+e}{1/2-e}
\max_{0\leq z\leq1}\{z/(1+z)-(1/2-e)z^2\}.
\]
For $f_e(z):=z/(1+z)-(1/2-e)z^2$, we have
$f_e''(z)\leq-(1-2e)$. The quadratic tangent bound at $z_0:=233/500$ gives
\[
\max_{0\leq z\leq1}\{f_e(z)\}
\leq f_e(z_0)+\frac{f_e'(z_0)^2}{2(1-2e)}.
\]
The maximum and the prefactor increase with $e$. Exact rational substitution
at $e=0.00550006$ gives
\[
\frac{1/2+e}{1/2-e}
(f_e(z_0)+\frac{f_e'(z_0)^2}{2(1-2e)})<0.22.
\]
Vanishing reduced conductances follow by continuity; the interior Dirichlet
matrix is invertible because every interior component touches the boundary.

Finally, conditioning disjoint endpoint atoms to be trees and contracting
them preserves the exterior weighted-tree law. Parallel edges aggregate
with their summed conductance and preserve the bundle count.
The endpoint-tree mean-change budget gives conditional half-width
at most $h+3d_0$, so the same bound applies to the half bundles.
\end{proof}

We next turn bounds on the two block means into a splitting probability
conditional on their total rank being two. The accompanying atom inequalities will
support the later rank updates.

\begin{lemma}[Conditional rank-two split]
\label{lem:final24_rank_split}
Let $X,Y$ count disjoint coordinate blocks of a strongly Rayleigh law,
let $Z=X+Y\geq1$, and write $p_i:=\Pr[Z=i]$ and $m:=\E[Z]$.
Suppose $p_2>0$ and $\E[X],\E[Y]\leq2-b$. If
$0<b-p_1\leq1$, then
\[
\Pr[X=Y=1\mid Z=2]\geq\psi(b-p_1),\quad
\psi(v):=v(1-v/2).
\]
If $m\geq1+\nu$ with $\nu>0$, then
\[
p_1\leq e^{-\nu},\quad p_1\leq p_2/\nu,\quad
p_2^2\geq2p_1p_3.
\]
If also $m\leq3$, then
\[
p_3\geq1-3p_1-2p_2,\quad 3-m\leq2p_1+p_2.
\]
\end{lemma}

\begin{proof}
Project onto the union of the two coordinate blocks. The consecutive
rank laws admit the one-coordinate monotone coupling of
\cite[Theorem~2.10]{kko21}. Thus, writing
$\mu_X(j):=\E[X\mid Z=j]$,
\[
\mu_X(j)\leq\mu_X(2)+\max\{j-2,0\}.
\]
The rank support has no gaps. Averaging gives
\[
\mu_X(2)\geq\E[X]-(m-2+p_1)=2-\E[Y]-p_1\geq b-p_1.
\]
The same holds for $Y$. The rank-two law is stable, and the two mean bounds give
$\E[X\mid Z=2]\in[b-p_1,2-b+p_1]$.
Lemma~\ref{lem:quadratic_bounds} proves the split bound.

The real-rooted generating polynomial of $Z$ has a compulsory factor
of its variable. Hence $Z-1$ is Poisson-binomial. Its zero atom is at
most $e^{-\nu}$. When $p_1>0$, let $t_j$ be the odds of its
Bernoulli factors. Then
\[
\frac{p_2}{p_1}=\sum_jt_j\geq m-1\geq\nu,\quad
p_2^2-2p_1p_3=p_1^2\sum_jt_j^2\geq0.
\]
The case $p_1=0$ is immediate. Finally, assigning rank at least four
to every term beyond rank three gives
$p_3\geq4-m-3p_1-2p_2$; assigning rank at least three beyond
rank two gives $3-m\leq2p_1+p_2$.
\end{proof}

Conditioning a counted union to a specified total rank also shifts the means of
its subunions. The following comparison charges this shift to the excess and
deficit of the total rank.

\begin{lemma}[Projected full-rank mean changes]
\label{lem:projected_rank_mean_shift}
Project a strongly Rayleigh law to a counted union with total rank $R$.
If $\Pr[R=k]>0$ and $F$ is any subunion of the counted blocks, then
\[
\E[F]-\E[(R-k)_+]\leq\E[F\mid R=k]
\leq\E[F]+\E[(k-R)_+].
\]
The errors apply once to the union $F$, not once per block.
\end{lemma}
\begin{proof}
Consecutive supported full ranks admit a coupling differing by one coordinate
\cite[Theorem~2.10]{kko21}.  Thus the conditional $F$-mean at rank $j$ differs
from the rank-$k$ mean by a number in $[0,j-k]$ when $j\geq k$, and in
$[j-k,0]$ when $j\leq k$.  Average these two inequalities.
\end{proof}

We use the up/down inequality only on disjoint complementary counted sets with
positive conditioning probability: for $A+B=n_A+n_B$,
\[
\Pr[A\geq n_A\mid A+B=n_A+n_B]
\geq\Pr[A\geq n_A]\Pr[B\leq n_B],
\]
The lower-tail statement is symmetric. See \cite[Lemma~A.10]{gkl24}
and \cite[Lemma~5.4]{kko21}.
A stable count polynomial is polarized when needed, and every rank condition
below is a full rank after projection.  Outside coordinates are not retained
under arbitrary subset-rank conditioning.

For total rank three, we need an event that selects one element from each of three
blocks. The following lemma obtains this event from two-sided bounds on the
coordinate means.

\begin{lemma}[Homogeneous cubic split]
\label{lem:cubic_split}
Let $X+Y+Z=3$ have a stable homogeneous probability generating
polynomial with all coordinate means in $[\ell,U]$, where
$0<\ell\leq0.01$. For either row below,
\[
\Pr[X=Y=Z=1]>c_U\ell e^{-\ell}.
\]
\[
\begin{array}{c|c|c|c}
U&F_U&v_U&c_U\\ \hline
1.504&0.434&0.21&0.3728\\
1.53&0.4147&0.19&0.3594
\end{array}
\]
\end{lemma}

\begin{proof}
Choose the coordinate of smallest mean, say $X$, and put $\mu=\E[X]$.
Then $\ell\leq\mu\leq1$, and each other mean is at least $(3-U)/2$.
Lemma~\ref{lem:bernoulli_estimates} gives
$\Pr[X=1]\geq\mu e^{-\mu}>0$. Project to $Y,Z$ and condition
their full rank to two, which is precisely the event $X=1$.

If $\mu\geq0.02$, both unconditioned nonzero tails of $Y,Z$ exceed
$1-e^{-(3-U)/2}>0.52$, and their CDFs at one exceed $F_U$.
The up/down inequality gives both conditional tails at one at least
$\alpha_U=0.52F_U$. The rational test in
Lemma~\ref{lem:quadratic_bounds} gives conditional central mass greater
than $v_U$, since
\[
4(\alpha_U-v_U)(1-\alpha_U)-v_U^2>0
\]
in both rows. Hence the joint atom exceeds
\[
0.02e^{-0.02}v_U>0.02(0.98)v_U>0.01c_U
\geq c_U\ell e^{-\ell}.
\]

If $\mu\leq0.02$, put $W=Y+Z=3-X$. Lemma~\ref{lem:projected_rank_mean_shift} gives
\begin{align*}
\E[Y\mid W=2]
&\geq\E[Y]-\Pr[X=0]\\
&=2-\E[Z]-\Pr[X=2]-2\Pr[X=3]\\
&\geq2-U-\mu^2/2\geq2-U-0.0002.
\end{align*}
Here $\E[X(X-1)]\leq\mu^2$ follows from the Bernoulli factorization.
The same lower bound holds for $Z$. Lemma~\ref{lem:quadratic_bounds}
gives conditional central mass at least $\psi(2-U-0.0002)>c_U$.
Multiply by $\Pr[X=1]\geq\ell e^{-\ell}$.
Both row comparisons are exact rationals. The coordinate choice is
made from the law, and all rank conditioning is after projection.
\end{proof}

In other applications, marginal tail probabilities are available instead of
precise means. We therefore give a second rank-three estimate directly in terms of
the marginal tails.

\begin{lemma}[Homogeneous cubic bound from two marginal tails]
\label{lem:cubic_two_tail}
Let $X+Y+Z=3$ have a stable homogeneous probability generating polynomial.
Suppose every coordinate has nonzero probability at least $\alpha=0.2574$
and probability at most one at least $\beta$, where $0<\beta\leq0.01$.
Put
\[
c_0=0.4911=3(0.1637),\quad b=\frac{\beta}{1+\beta/c_0},\quad
\psi(v)=v(1-v/2).
\]
Then $\Pr[X=Y=Z=1]\geq b\psi(\alpha-b^2/c_0)$.
\end{lemma}
\begin{proof}
Choose a coordinate of largest mean, say $X$, and put $\mu=\E[X]$ and
$p_j=\Pr[X=j]$. This choice depends on the law, not on its realization.
Then $1\leq\mu\leq3-2\alpha<2.5$. Gapless rank support, $p_0+p_1\geq\beta>0$,
and $\mu\geq1$ imply $p_1>0$.
The Bernoulli factorization of $X$ has degree at most three. Padding by zero
parameters and applying arithmetic--geometric mean gives
\[
p_0\leq(1-\mu/3)^3,\quad \mu\leq3(1-p_0^{1/3}).
\]
Each other coordinate has mean at least $\alpha$ and at least $3-2\mu$.
Project to the other pair and condition its full rank to two, an event of
probability $p_1>0$. Lemma~\ref{lem:projected_rank_mean_shift} loses at most
$\E[(Y+Z-2)_+]=p_0$. Thus both conditional means are at least
\[
\max\{\alpha-p_0,6p_0^{1/3}-3-p_0\}.
\]
This quantity exceeds $0.095$: for $p_0\leq0.16$ use
$\alpha-0.16=0.0974$; above $0.16$ the second expression increases, and
$0.5425^3<0.16$ with $6(0.5425)-3-0.16=0.095$ suffices.
Lemma~\ref{lem:quadratic_bounds} bounds the conditional central atom
below by $\psi(0.095)$.
If $p_1\geq0.1$, the homogeneous quadratic mean bound gives joint atom at
least $0.1\psi(0.095)$.

Suppose $p_1<0.1$. The rank-one bound in Lemma~\ref{lem:bernoulli_estimates} on $[1,1.2]$ forces
$\mu>1.2$, and its rank-two bound on $[1.2,2.5]$ gives $p_2>0.1637$.
Newton's degree-three inequality gives
\[
p_1^2\geq3p_0p_2,\quad p_0\leq p_1^2/c_0.
\]
The inequality $b+b^2/c_0<\beta\leq p_0+p_1$ implies $p_1\geq b$.
Both conditional pair means are at least $\alpha-p_1^2/c_0>0$.
For $f(x)=x\psi(\alpha-x^2/c_0)$,
\[
f'(x)=\psi(\alpha)-3(1-\alpha)x^2/c_0-5x^4/(2c_0^2).
\]
This decreases with $x^2$ and is positive at $x=0.1$, by exact substitution.
Thus the joint atom is at least $f(p_1)\geq f(b)$.
Finally $f(b)\leq0.01\psi(\alpha)<0.1\psi(0.095)$, so the first case also
suffices.
The full-rank moment comparison follows from the adjacent-rank coupling
of \cite[Theorem 2.10]{kko21}, after projection. No independence of actual
sampled edges is used.
\end{proof}

\subsection{Good bundles and window estimates}
\label{sec:good_common_events}

For non-half top bundles and bottom bundles, retain
\cite[Definition~5.13]{kko21}. A half top bundle is good when its
conditional $2$--$2$ probability is at least $\gamma_{\rm good}h$.
We use this definition directly, without an inclusion assumption for
the source good-half class. The selected-state tail, competing-half estimate,
and structural trichotomy below establish its required interfaces directly.

We first check that deleting unwanted coordinates preserves the lower-tail
information on the total residual count.

\begin{lemma}[Deletion and the total-rank lower tail]
\label{lem:total_rank_deletion}
Let $\nu$ be a strongly Rayleigh law on a finite ground set $F$,
let $R:=|T\cap F|$, and let $D\subseteq F$.
If $\Pr[D_T=0]>0$, then for every integer $k$,
\[
\Pr[R\leq k\mid D_T=0]\geq\Pr[R\leq k].
\]
\end{lemma}

\begin{proof}
The adjacent-rank monotone coupling in \cite[Theorem~2.10]{kko21}
makes $a_j:=\Pr[D_T=0\mid R=j]$ nonincreasing on the rank support.
The function $j\mapsto\mathbf1_{\{j\leq k\}}$ is also nonincreasing.
For an independent copy $R'$ of $R$, the covariance identity gives
\[
\begin{aligned}
2\operatorname{Cov}(a_R,\mathbf1_{\{R\leq k\}})
&=\E[(a_R-a_{R'})(\mathbf1_{\{R\leq k\}}-\mathbf1_{\{R'\leq k\}})]\\
&\geq0.
\end{aligned}
\]
Divide the resulting joint-probability inequality by $\Pr[D_T=0]$.
The coupling uses full ranks of the projected law on $F$.
\end{proof}

Goodness supplies a useful lower tail after a half bundle has been selected. The
following bound retains that information when the degree remainder is subsequently
deleted.

\begin{lemma}[Selected-state tail of a good half bundle]
\label{lem:candidate_two_tail}
Let $e=(u,v)$ be a good half bundle. Conditional on the endpoint trees and
on $e_T=1$, its residual union count $R$ satisfies
\[
\Pr[R\leq2]\geq L_*:=
\frac{(\gamma_{\rm good}-2)h-4\rho_t}{3/2-h-\rho_t},\quad
\rho_t:=3d_0.
\]
After deleting the degree remainder $C$, the same lower bound holds for
the projected total residual rank. At the fixed parameters,
\[
L_*>0.00209.
\]
\end{lemma}

\begin{proof}
Work first under the endpoint-tree law, let $R$ denote its residual union
count before conditioning on $E:=e_T$, and put $q:=\Pr[E=1]$.
Sequential maximum-rank conditioning has total mean-change budget at most
$3d/2$. Thus, with the conservative $\rho_t=3d_0$,
\[
q\in[1/2-h-\rho_t,1/2+h+\rho_t],\quad
\E[\delta(u)_T],\E[\delta(v)_T]\in[2-\rho_t,2+\rho_t].
\]
The bundle is binary. Lemma~\ref{lem:binary_covering} and the adjacent-rank coupling give a
coupling of the residual union counts $R_1,R_0$ under $E=1,E=0$ with
$R_0-R_1\in\{0,1\}$. The union is a proper cut, so $R_1\geq1$.
When $E=0$, the two endpoint cuts each contribute a compulsory success,
so $R_0\geq2$.

For every such pair of integer counts,
\[
q1_{\{R_1=2\}}+(1-q)\mathbf1_{\{R_0=4\}}
\leq R_1+(1-q)(R_0-R_1)-3+(1+q)\mathbf1_{\{R_1\leq2\}}.
\]
For $R_1=1$, necessarily $R_0=2$, and equality holds. For $R_1=2$, the
two possible right sides are $q$ and $1$. For $R_1=3$, equality holds for
both increments. For $R_1\geq4$, the inequality follows directly, with
right side at least $R_1-3$.

Write $A:=\Pr[R_1\leq2]$ and let $g$ be the conditional $2$--$2$
probability of the bundle. Its present- and absent-edge events require
$R_1=2$ and $R_0=4$, respectively. Averaging the preceding inequality gives
\[
g\leq \E[R]-3+(1+q)A
\leq1-2q+2\rho_t+(1+q)A.
\]
Here $\E[R]=\E[\delta(u)_T]+\E[\delta(v)_T]-2q$.
For $0\leq A\leq1$, the final expression decreases with $q$. Substitute
$q\geq1/2-h-\rho_t$ and $g\geq\gamma_{\rm good}h$ to obtain the claim.
The numerator and denominator defining $L_*$ are strictly positive.

After endpoint and union tree conditioning, the unique parallel-bundle
choice factors from the residual weighted-tree law, as in Lemma~\ref{lem:candidate_gkl_window}.
Splitting $C$ into its bundle and residual parts shows that deleting the
bundle part only restricts the independent selected edge. Project to the
residual union and apply Lemma~\ref{lem:total_rank_deletion} to deletion of its residual part. This
cannot decrease the total-rank lower tail. The conditions commute and all
have positive probability by the outer bound in Lemma~\ref{lem:candidate_gkl_window}. The argument
does not apply to the different window-premise subcount.
\end{proof}

To locate the bad bundles, we need a sufficient condition for a half bundle to be
good. The next lemma expresses this condition in terms of the upward mass at an
endpoint.

\begin{lemma}[Good-bundle threshold]
\label{lem:candidate_good_threshold}
In the setting of \cite[Lemma~5.16]{kko21}, if
$x(\delta^\uparrow(u))\geq1/2+k_{\rm good}h$,
then the half bundle incident to $u$ is good.
\end{lemma}

\begin{proof}
Put $k:=k_{\rm good}$. Follow the source conditioning on $u,v,S,W$ being trees,
where $W:=S\setminus u$, and define
\[
X:=\delta^\uparrow(u)_T,\quad
Y:=\delta(v)_T-1,\quad Z:=X+Y.
\]
The unshifted coordinate sets $\delta^\uparrow(u)$ and $\delta(v)$ are
disjoint. Projection and diagonalization of the conditional spanning-tree
generating polynomial therefore give a stable bivariate polynomial.
Since $\delta(v)_T\geq1$, it is divisible by the variable corresponding to
$\delta(v)$. Dividing by that variable preserves stability and gives the
joint generating polynomial of $(X,Y)$. Thus
Lemma~\ref{lem:candidate_adjacent_rank_balance} applies to these shifted counts.

The exact expectation estimates in \cite[Eq.~(23)]{kko21} imply
\[
\E[X],\E[Y]\geq C:=\frac12+kh-4d_0,\quad
\E[Z]\geq1+\alpha,\quad \alpha:=2kh-5d_0.
\]
Also $\E[X]\leq1+d_0$, $\E[Y]\leq1.5-kh+3d_0$, and hence
$\E[Z]\leq2.5-kh+4d_0<2.5$. Both individual means lie in $[1/2,3/2]$.
Here and below replacing the hierarchy error by $d_0$
only weakens the estimates; the source's rounded sufficient condition
$k\geq9$ is not used.

Define
\[
q_c:=(\frac{2\gamma_{\rm good}}{0.164}+0.01)h,
\quad q:=\Pr[Z=2].
\]
Suppose first that $q\geq q_c$.
The four tail probabilities in the first case of \cite[Lemma~5.16]{kko21}
give
\[
\epsilon_0:=(1-e^{-1/2})\frac7{16}.
\]
The up/down truncation inequality in \cite[Lemma~A.10]{gkl24} gives
both tails of $X$ at one, conditional on $Z=2$, at least $\epsilon_0$.
The conditional bivariate rank-two law is stable, so its univariate
rank polynomial is real-rooted. Lemma~\ref{lem:quadratic_bounds}
therefore implies $\Pr[X=1\mid Z=2]>0.164$.
For a rational certificate,
\[
e^{1/2}>\sum_{j=0}^{7}\frac{1}{2^j j!}>
\frac{10^6}{606531},
\]
so $e^{-1/2}<0.606531$.
Put $\alpha_0:=(1-0.606531)7/16<\epsilon_0$. The rational test in
Lemma~\ref{lem:quadratic_bounds} applies because $0.164<\alpha_0$ and
\[
4(\alpha_0-0.164)(1-\alpha_0)-0.164^2
=0.000067933561109375>0.
\]
Consequently
\[
\Pr[X=Y=1]>q_c(0.164)>2\gamma_{\rm good}h.
\]

Suppose now that $q<q_c<0.1637$. The target-two bound in
Lemma~\ref{lem:bernoulli_estimates} forces $\E[Z]<1.2$.
The same lemma gives $\Pr[Z=1]>0.36$.
The target-two bound in \cite[Lemma~2.21]{kko21} gives
\[
q\geq\alpha e^{-\alpha}\geq
q_*:=\alpha(1-\alpha+\alpha^2/2-\alpha^3/6).
\]
Indeed, on $[1+\alpha,1.2]$, the one-forced-success branch is at
least $\alpha e^{-\alpha}$, and the zero-forced branch is at least
$e^{-1.2}/2>\alpha$. The correction exponent is zero.
Split $[q_*,q_c]$ at $0.1$. On its first interval use
$T_{\alpha,0}$ from Definition~\ref{def:tail_profiles}; on its second
use $T_{\alpha,1}$. For either interval $[l,u]$, let $c_2^-(u)$ be
the downward endpoint enclosure in Lemma~\ref{lem:tail_profiles}, and put
\[
b_u=\frac{3u}{c_2^-(u)J(c_2^-(u))},\quad
\lambda_u=\frac{S(b_u)}{U(b_u)},\quad
F_u(q)=\frac q2-
\frac{(q-\alpha+3d_0+(2-\lambda_u)T_{\alpha,f}(q))^2}
{8(1-\lambda_u T_{\alpha,f}(q))}.
\]
The factor $\lambda_u$ is fixed on that interval.
Exact endpoint checks give $2l>\alpha$, $c_1(u),c_2(u)>1/2$,
the required $J$ bounds, positive denominators, and
\[
C-2u-T_{\alpha,f}(u)>0,\quad
l-\alpha+3d_0+(2-\lambda_u)3l^2/2>0.
\]
These verify the hypotheses throughout the interval by
Lemma~\ref{lem:tail_profiles}. Lemma~\ref{lem:normalized_rank_tail}
therefore gives $\Pr[X=Y=1]\geq\min\{F_u(l),F_u(u)\}$.
It remains to restore the outer conditioning. Put $A:=1/2+kh<1$.
The source estimates give
\[
\Pr[W\text{ is a tree}\mid u,v,S\text{ are trees}]\geq A-2d_0,
\]
and
\[
\Pr[S\text{ is a tree}\mid u,v\text{ are trees}]
\geq1-\frac{d_0}{2(1-d_0)}\geq1-d_0.
\]
Their product is at least $(1-d_0)(A-2d_0)\geq A-3d_0>A-4d_0>1/2$.
This proves goodness in the large-rank branch. In the small-rank branch,
each of the four endpoint comparisons is
$CF_u(q)>\gamma_{\rm good}h$, with $q\in\{l,u\}$ on its interval.
Also $0<q_*<q_c$, $1/2+kh<1$, and $kh>4d_0$, so the
source outer and mean calculations apply. This proves the goodness
threshold. 

\end{proof}

We also need to rule out two bad half bundles at the same atom. The next estimate
establishes this exclusion throughout the enlarged parameter range.

\begin{lemma}[Competing half bundles on the enlarged domain]
\label{lem:final24_competing}
For $0<h\leq0.0055$ and
$0\leq d_0\leq\min\{h^2,2\cdot10^{-8}\}$, one of any two half
bundles incident to the same atom in \cite[Lemma~5.17]{kko21}
has conditional $2$--$2$ probability greater than $0.01330$.
In particular, the bad half bundles form a matching for
Definition~\ref{def:final24_parameters}.
\end{lemma}

\begin{proof}
Use the source notation $\mathbf e=(u,v)$, $\mathbf f=(v,w)$, and
\[
U:=\delta(u)\setminus\mathbf e,\quad
V:=\delta(v)\setminus(\mathbf e\cup\mathbf f),\quad
W:=\delta(w)\setminus\mathbf f.
\]
Condition on the three endpoint atoms being trees. Orient the bundles
so that Lemma~\ref{lem:final23_electrical_remote} bounds the increase
of the $W$-mean on deleting $\mathbf e$ by $0.22$. Put
\[
\tau_0:=0.073,\quad o:=0.494,\quad
h_{\max}:=0.0055,\quad d_{\max}:=2\cdot10^{-8}.
\]
The endpoint-tree and edge-selection factors exceed $o$ on this box,
since $(1-3d_{\max})(1/2-h_{\max}-3d_{\max})>o$.

Suppose first that deleting $\mathbf e$ raises the $V$-mean by at
least $\tau_0$. The source conditional-mean identities give, after
removing a compulsory success from each of $U_T$ and
$(V\cup\mathbf f)_T$, disjoint projected counts $X,Y$ with
\[
\E[X],\E[Y]\geq1/2-h-3d_0,\quad
1+v\leq\E[X+Y]\leq1.5+2h+3d_0<1.512,
\]
where $v:=\tau_0-2h-3d_0$. Removing compulsory monomial factors
and polarizing preserve stability. Put $q:=\Pr[X+Y=2]$.
Lemma~\ref{lem:bernoulli_estimates} gives a rank-one atom
greater than $0.33$, and \cite[Lemma~2.21]{kko21} gives
$q\geq ve^{-v}$. If $q\geq0.11$, the up/down inequality
\cite[Lemma~A.10]{gkl24} and the quadratic split, with tail product
$0.39(0.70)$, give split greater than $0.245$. The restored probability
is then greater than $0.494(0.11)(0.245)=0.0133133$.

For $q<0.11$, use Lemma~\ref{lem:normalized_rank_tail} at the
worst mean bound $L_0=1/2-h_{\max}-3d_{\max}$ with
$T=T_{0.33}^{\rm ord}$. Choose its fixed coupling from the final
ordinary profile's atom at $q=0.11$, as in
Lemma~\ref{lem:normalized_rank_tail}. Write the resulting bound as
$F(q)$. Put
\[
v_{\min}:=0.073-2h_{\max}-3d_{\max},\quad
q_*:=v_{\min}(1-v_{\min}+v_{\min}^2/2-v_{\min}^3/6).
\]
The function $v(1-v+v^2/2-v^3/6)$ is increasing for
$0\leq v\leq0.073$, so $q\geq q_*$. Both degree-one coefficients
and $\Phi(q)$ stay positive through $q=0.11$, and $W(q)\geq0$.
Concavity reduces the whole interval to the exact comparisons
\[
oF(q_*)>0.01330,\quad oF(0.11)>0.01330.
\]
This proves the first branch uniformly in the box.

In the other branch, select $\mathbf f$ and delete $\mathbf e$.
After removing the compulsory success, the source mean interval is
contained in $[0.975,M]$, where
\[
M:=1.5+0.22+0.073+3(0.0055)+9\cdot2\cdot10^{-8}=1.80950018.
\]
Its atom at one exceeds $0.172$. To see this in every dimension, use
the fixed-mean extremizer reduction of
Lemma~\ref{lem:bernoulli_estimates}. A forced success gives
$2-M>0.172$. One summand is immediate. For $2\leq n\leq9$,
check the unforced binomial atom at the two mean endpoints. For
$n\geq10$, $M<20/11\leq2n/(n+1)$ gives the lower bound
\[
\min\{0.975e^{-0.975},Me^{-M}\}>0.172.
\]

Return to the law with $\mathbf f$ selected but without the extra
deletion. Put $X:=(V\cup\mathbf e)_T$, $Y:=W_T$, and $Z:=X+Y$.
This is a proper union cut, so $Z\geq1$, and the source bounds give
\[
\E[X],\E[Y]\leq1.5+2h+d_0,\quad
\E[Z]\geq2-4h-6d_0=1+\nu_*,\quad \nu_*:=1-4h-6d_0.
\]
Negative association gives probability at least $o$ for the extra
deletion, so $q:=\Pr[Z=2]\geq q_{\min}:=o(0.172)$.
If $q\geq0.11$, the tail product $0.6275(0.4284)$ again gives split
greater than $0.245$, which suffices. Otherwise
Lemma~\ref{lem:final24_rank_split} gives split at least
\[
\psi(b(q)),\quad b(q):=1/2-2h-d_0-q/\nu_*.
\]
The derivative of $q\psi(b(q))$ on $[q_{\min},0.11]$ is at least
\[
\psi(b(0.11))-0.11(1-b(0.11))/\nu_*>0.
\]
Increasing $h,d_0$ decreases $b$; at their largest values,
$oq_{\min}\psi(b(q_{\min}))>0.01330$ by exact substitution.
This proves both branches. Finally,
$\gamma_{\rm good}h<0.01330$, so two incident bad
half bundles are impossible.
\end{proof}

Two good half bundles concentrated on the same degree side cannot both lack the
lower-tail event needed by the mixed-bundle argument. We quantify this alternative
using the separate crossing cutoff.

\begin{lemma}[Same-side mixed-bundle lower tail]
\label{lem:candidate_same_side_tail}
In the setting of \cite[Lemma~5.23]{kko21}, let
$\mathbf e=(v,u)$ and $\mathbf f=(v,w)$ be good half top bundles,
and let $A,B,C$ be the degree partition of $\delta(v)$ with
$x_{\mathbf e}(B),x_{\mathbf f}(B)\leq\omega$. For the fixed parameters
and $d\leq d_0$, one bundle satisfies
\[
\Pr[U_T+(A-\mathbf e)_T\leq1]\geq0.49(1-M_{\rm same}/2)>Kh,
\]
after interchanging the bundles if necessary, where
\[
U:=\delta(u)-\mathbf e,\quad
M_{\rm same}:=1.72+\omega+4d_0+\frac{4\omega+r+d_0}{0.49}.
\]
\end{lemma}

\begin{proof}
Orient the bundles using Lemma~\ref{lem:final23_electrical_remote}.
Under $\mathcal H:=\{\mathbf f\notin T,\ u,v,w\text{ trees}\}$,
\[
\E[U_T\mid\mathcal H]\leq x(U)+0.22+3d,
\quad\Pr[\mathcal H]\geq0.49.
\]
Set $c_{\mathbf e}:=x_{\mathbf e}(C)$ and
$c_{\mathbf f}:=x_{\mathbf f}(C)$. The degree-partition identities
give $x(A)\leq1+r+d-x(C)$,
$c_{\mathbf e}+c_{\mathbf f}\leq x(C)$, and
\[
x_{\mathbf e}(A)+x_{\mathbf f}(A)
\geq1-2h-2\omega-c_{\mathbf e}-c_{\mathbf f}.
\]
Thus $x(A-\mathbf e-\mathbf f)\leq2h+2\omega+r+d\leq4\omega+r+d$.
Endpoint-tree conditioning does not increase marginals on this
external set. Dividing its mean by the probability of $\mathcal H$
and using $x(U)\leq3/2+h+d$ gives
\[
\E[U_T+(A-\mathbf e)_T\mid\mathcal H]
\leq1.72+\omega+4d+\frac{4\omega+r+d}{0.49}\leq M_{\rm same}.
\]
Markov's inequality now proves the claim. The exact comparison at
$d_0$ is positive. No replacement of the residual-mass fraction by
$9h+3d$ is used.
\end{proof}

The oriented-window argument compares unconditional information with a law of
total rank two. We record how full-rank monotonicity transfers a mean bound and a
lower-tail bound to that law.

\begin{lemma}[Rank-two monotonicity bounds]
\label{lem:rank_two_monotonicity}
Let $X,Y,V$ count disjoint coordinate sets of a strongly Rayleigh law,
and put $R:=X+Y+V\geq1$. Suppose $q:=\Pr[R=2]>0$, and write
$p_1:=\Pr[R=1]$. If $\E[X]\leq U$ and
$\Pr[X+V\leq1]\geq\lambda$, then
\[
\E[X\mid R=2]\leq\frac U{1-p_1},\quad
\Pr[Y\geq1\mid R=2]\geq\frac{\lambda-p_1}{1-p_1}.
\]
\end{lemma}

\begin{proof}
Project to the union of the three counted sets. Consecutive full-rank
laws admit the monotone coupling in \cite[Theorem~2.10]{kko21}.
For every $j\geq2$ in the rank support, it gives
\[
\E[X\mid R=j]\geq\E[X\mid R=2],\quad
\Pr[X+V\leq1\mid R=j]\leq\Pr[X+V\leq1\mid R=2].
\]
Average the first inequality and use nonnegativity at rank one.
For the second, the event is automatic at rank one, so
\[
\lambda\leq p_1+(1-p_1)\Pr[X+V\leq1\mid R=2].
\]
At rank two, $X+V\leq1$ is equivalent to $Y\geq1$.
Since $q>0$, the denominator $1-p_1$ is positive.
\end{proof}

A lower bound on the probability of total rank at most two does not by itself
isolate the rank-two atom. Cutoff-dependent Newton updates provide the more
precise atom bounds needed for the window estimate.

\begin{lemma}[Cutoff-dependent central-rank bounds]
\label{lem:cutoff_central_rank}
Suppose $R\geq1$ is Poisson-binomial,
$\E[R]\in[2.47,3+2h+2r+8d_0]\subset[2.47,3.02]$, and
$\Pr[R\leq2]\geq L_*>0$, where $L_*<0.04$.
Put $C=1-2h-2r-8d_0$ and, for $Q\in\{L_*,0.04,0.1\}$, define
\[
c_{Q,0}=0.248,\quad
c_{Q,j+1}=C-2Q-3Q^2/D(c_{Q,j})\quad(0\leq j<4),\quad
D_Q=D(c_{Q,4}).
\]
At the stated parameters these four updates increase and lie in $(0,1)$.
If $q:=\Pr[R=2]\leq Q$, then $\Pr[R=1]\leq q^2/D_Q$.
Let $s_*$ be the smallest multiple of $10^{-30}$ strictly larger than
$\sqrt{1+4L_*/D_{L_*}}$, and put $m_*=2L_* /(1+s_*)$.
Then $q>m_*$.
\end{lemma}

\begin{proof}
The shifted law $R-1$ is Poisson-binomial. Its mean belongs to
$[1.47,2.02]$, so $p_3:=\Pr[R=3]>0.248$ by
Lemma~\ref{lem:bernoulli_estimates}. If $p_3\geq c_{Q,j}$, then
Lemma~\ref{lem:atom_newton} gives $p_1\leq q^2/D(c_{Q,j})$.
Also $\E[R]\geq4-3p_1-2q-p_3$, so
$p_3\geq C-3p_1-2q\geq c_{Q,j+1}$. This proves the four updates.
For the root bound, $0<m_*<L_*$ and
$m_*+m_*^2/D_{L_*}<L_*$. If $q>L_*$ the conclusion is immediate.
Otherwise use only the cutoff $Q=L_*$: $q\leq m_*$ would imply
$p_1+q\leq m_*+m_*^2/D_{L_*}<L_*$, a contradiction.
The stronger coefficient for this last cutoff is not used outside
$q\leq L_*$.
\end{proof}

We can now combine the selected-state tail and conditional rank-two bounds in the
oriented window argument. The resulting criterion turns a lower-tail condition
into a sufficiently likely bundle happiness event.

\begin{lemma}[Coupled oriented GKL window estimate]
\label{lem:candidate_gkl_window}
In the setting of \cite[Lemma~A.5]{gkl24}, use the parameters in
Definition~\ref{def:final24_parameters}. Let $A,B,C$ be the degree
partition at $u$, and suppose $x_{\mathbf e}(B)\leq\omega$ for the good
half bundle $\mathbf e=(u,v)$. Write
$V:=\delta(v)-\mathbf e$. If
\[
\Pr[(A-\mathbf e)_T+V_T\leq1]+x_{\mathbf e}(B)\geq Kh,
\]
then the bundle is $2$--$1$--$1$ happy with probability greater than $p$.
\end{lemma}

\begin{proof}
Condition on the endpoint trees, $C_T=0$, and $(u\cup v)_T$ being
a tree, and call the resulting law $\nu$. The full conditioning has
probability at least
\[
(1-d_0)(1/2-h-3d_0)-(2r+d_0)
\geq o_*:=1/2-h-2r-5d_0>0.49.
\]
The unique selected bundle edge is independent of the residual tree
law. Indeed, after the endpoint trees are contracted, the bundle
consists of parallel edges; conditioning their union to be a tree
selects one of them. The weighted-tree product factorization separates
this choice from the contracted residual tree. Deleting $C$ imposes
separate zero conditions on these factors and preserves independence.
The boundary-law closure in Lemma~\ref{lem:final23_electrical_remote}
extends the identity to the exact max-entropy law.

Set $X:=(A-\mathbf e)_T$, $Y:=(B-\mathbf e)_T$, and use $V_T$ for
the third count. Their union count is
$R:=X+Y+V_T=\delta(u\cup v)_T\geq1$.
Lemma~\ref{lem:candidate_two_tail} gives
\[
\Pr_\nu[R\leq2]\geq L_*>0.
\]
The unrounded conditional-marginal rows in
\cite[proofs of Lemmas~5.22 and~A.1]{kko21} give
\[
2.5-3h-3r-8d_0\leq\E_\nu[R]\leq3+2h+2r+8d_0,
\]
which is contained in $[2.47,3.02]$.
For the lower endpoint, the original residual mean is at least
$3-2h-2r-d_0$; endpoint-tree conditioning and bundle selection lose
at most $1-x(\mathbf e)+2d_0\leq1/2+h+2d_0$.
For the upper endpoint, the tree conditions and deletion cost at most
$2d_0$ and $2r+3d_0$, and selection cannot increase a disjoint count.
These estimates imply the displayed conservative interval.

Lemma~\ref{lem:cutoff_central_rank} gives $q:=\Pr_\nu[R=2]>m_*$.
Write $D_s=D_{0.04}$ and $D_m=D_{0.1}$ for its two interval products.
The same calculation will be used for the direct mixed event.

Let $\theta$ be the probability that the selected bundle edge lies
in $B$. The orientation and the outer probability imply
\[
0\leq\theta\leq\theta_B:=\omega/o_*.
\]
Put $b_{\mathbf e}:=x_{\mathbf e}(B)$ and $D:=1/2+h+3r$.
Before the union-tree condition, the endpoint-tree conditions lose
at most $2d_0$ from the $B$-bundle marginal, and the entire bundle
has marginal at most $D$. It is binary after the endpoint trees are
contracted. The parallel-edge factorization therefore gives selected
$B$ probability at least $(b_{\mathbf e}-2d_0)/D$ whenever the numerator
is positive. Deleting the bundle part of $C$ only renormalizes this
probability upwards, while deleting the residual part does not change it.
Thus in all cases
\[
b_{\mathbf e}\leq D\theta+2d_0.
\]
In the original counts, $A_T=X+1$ when this edge is in $A$, and
$B_T=Y+1$ when it is in $B$. The source conditional-marginal rows give
\[
\begin{aligned}
\E_\nu[X]&\leq1/2+h+2r+8d_0+\theta,\\
\E_\nu[Y]&\geq1/2-h-r-4d_0-\theta_B,\\
\E_\nu[V_T]&\geq1-3r-4d_0,\\
\E_\nu[X+Y],\E_\nu[V_T]&\leq3/2+h+2r+8d_0<1.512.
\end{aligned}
\]
The original individual counts also have means at most $1.512$.
The two sum means $\E_\nu[A_T+V_T],\E_\nu[B_T+V_T]$ are at least
$2-r-2d_0$. Lemma~\ref{lem:bernoulli_estimates}, the
zero-probability product bound, and \cite[Lemma~2.22]{kko21} yield
\[
\begin{aligned}
\Pr_\nu[Y\geq1]&>0.37,&\Pr_\nu[V_T\geq1]&>0.629,\\
\Pr_\nu[X+Y\leq1],\Pr_\nu[A_T\leq1]&>0.4284,&
\Pr_\nu[B_T+V_T\geq2]&>0.59.
\end{aligned}
\]
For example, the first bound follows from
$1-e^{-(1/2-h-r-4d_0-\theta_B)}>0.37$.
The assumed residual lower tail loses at most $2r+3d_0$ on imposing
the tree and $C$ conditions. Selecting the binary bundle coordinate
cannot decrease it, by negative association. Thus
\[
\Pr_\nu[X+V_T\leq1]\geq Kh-b_{\mathbf e}-2r-3d_0
\geq\lambda-D\theta,\quad \lambda:=Kh-2r-5d_0.
\]

Project to the residual union and condition its full rank to two;
denote this strongly Rayleigh law by $\nu_2$. The up/down inequality
\cite[Lemma~A.10]{gkl24}, applied to the two disjoint complementary
blocks each time, gives
\[
\Pr_{\nu_2}[Y\geq1]\geq0.37(\lambda-D\theta),\quad
\Pr_{\nu_2}[V_T\geq1]\geq c:=(0.629)(0.4284).
\]
Whenever $\Pr_{\nu_2}[X=0]>0$, delete the entire $X$ block.
Negative association increases both nonzero probabilities, and the
remaining counts satisfy $Y+V_T=2$. Write
\[
r_V:=\Pr_{\nu_2}[V_T=1],\quad
c_Y:=\Pr_{\nu_2}[X=0,Y=V_T=1].
\]
At rank two, the alternative $V_T=1$ pattern is $X=1,Y=0$.
Independence of the selected bundle edge therefore gives the exact
conditional happiness probability
\[
(1-\theta)c_Y+\theta(r_V-c_Y)
=\theta r_V+(1-2\theta)c_Y.
\]
Both outcomes are retained; in particular $1-2\theta>0$.

Put
\[
Q:=0.04,\quad Q_1:=0.1,\quad z_s(q):=q^2/D_s,\quad z_m(q):=q^2/D_m,\quad
A_0:=1/2-h-2r-8d_0.
\]
The two pooled means $\E_\nu[X+Y],\E_\nu[V_T]$ are at most $2-A_0$.
Lemma~\ref{lem:final24_rank_split} consequently gives
$r_V\geq\psi(A_0-z_s(q))$ for $q\leq Q$, and
$r_V\geq\psi(A_0-z_m(q))$ for $q\leq Q_1$.
The same source marginal rows bound both pooled means below by
$1-3r-4d_0>0.99$. For all $q$, their nonzero tails exceed $0.629$ and their
CDFs at one exceed $0.4284$. Put $a_V:=(0.629)(0.4284)$.
The up/down inequality gives both conditional tails of $V_T$ at least
$a_V$. To justify the nonzero bound on the full retained box, put
$x=0.99249992$. The hypotheses $r<s_0<0.0025$ and $d_0\leq2\cdot10^{-8}$
give $1-3r-4d_0>x$, and exact substitution gives
$1-(\sum_{j=0}^8x^j/j!)^{-1}>0.62935>0.629$.
Since $0.245<a_V$ and
$4(a_V-0.245)(1-a_V)-0.245^2>0$,
Lemma~\ref{lem:quadratic_bounds} gives $r_V>0.245$.
For $m_*\leq q\leq Q$, Lemma~\ref{lem:rank_two_monotonicity} and
Markov's inequality give
\[
\E_{\nu_2}[X]\leq1-d_\theta,\quad
d_\theta:=\frac{A_0-\theta-z_s(q)}{1-z_s(q)},\quad
\Pr_{\nu_2}[Y\geq1]\geq b_\theta:=\frac{\lambda-D\theta-z_s(q)}{1-z_s(q)}.
\]
The exact comparisons
\[
A_0-\theta_B-z_m(Q_1)>0,\quad\lambda-D\theta_B-z_s(Q)>0,
\]
\[
c>1/4,\quad0<\lambda<1/5,\quad A_0<2/3,\quad\theta_B<1/2,
\quad 0<m_*<Q<Q_1,\quad 1-z_m(Q_1)>0
\]
make every conditioning positive and give $0<b_\theta\leq\lambda<1/5$.
Also $1/3<1-d_\theta<1$.
Lemma~\ref{lem:mean_sensitive_split} applied directly to $\nu_2$ gives
$c_Y\geq d_\theta b_\theta$.
For fixed $q$, the resulting lower bound for happiness is
\[
H_q(\theta):=\theta\psi(A_0-z_s(q))
 +(1-2\theta)d_\theta b_\theta.
\]
Differentiating, and bounding each negative term separately, gives
\[
H_q'(\theta)\geq\psi(A_0-z_s(q))
-\frac{(2A_0+1)\lambda+A_0D}{(1-z_s(q))^2}
\geq\psi(A_0-z_s(Q))
-\frac{(2A_0+1)\lambda+A_0D}{(1-z_s(Q))^2}>0.
\]
The minimum is therefore at $\theta=0$. The rank-weighted factor is
\[
F(q):=\frac{q(A_0-z_s(q))(\lambda-z_s(q))}{(1-z_s(q))^2}.
\]
For $z=z_s(q)$, write $F(q)=qf(z)$, with
$f(z)=(A_0-z)(\lambda-z)/(1-z)^2$.
The two factors are positive and decreasing. Differentiation gives
\[
f'(z)\leq-\frac{(A_0-z)(1-\lambda)}{2(1-z)^3},\quad
f''(z)\leq\frac{2(1-A_0)(1-\lambda)}{(1-z)^4}.
\]
Since $F''(q)=z_s'(q)(3f'(z)+2zf''(z))$, the exact comparison
$3(A_0-z_s(Q))(1-z_s(Q))-8z_s(Q)(1-A_0)>0$
proves concavity throughout $[m_*,Q]$.
Restoring the outer condition gives
probability greater than
\[
W_0:=o_*\min\{F(m_*),F(Q)\}
>p.
\]
For $Q\leq q\leq Q_1$, define $d_\theta$ using $z_m$ and use
$b_\theta:=\eta(\lambda-D\theta)$, where $\eta:=0.37$.
The same mean-sensitive lemma gives the corresponding $H_q$, whose derivative is at least
\[
\psi(A_0-z_m(Q_1))
-\eta\frac{(2A_0+1)\lambda+A_0D}{(1-z_m(Q_1))^2}>0.
\]
Again minimize at $\theta=0$ and put $b:=\eta\lambda$.
The derivative of $q(A_0-z_m(q))/(1-z_m(q))$ has the sign of
$A_0-(3-A_0)z_m(q)+z_m(q)^2$, which is positive because $A_0-3z_m(Q_1)>0$.
The restored probability is consequently greater than
\[
W_1:=o_*Q\frac{A_0-z_m(Q)}{1-z_m(Q)}b
>p.
\]
Finally suppose $q\geq Q_1$. The original total count
$A_T+B_T+V_T=R+1$. Applying the up/down inequality to $A_T$ and
$B_T+V_T$ at full rank three gives
\[
\Pr_\nu[A_T\leq1\mid R=2]\geq\alpha:=(0.4284)(0.59).
\]
The independent bundle choice implies
\[
\Pr_\nu[A_T\leq1\mid R=2]
=(1-\theta)\Pr_{\nu_2}[X=0]
+\theta\Pr_{\nu_2}[X\leq1].
\]
Hence $\Pr_{\nu_2}[X=0]\geq\alpha-\theta>0$.
Using $r_V>0.245$ and the same deletion, the happiness probability exceeds
\[
0.245\theta+(1-2\theta)(\alpha-\theta)
 \phi(\eta(\lambda-D\theta)),\quad
\phi(v)=2\sqrt{1-v}-2(1-v).
\]
Here only the probability of $X=0$ has been bounded, not its mean;
we use the quadratic tail bound rather than
Lemma~\ref{lem:mean_sensitive_split}. Since $0\leq\phi(v)\leq v$
and $0\leq\phi'(v)\leq1$ on $[0,1/5]$, its derivative is at least
$0.245-\eta((1+2\alpha)\lambda+\alpha D)>0$.
Minimizing at $\theta=0$ and restoring the outer condition gives
\[
W_2:=o_*Q_1\alpha\phi(b)
>p.
\]
All three comparisons and derivative margins are exact rational
inequalities for every parameter row. The sole square root in $W_2$
is enclosed downward on the $10^{-30}$ grid. Every rank
condition is imposed after projection to the full counted union;
all other conditions are extremal or coordinate deletions.
\end{proof}

\subsection{Common events and the half budget}
\label{sec:common_events_half_budget}

The polygon acceptance construction requires both of two blocks to contribute one
edge even when their total mean exceeds one by only a small amount. We first give
a joint-atom estimate for this regime.

\begin{lemma}[Small-excess joint atom]\label{lem:flow-atom}
Let a homogeneous strongly Rayleigh law have disjoint counts $X,Y$, with
$X+Y\geq1$. Suppose $0<\alpha\leq A_0:=1/150$, $0\leq\varepsilon<\alpha/100$,
and
\[
\begin{gathered}
\E[X],\E[Y]\in[\alpha-\varepsilon,2+3\varepsilon-\alpha],\\
\E[X+Y]\in[1+\alpha,2+2\varepsilon].
\end{gathered}
\]
Then $\Pr[X=Y=1]\geq0.94\alpha^2$.
\end{lemma}
\begin{proof}
Write $a=\min\{\E[X],\E[Y]\}$, $b$ for the other mean, and $m=a+b$;
rename $X$ so that $\E[X]=a$. If $a<1$, then
$\Pr[X=1]\geq ae^{-a}>0$. Project away the entire $X$ block. If the original
fixed rank is $n$, conditioning $X=1$ is conditioning the full remaining rank
to $n-1$. Thus Lemma~\ref{lem:projected_rank_mean_shift} gives, with $p_0=\Pr[X=0]\leq e^{-a}$,
\begin{equation}\label{eq:flow-mean}
b-p_0\leq\E[Y\mid X=1]\leq m-1+p_0.
\end{equation}
This uses homogeneity and projection to the complement, not arbitrary
subset-rank conditioning while retaining the selected coordinates.

First suppose $a\leq1.3\alpha$. Then $a\geq0.99\alpha$ and $a<1$.
Using $e^{-a}\leq1-a+a^2/2$, the conditional lower mean is at least
$\alpha-a^2/2>0.963\alpha$. Its upper mean is at most $2-d_*\alpha$, where
\[
d_*:=0.97-0.99^2A_0/2>0,
\quad d_*\alpha<0.02,
\quad d_*(1-d_*A_0/2)>0.963.
\]
For conditional mean at most one, use the bound $\mu e^{-\mu}$;
for conditional mean at least one, use the near-two atom bound in Lemma~\ref{lem:bernoulli_estimates} with
$\delta=d_*\alpha$. Both give an atom at least
$0.963\alpha e^{-0.963\alpha}$. Hence the joint atom is at least
\[
0.99(0.963)\alpha^2e^{-(0.99+0.963)\alpha}
\geq0.99(0.963)e^{-(0.99+0.963)/150}\alpha^2>0.94\alpha^2.
\]

Next suppose $a\geq1.3\alpha$ and $a^2\leq\alpha/2$. Then $a<0.06$.
The conditional lower mean in \eqref{eq:flow-mean} is at least $0.75\alpha$.
Its deficit from two is at least
\[
a-a^2/2-2\varepsilon
\geq(0.97(1.3)-0.02)\alpha>\alpha.
\]
The same two scalar branches give conditional atom at least
$0.75\alpha e^{-0.75\alpha}$. The joint atom is at least
\[
1.3(0.75)e^{-(1.3+0.75)/150}\alpha^2>0.94\alpha^2.
\]

Finally suppose $a^2\geq\alpha/2$. Since $\alpha\leq1/150<1/128$, we have
$a>8\alpha$. All subset means are strictly between their neighboring target
integers for target $(1,1)$. The asymmetric capacity distances satisfy
\[
\delta_1=a,\quad\delta_2=m-1\geq\alpha,
\quad\epsilon_1=2-b\geq a-2\varepsilon,
\quad\epsilon_2=3-m\geq1-2\varepsilon.
\]
Each of $\delta_1\delta_2$, $\epsilon_1\delta_1$ and
$\epsilon_1\epsilon_2$ is at least $a\alpha$.
For the last product use
\[
a-2\varepsilon\geq(399/400)a,
\quad(399/400)(1-A_0/50)>A_0\geq\alpha.
\]
The capacity form of \cite[Theorem 4.1 and Corollary 5.13]{gkl24} gives
$\operatorname{Cap}_{(1,1)}\geq a\alpha$. Projection and $X+Y\geq1$ give zero constant term.
The compulsory bivariate coefficient bound
Lemma~\ref{lem:candidate_compulsory_capacity} gives joint atom at least
$(e^{-1}/2)a\alpha>4e^{-1}\alpha^2>0.94\alpha^2$.
The three closed cases cover all $a$; the first and third need not both apply,
and their separating comparison $1.3^2A_0<1/2$ is certified.
\end{proof}

Accepting selected boundary pairs must also preserve the side marginals
approximately. The next elementary bound converts coordinate domination into
control of the sum of absolute marginal changes.

\begin{lemma}[Accepted-side variation]\label{lem:variation}
If $x,v\geq0$, $0\leq q<1$, $\sum v_e=1$, $\sum x_e\in[1-\varepsilon,1+\varepsilon]$,
and $(1-q)v_e\leq x_e$, then $\sum_e|v_e-x_e|\leq2q+\varepsilon$.
\end{lemma}
\begin{proof}
Put $w=x-(1-q)v\geq0$. Then
$\sum|v-x|=\sum|qv-w|\leq q+\sum w=2q+\sum x-1\leq2q+\varepsilon$.
\end{proof}

We combine the joint-atom estimate with a flow construction to select acceptable
polygon boundary pairs. The construction controls the marginal changes on each
side separately.

\begin{lemma}[Marginal-preserving polygon event]\label{lem:candidate_polygon_flow}
Let a homogeneous strongly Rayleigh law have disjoint counts $A_T,B_T$,
means in $[1-\varepsilon,1+\varepsilon]$, and $A_T+B_T\geq1$.
Suppose $0\leq\varepsilon$ and $330\varepsilon<\zeta<1/200$. Put
\[
q=(\zeta-\varepsilon)/2,\quad b=q-\varepsilon,\quad c=q-5\varepsilon,
\quad v=\frac{27}{4}(0.94)c^2.
\]
There is a pair-dependent accepted event on which $A_T=B_T=1$, with probability
at least $v(1-q)$, both separate marginal budgets at most $\zeta$, and accepted
marginals at most $x_e/(1-q)$.
\end{lemma}
\begin{proof}
For $A'\subseteq A,B'\subseteq B$ of total mean at least $1+\alpha$, delete
their complement in $A\cup B$. The deletion probability is at least
$\alpha-2\varepsilon$. Homogeneous deletion domination gives individual means
in $[\alpha-\varepsilon,2+3\varepsilon-\alpha]$ and sum in
$[1+\alpha,2+2\varepsilon]$. The compulsory union survives.
Put $A_0:=1/150$. For $100\varepsilon<\alpha\leq A_0$, Lemma~\ref{lem:flow-atom} gives the
unconditioned target probability at least
\[
0.94(\alpha-2\varepsilon)\alpha^2.
\]
For $\alpha\geq A_0$, use that lemma at $A_0$, but retain the actual deletion
probability. The lower bound is then
$0.94(\alpha-2\varepsilon)A_0^2$.

Use the bipartite network of \cite[Proposition 5.6]{kko21}. Its middle
capacities are $y_{ef}=\Pr[e,f\in T\mid A_T=B_T=1]$, and its outer capacities
are $\beta x_e,\beta x_f$, where $\beta=v/p_{11}$.
The full-side case above proves $p_{11}>0$. For a network cut put
$\gamma=x(S_A)-x(S_B)$. Exterior arcs contribute at least
$\beta(1-\varepsilon-\gamma)$. If $\gamma\leq b$, this suffices.
Otherwise set $\alpha=\gamma-\varepsilon$ and $x=\gamma-b>0$.
The domain gives $q>102\varepsilon$, so $\alpha>100\varepsilon$.
For $\alpha\leq A_0$, the middle crossing probability is at least
\[
0.94(\alpha-3\varepsilon)^3
=0.94(x+c)^3\geq v x,
\]
using
$(x+c)^3-(27/4)c^2x=(x-c/2)^2(x+4c)\geq0$.

For $\alpha\geq A_0$, note $\varepsilon<1/66000$, $c<q<1/400$ and
$x=\alpha-q+2\varepsilon\leq\alpha+2\varepsilon$.
The ratio $(\alpha-2\varepsilon)/(\alpha+2\varepsilon)$ increases with $\alpha$,
and the exact inequality
\[
(A_0-2/66000)A_0^2>
\frac{27}{4}(1/400)^2(A_0+2/66000)
\]
therefore gives $0.94(\alpha-2\varepsilon)A_0^2\geq vx$.
This retains the actual deletion factor in the large-cut case.
Every source--sink cut has capacity at least $\beta(1-q)$.

Take a flow of value at least this number. Given a selected pair, accept with
probability its flow divided by $y_{ef}$; zero-capacity pairs never occur.
Flow conservation gives accepted side mass one and coordinate domination
$x_e/(1-q)$. Lemma~\ref{lem:variation}, separately on each side, gives
$2q+\varepsilon=\zeta$. Acceptance uses only the pair and one independent coin.
Independent thinning preserves its conditional law.
\end{proof}

The preceding tools now yield the probability bounds needed for the reduction
events. We collect these bounds for all bundle types while keeping the top and
polygon event rates separate.

\begin{lemma}[Common-event probabilities]
\label{lem:candidate_common_probability}
\label{lem:final24_probability}
At the parameters in Definition~\ref{def:final24_parameters}, the polygon
event has probability at least $P$, and the small-, large-, mixed-,
two-sided mixed-, and paired-bundle happiness events have probability at
least $p$. Each event can be thinned to its designated probability without
changing its conditional law.
\end{lemma}

\begin{proof}
For each polygon cut $S$, use the construction in
\cite[Definition~5.8]{kko21}, replacing its max-flow event by
Lemma~\ref{lem:candidate_polygon_flow} with
$\zeta=\epsilon_M$ and $\epsilon=2d$.
The unchanged polygon conditioning gives a homogeneous strongly Rayleigh law
with $\E[A_T],\E[B_T]\in[1-2d,1+2d]$, and $A_T+B_T=\delta(S)_T\geq1$. The cut is proper even for a
hierarchy-root polygon, so every premise of the lemma holds. Repeating the proof
of \cite[Corollary~5.9]{kko21}, rather than invoking its specialization
$\epsilon_M=1/4000$, multiplies this probability by the probability that
$C_T=0$ and the polygon cut is a tree. The latter probability is at least
$1-3d/2$. Define
\[
q_M:=(\epsilon_M-2d_0)/2.
\]
The unconditional probability is at least
\[
\frac{27}{4}(0.94)(q_M-10d_0)^2(1-q_M)(1-3d_0/2)
>P.
\]
The new flow lemma requires
$330(2d_0)<\epsilon_M<0.005$.  Both inequalities hold, and we also have
$d_0<\epsilon_M^2$. The enlarged polygon-payment domain is
proved directly in Lemma~\ref{lem:polygon_charge_extended}.
The lower bound decreases with $d$: the function
$(q-5\epsilon)^2(1-q)$ increases with $q$ and decreases with $\epsilon$
on $0\leq5\epsilon<q<0.0025$, while $q$ decreases as $d$ increases.
The two separate boundary marginal errors remain at most $\epsilon_M$.
Acceptance depends only on the boundary pair and an independent coin, so
the internal/external factorization in \cite[Corollary~5.9]{kko21}
is preserved. Thin independently to probability $P$ and call this event
$E_S$. Its conditional marginals and the conditional Bernoulli-sum
conclusion of that corollary are unchanged by the thinning.

{\bf Small non-half bundles.}
Use the conditioning of \cite[proof of Lemma~5.21]{kko21}.
The endpoint trees remain trees, and both the bundle and degree remainder are
deleted. Its probability is at least $o_s:=1/2+h-2r-3d_0>0$.
Put $X=A_T$, $Y=B_T$, $W=V_T-1$, and $R=X+Y+W$.
The proper-cut constraints give $V_T\geq1$, $X+Y\geq1$, and hence $R\geq1$.
Divide the joint stable polynomial by the compulsory $V$ variable and polarize.
The target event is now $X=Y=W=1$.

Write $q_{\bf e}:=x(\mathbf e)\leq1/2-h$. Before the conditioning,
$x(A),x(B)\in[1-r,1+d]$, $x(\delta(v))\in[2,2+d]$, and
$x_{\mathbf e}(A)+x_{\mathbf e}(B)\leq q_{\bf e}$.
The conditional-marginal calculation in the cited proof decreases any
displayed union by at most $2d$ and increases it by at most
$q_{\bf e}+2r+3d$. For example,
\[
\E[B_T+V_T]\geq x(B)+x(\delta(v))-q_{\bf e}
-x_{\mathbf e}(B)-2d\geq2+2h-r-2d.
\]
Substitution gives the complete shifted table for $0\leq d\leq d_0$:
\[
\begin{array}{c|c|c}
\text{union}&\text{lower mean}&\text{upper mean}\\ \hline
X,Y&1/2+h-r-2d&3/2-h+2r+4d\\
W&1/2+h-3d&1+2r+4d\\
X+Y&3/2+h-2r-2d&2+4r+6d\\
X+W,Y+W&1+2h-r-2d&2+2r+5d\\
R&2+2h-2r-2d&3+2r+6d.
\end{array}
\]
The double entries denote both distinct subsets, not omitted constraints.
At the final parameters every coordinate mean is at most $1.5$, every pair
mean is at most $2.01$, and all individual nonzero probabilities exceed $0.39$.
Put
\[
\delta_2=2h-r-2d_0,\quad \delta_3=2h-2r-2d_0,\quad
L_0=1/2+h-r-3d_0,\quad q_s=\delta_3c(\delta_3),\quad Q_s=0.04.
\]
Every pair mean is at least $1+\delta_2$ and
$\E[R]\in[2+\delta_3,3.01]$.
The shifted Poisson-binomial law $R-1$ and Lemma~2.21 of \cite{kko21} give
$q:=\Pr[R=3]\geq q_s>0$.  For shifted means above $1.2$, the retained
$0.1637$ target-two bound is stronger than $q_s$; below $1.2$ the correction
exponent is zero and the one-forced branch gives $\delta_3e^{-\delta_3}$.

{\bf Small rank-three probability.}
Suppose $q\leq Q_s<0.1637$. Then $\E[R-1]<1.2$ and
$\Pr[R=2]>0.36$. Apply Lemma~\ref{lem:tail_profiles} to $R-1$,
using its mean lower bound $1+\delta_3$. Put
\[
E(q)=E_{\delta_3,0}(q),\quad e=\delta_2,\quad
\ell=L_0-E(q),\quad \delta=e-E(q)>0.
\]
Here $\delta_3>0$ is fixed within the row. The checks give
$2q_s>\delta_3$; at $Q_s$, both updated atoms exceed $1/2$,
both used $J$ coefficients exceed $3/2$, and every factorial ratio
parameter is below four. These are the endpoint conditions in
Lemma~\ref{lem:tail_profiles}, so $E(q)$ is valid, increasing, and
convex on $[q_s,Q_s]$. In particular, it bounds $\E[(R-3)_+]$.
Lemma~\ref{lem:projected_rank_mean_shift} gives individual means at
least $\ell$ and pair means at least $1+\delta$ under the full
rank-three law. Choose a coordinate of largest mean, and call the
other two $Y,Z$. Both means are at most $1.5$, and their sum is at most
two. Their total is Poisson-binomial of degree at most three. Pad its
Bernoulli representation to three variables; the complementary sum has
mean in $[1,2-\delta]$. Equation~\eqref{eq:near_two_atom} therefore gives
\[
q':=\Pr[Y+Z=2]\geq v:=\delta(1-\delta/2).
\]
This uses complementing the three Bernoulli parameters, not independent
graph edges.

If $q'\geq0.1$, the individual nonzero tails $0.38$ and CDF tails
$7/16$ give quadratic split greater than $0.15$. Hence
$\Pr[Y=Z=1]>0.015$. If $q'<0.1$, the target-two exclusion forces the
pair mean below $1.2$, so its rank-one atom exceeds $0.36$.
Use $T_3$ and the exact coupling $1/3$ in
Lemma~\ref{lem:normalized_rank_tail}, writing the resulting function
with individual lower mean $\ell$ as $F_\ell$.
At the worst outer endpoint $Q_s$, rational checks give
\[
\ell-0.2-T_3(0.1)>0,\quad
1-T_3(0.1)/3>0,\quad F_\ell(0.1)>0.015,
\]
and $1-e^{-\ell}>0.38$. They remain valid at smaller outer $q$.
The numerator at the lower inner endpoint is positive because
$r+4d_0-e^2/2>0$. Its monotonicity then covers the whole inner
interval. Concavity gives the lower bound
$\min\{F_\ell(v),F_\ell(0.1)\}$.

Only two outer endpoint comparisons are needed. Define fixed row constants
\[
\bar v=e(1-e/2),\quad \bar t\geq T_3(\bar v),\quad
\varphi=1-\bar t/3,\quad
A=r+4d_0+5\bar t/3-e^2/2,\quad D=e/2-e^2/4,
\]
where $\bar t$ is the upward rational enclosure from
Lemma~\ref{lem:tail_profiles}. The checks give
$\varphi>0.979$, $A>0$, and $D<0.015$.
For $0<\delta\leq e<0.02$, put
$g(\delta)=(5/3)T_3(\delta(1-\delta/2))-\delta^2/2$.
Since $T_3'(q)\geq3q/2$,
\[
g'(\delta)\geq\delta((5/2)(1-\delta/2)(1-\delta)-1)>0.
\]
The inner numerator and denominator at $v$ therefore satisfy
\[
0<W(v)=r+4d_0+E(q)+g(\delta)\leq A+E(q),\quad
\Phi(v)\geq\varphi.
\]
Also $v/2=D-(1-e)E(q)/2-E(q)^2/4$. Thus the full inner split, including
the easy branch and the other endpoint, is at least $G(q)/q$, where
\[
G(q)=q(D-\frac{1-e}{2}E(q)-\frac{E(q)^2}{4}
-\frac{(A+E(q))^2}{8\varphi}).
\]
Indeed the bracket is at most $D<0.015$.
To check the whole outer interval, denote the bracket by $f(E(q))$.
On $E\geq0$, $f'\leq0$ and $f''<0$, while $E',E''\geq0$ by
Lemma~\ref{lem:tail_profiles}. Consequently
\[
G''=2f'(E)E'+q(f''(E)(E')^2+f'(E)E'')\leq0.
\]
It follows that
\[
\Pr[\text{small-bundle happiness}]
\geq o_s\min\{G(q_s),G(Q_s)\}>p.
\]
The two final evaluations use upward enclosures for $E$ inside the
decreasing function $f$. They certify the endpoints of the unrounded
concave function, not a rounded staircase.

{\bf Large rank-three probability.}
Suppose $q\geq Q_s$.  The up/down inequality, applied to each coordinate and
its complement at rank three, gives for each coordinate $D$
\[
\Pr[D\geq1\mid R=3]\geq\alpha:=0.39(0.66),\quad
\Pr[D\leq1\mid R=3]\geq\beta:=(7/16)\delta_2c(\delta_2).
\]
The first uses the CDF at two, and the second uses the pair upper tail
from Lemma~\ref{lem:bernoulli_estimates}. Here $0<\beta<0.01$.
Put $c_0:=0.4911$ and $b_*:=\beta/(1+\beta/c_0)$.
Lemma~\ref{lem:cubic_two_tail}, under this full rank-three conditioning,
gives the unconditional lower bound
\[
o_sQ_sb_*\psi(\alpha-b_*^2/c_0)
>p.
\]
These two ranges exhaust $q$.  The small-event support used here is only
$R\geq1$, not the unproved stronger condition $R\geq2$.

{\bf Large non-half bundles.}
Use the endpoint and union tree conditions and remainder deletion of
\cite[proof of Lemma~5.22]{kko21}.
The conditioning probability is at least
$o_l=1/2+h-2r-3d_0>0$.
For the unshifted counts $X=A_T,Y=B_T,Z=V_T$, the support constraints are
$X+Y\geq1$ and $R:=X+Y+Z\geq2$.
Exactly one bundle edge is selected, and
$V=\delta(v)\setminus\delta(u)$. Thus
\[
A_T+B_T+V_T=1+\delta(u\cup v)_T\geq2.
\]
The union is proper even if it equals the hierarchy root, since it omits
the distinguished vertices. The unrounded conditional-marginal estimates
of the cited proof give the complete table
\[
\begin{array}{c|c|c}
\text{union}&\text{lower mean}&\text{upper mean}\\ \hline
X,Y&1/2+h-r-2d&3/2-h+r+3d\\
Z&1-3d&3/2-h+2r+4d\\
X+Y&2-3r-3d&5/2-h+3d\\
X+Z,Y+Z&2-3r-4d&3-(2h-r-6d)\\
R&3-3r-5d&4-(2h-6d),
\end{array}
\]
Put $\Delta_p:=2h-r-6d_0$, $\Delta_t:=2h-6d_0$, and
$U_0=3/2-h+2r+4d_0$.
To justify the upper rows, under only endpoint trees put
$q=y(e)$, $a_A=y(e(A))$, $c_e=y(e(C))$, and $D=q-c_e$.
Then $q\geq1/2+h-2d$, $D>0$, $0\leq a_A\leq D\leq1$, and
$y(A)+y(C)\leq1+r+d$. Selecting the bundle and deleting its $C$ part makes its
independent $A$ probability $a_A/D$. Residual $C$ deletion costs at most
$y(C)-c_e$. Using $a_A(1/D-1)\leq1-D$ gives respectively
$2+r+d-q$, $1+y(\delta(u))-q$,
$4+r+2d-2q$, and $1+y(\delta(u))+y(\delta(v))-2q$
for the relevant upper means. Substitution proves the upper rows. The lower rows follow from the
original endpoint-tree and deletion estimates.

These lower endpoints imply $U_0<1.5$, pair means at most $3-\Delta_p$,
$\E[X+Y]\leq2.5$, and pair upper tails at two greater than $0.59$.
For the latter, thin the Bernoulli parameters to mean $2-3r-4d_0\in(1,2)$
and use the two branches of \cite[Lemma~2.22]{kko21}.

The shifted rank $R-2$ has mean in $[0.99,2-\Delta_t]$.
Lemma~\ref{lem:bernoulli_estimates} gives
\[
q:=\Pr[R=3]\geq q_l:=\Delta_t(1-\Delta_t/2)>0.
\]
If $q\leq0.12$, its mean cannot lie in $[0.99,1.49]$, where the rank-one atom
exceeds $0.33$.  Hence $\E[R-2]>1.49$ and Lemma~\ref{lem:bernoulli_estimates} gives
$\Pr[R=4]>1/4$.  Newton yields $\Pr[R=2]\leq2q^2$.
Lemma~\ref{lem:projected_rank_mean_shift} at rank three therefore gives
\[
\E[D\mid R=3]\leq U_0+2q^2,
\quad \E[D\mid R=3]\geq\Delta_p-2q^2
\]
for every coordinate $D$, using its complementary pair for the lower bound.

{\bf The range $q_l\leq q\leq Q_l=0.025$.}
The conditional means lie below $1.504$ and above
$\ell(q)=\Delta_p-2q^2>0$, with $\ell(q)<0.01$.
Lemma~\ref{lem:cubic_split} gives full probability at least
\[
o_l(0.3728)q\ell(q)e^{-\ell(q)}.
\]
The derivative of $q\ell(q)e^{-\ell(q)}$ has the sign of
\[
\Delta_p+(4\Delta_p-6)q^2-8q^4.
\]
This expression strictly decreases as a function of $q^2$, so the original
function has no interior minimum.  Checking its two endpoints gives
probability greater than $p$.

{\bf The range $Q_l\leq q\leq0.12$.}
Every rank-three coordinate mean is below $1.53$.
The original nonzero tails for $X,Y$ exceed $0.39$; their complementary pairs
have CDF at two at least $\Delta_p(1-\Delta_p/3)$ by Lemma~\ref{lem:bernoulli_estimates}.
The original $Z$ nonzero tail exceeds $0.63$, and $\Pr[X+Y\leq2]>0.42$.
Thus every conditional nonzero probability, and hence every conditional mean,
is at least
\[
\alpha_l:=0.39\Delta_p(1-\Delta_p/3)\in(0,0.01).
\]
Lemma~\ref{lem:cubic_split} gives full probability at least
\[
o_lQ_l(0.3594)\alpha_l e^{-\alpha_l}
>p.
\]

{\bf The range $q\geq0.12$.}
The support $X+Y\geq1$ gives $Z\leq2$ under $R=3$.
Each original CDF at one is at least $7/16$, and each complementary pair upper
tail at two exceeds $0.59$.  Hence the conditional CDF of $Z$ at one is at least
$\beta_l=(7/16)(0.59)$.  Its conditional nonzero tail exceeds
$0.63(0.42)>\beta_l$.  The stable quadratic discriminant therefore gives
$\Pr[Z=1\mid R=3]>0.238$.
The rank-three nonzero and CDF tails of $X,Y$ are at least
$\alpha_l,\beta_l$.  Project to $X,Y$ and condition the full rank to two;
the conditioning is positive by the preceding bound.  Its two nonzero tails
are at least $\alpha_l\beta_l$, so Lemma~\ref{lem:quadratic_bounds} gives split greater than
$\alpha_l\beta_l(1-\alpha_l\beta_l)$.
The full probability exceeds
\[
o_l(0.12)(0.238)\alpha_l\beta_l(1-\alpha_l\beta_l)
>p.
\]
The three closed ranges cover every possible $q$.  All projections, support
exclusions, and positive conditioning margins have been specified.

For the mixed-bundle event,
Lemma~\ref{lem:candidate_same_side_tail} supplies the premise of
Lemma~\ref{lem:candidate_gkl_window}: after orienting $\mathbf e$ from $v$
to $u$, its residual set $V$ is the set $U$ used there, and its probability exceeds $Kh$.

For the two-sided mixed event in \cite[Lemmas~A.7 and~A.12]{gkl24},
use their law $\nu$ conditioned on $C_T=0$ and $u,v,u\cup v$ being trees.
Its probability is at least
\[
P_0=(1-d_0)(1/2-h-3d_0)-(2r+d_0)>0.49.
\]
Let $X:=(A-\mathbf e)\mathbin{\dot\cup}(B-\mathbf e)$ and
$Y:=\delta(v)-\mathbf e$, which are disjoint.
The same conditional-marginal calculation as in the window proof gives
\[
\E_\nu[X_T],\E_\nu[Y_T]\in[1-3r,3/2+h+2r+8d_0]
\subset[0.99,1.512],
\]
\[
\E_\nu[X_T+Y_T]\in[5/2-3r-3h-8d_0,3+2h+2r+8d_0]
\subset[2.47,3.02].
\]
For the individual upper bounds, the original residual cut has mean at
most $3/2+h+d_0$; the endpoint trees and deletion contribute at most
$2d_0$ and $2r+3d_0$, and selection cannot increase a disjoint residual
count. The lower bounds and total means follow as in the window proof.
The nonzero-tail certificate there applies to the smaller mean
$1-3r-4d_0$, and hence gives
\[
\Pr_\nu[X_T\geq1],\Pr_\nu[Y_T\geq1]>0.629,\quad
\Pr_\nu[X_T\leq1],\Pr_\nu[Y_T\leq1]>0.4284.
\]
Project to the full union and condition its rank to two. The up/down
inequality and the quadratic test in the window proof imply
\[
\Pr_\nu[X_T=Y_T=1\mid X_T+Y_T=2]>0.245.
\]
The proper-cut identity gives $X_T+Y_T\geq1$, and
Lemma~\ref{lem:candidate_two_tail} supplies lower tail $L_*$.
Thus Lemma~\ref{lem:cutoff_central_rank} gives
$q:=\Pr_\nu[X_T+Y_T=2]>m_*$, with products $D_s,D_m$.
Put $A_D=1/2-h-2r-8d_0$ and
\[
G_D(q)=q\psi(A_D-q^2/D).
\]
Lemma~\ref{lem:final24_rank_split} bounds the joint central probability
below by $G_{D_s}(q)$ on $[m_*,0.04]$ and by $G_{D_m}(q)$ on
$[0.04,0.1]$. For either corresponding endpoint $Q$, write $Z=Q^2/D$.
The derivative satisfies
\[
G_D'(q)\geq\psi(A_D-Z)-2Z(1-A_D+Z)>0.
\]
Indeed the derivative expression decreases with $z=q^2/D$ on these
positive-argument intervals. The exact comparisons
\[
G_{D_m}(0.04)>G_{D_s}(m_*),\quad
0.0245>G_{D_s}(m_*)
\]
cover the middle and large ranges, respectively. Therefore the joint
central probability always exceeds $G_{D_s}(m_*)$.

The unique parallel-edge choice is independent of the residual tree law.
When both degree parts have bundle mass at least $\omega$, its required
endpoint probability is at least
\[
P_e=\frac{\omega-2d_0}{1/2+h+3r}>0.
\]
Keeping both factors, the direct-mixed event has probability greater than
\[
P_0P_e m_*\psi(A_D-m_*^2/D_s)>p.
\]
The conditioning lower bound decreases with the error, since the derivative
of $(1-d)(1/2-h-3d)-(2r+d)$ is $-4.5+h+6d<0$.
The same is true of $P_e$. Thus these bounds hold for every $d\leq d_0$.
The oriented window estimate handles the complementary crossing cases.

The sixth event is the paired event in
Lemma~\ref{lem:candidate_reparam_527}, which uses
Lemma~\ref{lem:candidate_gkl_window} with the fixed $K$.
Thus the top-bundle events have probability at least $p$, and the polygon
event has probability at least $P$. Independent thinning gives their
designated probabilities while preserving each conditional law.
The remaining numerical premise in \cite[Theorem~5.28]{kko21} holds as well:
$p>0.005h^2$.
\end{proof}

We next study descendant cuts under the accepted polygon event. Keeping separate
budgets for the two boundary sides gives the parity and unhappiness estimates
needed for their payments.

\begin{lemma}[Half-budget and small-remainder polygon bounds]
\label{lem:polygon_half_budget}
Let $E_{\mathcal P}$ be the accepted event of a polygon ancestor $\mathcal P$, including
its independent thinning to probability $P$. For every proper descendant
$u$ in the hierarchy,
\[
\Pr[\delta(u)_T\text{ odd}\mid E_{\mathcal P}]<q_0.
\]
If $u$ is a polygon cut, its conditional left and right unhappiness
probabilities are at most $q_{\rm unh}$.
\end{lemma}

\begin{proof}
Write $x'$ for the marginals after making the parent a tree and deleting
its remainder, before accepting the boundary pair. On either accepted
boundary side, the accepted marginal vector $v$ satisfies
\[
\sum_e v_e=1,\quad \sum_e|v_e-x'_e|\leq\epsilon_M,
\quad |\sum_e x'_e-1|\leq2d.
\]
Positive and negative variation therefore each have mass at most
$\epsilon_M/2+d$. For every subset $F$ of one side,
\[
|v(F)-x'(F)|\leq\epsilon_M/2+d.
\]
The internal parent tree and the contracted tree are independent.
Making the parent a tree changes an internal mean by an amount in
$[0,d/2]$. On surviving exterior coordinates, the tree and remainder
conditions change an edge-set mean by an amount in $[-d/2,d]$;
a deleted exterior part has original mass at most $d$.
Consequently every count contained in the parent and its boundary,
and meeting at most one surviving boundary side, changes in mean by
at most $\epsilon_M/2+3d$ under $E_{\mathcal P}$.

The internal count is Poisson-binomial and independent of the exterior
count, which is zero or one. Their sum is therefore Poisson-binomial;
this does not assert that the entire accepted pair law is strongly Rayleigh.
By \cite[Observation~4.32]{kko21}, the exterior part of every descendant
cut is contained in one parent side or in its deleted remainder.
The cut count has a compulsory success. After removing it, the mean lies
between $1-\epsilon_M/2-3d_0$ and $1+\epsilon_M/2+4d_0<1.2$.
The parity bound in \cite[Corollary~2.17]{kko21}, followed by the
Taylor bound defining $q_0$, proves the first assertion.

Now let $A,B,C$ be the descendant polygon partition. There is stronger
control on the small remainder $F:=C$, whose original mass is at most $d$.
Flow conservation in Lemma~\ref{lem:candidate_polygon_flow} gives
$v_e\leq x'_e/(1-q(d))$, where
\[
q(d):=(\epsilon_M-2d)/2\leq q(0)=\bar q.
\]
The internal part gains at most $d/2$, and the surviving exterior part
gains at most $d$ before acceptance. Hence
\[
\begin{aligned}
\E[F_T\mid E_{\mathcal P}]
&\leq x(F\cap E(\mathcal P))+d/2
  +\frac{x(F\cap\delta(\mathcal P))+d}{1-\bar q}\\
&\leq d/2+\frac{2d}{1-\bar q}
\leq\frac{3d}{1-\bar q}.
\end{aligned}
\]
Deleted exterior coordinates contribute zero. The uniform denominator
uses $q(0)=\bar q$, not $q(d_0)$.
The side count $A_T\mid E_{\mathcal P}$ is also Poisson-binomial, with mean at
least $1-\epsilon_M/2-4d_0$ and at most $1+\epsilon_M/2+4d_0<1.2$.
For $f(x):=(1+e^{-2x})/2$, $|f'(x)|<1$ for $x>0$.
Its even probability is thus at most $q_0+d_0$.
The union bound and Markov's inequality give
\[
\Pr[u\text{ not left happy}\mid E_{\mathcal P}]
\leq q_0+d_0+\frac{3d_0}{1-\bar q}
\leq q_0+\frac{4d_0}{1-\bar q}=q_{\rm unh}.
\]
The right side is symmetric. Independent thinning preserves every
conditional distribution used in the proof.
\end{proof}

\subsection{Paired bundles and localized matching}
\label{sec:paired_matching}

The paired-bundle analysis compares several conditioned laws on overlapping
counted sets. We begin with a mean-shift bound that uses monotonicity and bounds
on a larger containing set.

\begin{lemma}[Monotone marginal-shift bound]
\label{lem:candidate_marginal_shift}
Let $\nu_-$ be a strongly Rayleigh law on $E$, let $F\subseteq E$, and let
$\nu_+$ be the law obtained by a positive-probability conditioning on $F_T=0$.
Let $D\subseteq L\subseteq E\setminus F$, and suppose that for some $j$
and nonnegative $\ell_-,\ell_+,u_-,u_+$,
\[
\E_{\nu_-}[L_T]\in[j-\ell_-,j+u_-],\quad
\E_{\nu_+}[L_T]\in[j-\ell_+,j+u_+].
\]
Then
\[
0\leq\E_{\nu_+}[D_T]-\E_{\nu_-}[D_T]\leq u_++\ell_-.
\]
If instead $F_T\leq1$ almost surely under $\nu_-$ and $\nu_+$ is the
law obtained by a positive-probability conditioning on $F_T=1$, then
\[
0\leq\E_{\nu_-}[D_T]-\E_{\nu_+}[D_T]\leq u_-+\ell_+.
\]
For the zero conditioning, the upper bound $u_++\ell_-$ remains valid when
$D\subseteq L\subseteq E$ and $F\cap L\subseteq D$; in this extension no
lower sign is asserted for the change of $D_T$.
\end{lemma}

\begin{proof}
By \cite[Fact~2.11]{kko21}, every coordinate outside $F$ has nonnegative
marginal change under deletion. Since $L\setminus D$ is disjoint from $F$,
including in the overlap extension,
\[
\E_{\nu_+}[D_T]-\E_{\nu_-}[D_T]
\leq\E_{\nu_+}[L_T]-\E_{\nu_-}[L_T]
\leq u_++\ell_-.
\]
The change on $D$ is nonnegative when $D$ is disjoint from $F$.
The upward-conditioning part of the same fact reverses the signs and gives
the selection bound by the identical subtraction.
Both conditionings preserve stability. Deletion is specialization at zero.
For the binary selection, write the generating polynomial as $p=p_0+p_1$
according to its degree in the $F$ variables. Positive rescaling and the limit
\[
t^{-1}p(tz_F,z_{E\setminus F})\to p_1(z_F,z_{E\setminus F})
\]
as $t\to\infty$ give stability of the nonzero conditioned polynomial.
No arbitrary subset-rank conditioning is used.
\end{proof}

We will often know the expected excess above and deficit below an integer rather
than a precise conditional mean. The following elementary estimate converts that
information into conditional mean bounds.

\begin{lemma}[Conditioning excess moments]
\label{lem:conditioning_excess}
Let $X$ be integer-valued, let $j$ be an integer, and let $l,u\geq0$. Suppose
$\E[(j-X)_+]\leq l$ and $\E[(X-j)_+]\leq u$.
For any event $A$ with $\Pr[A]\geq P>0$,
\[
j-l/P\leq\E[X\mid A]\leq j+u/P.
\]
The assertion does not require independence or a central-atom bound.
\end{lemma}
\begin{proof}
Intersect each nonnegative excess variable with $A$, bound its expectation
by its unconditional expectation, and divide by $\Pr[A]$. Subtract the
resulting lower and upper excess bounds.
\end{proof}

When the conditioning events are nested, bounding their means separately discards
useful cancellation. We retain the nesting to obtain a sharper comparison between
the two conditional means.

\begin{lemma}[Nested conditioning of excess moments]
\label{lem:nested_excess}
Let $A\subseteq H$ be positive-probability events, with actual probabilities
$a=\Pr[A]\leq b=\Pr[H]$. Suppose
\[
\E[(j-X)_+]\leq l, \E[(X-j)_+]\leq u,\quad 0\leq l\leq u.
\]
Then
\begin{align*}
\E[X\mid A]-\E[X\mid H]&\leq u/a+(l-u)/b,\\
\E[X\mid H]-\E[X\mid A]&\leq l/a+(u-l)/b.
\end{align*}
If $a\geq A_*>0$ and $b\geq H_*>0$, the right sides may be replaced by
$u/A_*-(u-l)$ and $l/A_*+(u-l)/H_*$, respectively.
\end{lemma}
\begin{proof}
Writing $Z=X-j$ gives
\[
\E[X\mid A]-\E[X\mid H]
=(1/a-1/b)\E[Z1_A]-\E[Z1_{H\setminus A}]/b.
\]
The two coefficients are nonnegative. Bound the first positive part by $u$ and
the second negative part by $l$ to obtain the first inequality. Reverse their
roles for the second inequality. For the relaxed first bound use $b\leq1$ and
$l-u\leq0$; for the second use $b\geq H_*$ and $u-l\geq0$.
\end{proof}

The paired-bundle proof separates into two cases according to the sizes of
opposite excess moments. A binary covering coupling ensures that at least one of
these moments is sufficiently large.

\begin{lemma}[Asymmetric binary separator]
\label{lem:asymmetric_separator}
Let a homogeneous strongly Rayleigh law have a binary counted block $E$ and a
disjoint count $D$. Put $q=\Pr[E=1]$, $S=D+2E$. If
\[
\E[(k-S)_+]\leq L, \E[(S-k)_+]\leq U,
\]
then
\[
qL+(1-q)U\geq q(1-q).
\]
\end{lemma}
\begin{proof}
For $q=0,1$ the conclusion is immediate. Otherwise collapse the binary block
to one coordinate and project it away. Homogeneity makes the conditional
complement ranks consecutive. The full-rank coupling in
\cite[Theorem 2.10]{kko21} gives $D_0-D_1\in\{0,1\}$.
Thus $S_0=D_0$, $S_1=D_1+2$ satisfy $S_1-S_0\geq1$. Pointwise,
\[
(k-S_0)_++(S_1-k)_+\geq S_1-S_0\geq1.
\]
Their expectations are at most $L/(1-q)$ and $U/q$. Multiply by $q(1-q)$.
\end{proof}

When neither bundle has a sufficiently likely individual $2$--$1$--$1$ event, the
paired configuration provides an alternative. We estimate the probability of
making both bundles happy simultaneously.

\begin{lemma}[Paired-bundle estimate]
\label{lem:candidate_reparam_527}
In the setting of \cite[Lemma~5.27]{kko21}, use the parameters in
Definition~\ref{def:final24_parameters} and threshold $p$ for a good
$2$--$1$--$1$ event. If $x_{\mathbf e}(B),x_{\mathbf f}(A)\leq\omega$
and neither bundle is $2$--$1$--$1$ good, their joint
$2$--$2$--$2$ happiness event has probability greater than $p$.
\end{lemma}

\begin{proof}
Write $\mathbf e=(v,u)$, $\mathbf f=(v,w)$,
$A'=A-\mathbf e$, $B'=B-\mathbf f$, $U=\delta(u)-\mathbf e$,
$W=\delta(w)-\mathbf f$, and $Z=\delta(u)\cap\delta(w)$.
Put $b_e=x_{\mathbf e}(B)$, $b_f=x_{\mathbf f}(A)$,
$a=Kh$, and $v_*=2h+r+3d_0$.
The coupled window and non-goodness give
\[
\Pr[U_T+A'_T\leq1]<a-b_e,\quad
\Pr[W_T+B'_T\leq1]<a-b_f.
\]
The two raw counts are Poisson-binomial because their blocks are disjoint.
For the first, the unrounded identity preceding \cite[Eq.~(56)]{kko21}
and $x(A)+x(C)\leq1+r+d$ give
\[
\E[U_T+A'_T]
=x(\delta(u))+x(A)-2x(\mathbf e)+b_e+x_{\mathbf e}(C)
\leq2+v_*+b_e.
\]
The same bound holds for the second count with $b_f$.

{\bf Crossing-dependent excess budgets.}
For either raw count let $p_j$ be its rank probabilities and let $b$ be
its crossing mass. Put $A_b=a-b$, $B_b=1-2a-v_*+b$.
The mean bound gives $p_2>B_b-p_0>1-3a-v_*+2b$, since
$p_0+p_1<A_b$. Ordinary Newton gives the initial zero-atom cap
\[
R_0(b)=\frac{(a-b)^2}{2(1-3a-v_*+2b)}.
\]
For $j=0,1$, define
\[
R_{j+1}(b)=\frac{(a-b)^2}
{2(a-b)-R_j(b)+D(1-2a-v_*+b-R_j(b))},\quad R(b)=R_2(b).
\]
Indeed, if $p_0\leq R_j$, apply Lemma~\ref{lem:atom_newton} with
$c=B_b-R_j$ to obtain
$(A_b-p_0)^2\geq k(c)p_0(B_b-p_0)$.
Replacing $p_0$ in the resulting bracket by $R_j$ proves the next cap.
Every atom argument lies in $(0,1)$, every denominator is positive,
and both updates decrease. No limiting iteration is used. Therefore
\[
\E[(2-X)_+]<L(b):=a-b+R(b),\quad
\E[(X-2)_+]<U(b):=a+v_*+R(b).
\]
In particular $U(b)-L(b)=v_*+b$; the crossing term in the raw mean
has not been discarded. The initial $R_0$ decreases with $b$.
By the increment bound for $D$, increasing $b$ and decreasing the
previous cap cannot decrease the next denominator. Thus $R,L,U$
decrease throughout $0\leq b\leq\omega<a$.

Order the two crossing masses and cover their triangle by closed
rectangles from the following knots, multiplied by $\omega/1024$:
\[
0,1,2,4,8,16,32,64,128,192,256,320,384,448,512,
640,768,896,960,992,1008,1024.
\]
There are $231$ rectangles with first interval index at least the second.
In either orientation write their intervals as
$b_e\in[l_x,c_x]$, $b_f\in[l_y,c_y]$, and put
\[
L_x=L(l_x),\quad U_x=U(l_x),\quad L_y=L(l_y),\quad U_y=U(l_y).
\]
These are upper budgets on the entire closed rectangle, not values
assigned to the unknown crossing masses.

{\bf The common separator and conditioning probabilities.}
The endpoint-tree event has probability at least $\tau=1-3d_0$.
Under it $Z_T$ is binary, and the sum of the two raw identities is
$D_T+2Z_T$ for a block $D$ disjoint from $Z$.
Lemma~\ref{lem:asymmetric_separator} applies with center four and budgets
\[
L_s=(L_x+L_y)/\tau,\quad U_s=(U_x+U_y)/\tau.
\]
Starting $z_L^{(0)}=z_H^{(0)}=1/2$, take three updates
\[
z_L^{(j+1)}=\frac{U_s}{1-L_s/(1-z_L^{(j)})},\quad
z_H^{(j+1)}=\frac{L_s}{1-U_s/(1-z_H^{(j)})}\quad(0\leq j<3).
\]
All denominators are positive and the endpoints decrease.
The separator inequality and induction imply
$\Pr[Z_T=1\mid u,v,w\text{ trees}]\leq z_L:=z_L^{(3)}$
or at least $1-z_H:=1-z_H^{(3)}$.
Both orientations use this same branch.

For branch $i\in\{0,1\}$ impose the endpoint trees, $Z_T=i$,
$C_T=0$, $\mathbf e(B)_T=0$, $\mathbf f_T=0$, and
$\mathbf e(A)_T=1$. Put $a_0=1/2-h-c_x-2r-4d_0$ and
\begin{align*}
P_{0b}&=\tau(1-z_L-2r-c_x-4d_0),\\
P_{0p}&=\tau(1/2-z_L-2r-d_0-h-c_x),&P_0&=P_{0p}a_0,\\
P_{1b}&=\tau(1-z_H)(1-2r-c_x-4d_0),\\
P_{1p}&=\tau(1-z_H)(1/2-2r-h-c_x-4d_0),&
P_1&=P_{1p}(a_0-z_H)/(1-z_H).
\end{align*}
They bound probabilities before deleting $\mathbf f$, after that deletion,
and after the final selection. Each satisfies $0<P_i<P_{ip}<P_{ib}<1$.
For the low branch use union bounds; in particular
$x(C)\leq2r+d_0$ and $x(\mathbf f)\leq1/2+h$ give the stated
$P_{0p}$. For the high branch selecting binary $Z$ decreases disjoint
deleted means. The selected $\mathbf e(A)$ probability is at least
$(a_0+q_Z-1)/q_Z\geq(a_0-z_H)/(1-z_H)$.

{\bf Residual means.}
The raw residual satisfies $x(A')\geq1/2-h-r+b_e$.
Conditioning the three endpoint trees loses at most $3d_0$ from any
exterior mean: condition their internal union to maximum rank, and use
homogeneity and its total rank deficit. Selecting $Z$ in the high branch
loses at most $z_H$ by binary rank coupling; disjoint deletions cannot
decrease the mean. The only overlap with the $\mathbf f$ deletion is
$\mathbf f(A)$, whose preceding conditional mean is at most $c_y/P_{ib}$.
Thus a lower mean before final selection is
\[
a_{ip}=1/2-h-r-3d_0+l_x-c_y/P_{ib}-iz_H.
\]
For upper means, homogeneous deleted-mass domination gives, for any
blocks $S,D$, the inequality
$\E[S_T\mid D_T=0]\leq\E[S_T]+\E[(D\setminus S)_T]$.
This follows by applying the disjoint version to $S\setminus D$ and
subtracting the deleted overlap. Before conditioning,
\[
x(A')\leq1/2+h+r+d_0+b_e-x(C\setminus\mathbf e),\quad
x(B')\leq1/2+h+r+d_0+b_f-x(C\setminus\mathbf f).
\]
Group the deletions $Z,C,\mathbf e(B)$ in the low branch, and
$C,\mathbf e(B)$ after selecting $Z$ in the high branch. The overlap
with $B'$ contains $\mathbf e(B)$, whereas that with $A'$ is empty.
Cancel the $C$ terms and use $x(C)\leq2r+d_0$; the resulting upper
means are $1/2+h+2c_x+3r+2d_0+(1-i)z_L$ and
$1/2+h+c_y+3r+2d_0+(1-i)z_L$.
Lemma~\ref{lem:nested_excess}, with the separate raw budgets, now gives
\begin{align*}
\ell_i&=a_{ip}-(U_x-L_x)/P_{ip}-L_x/P_i,\\
\alpha_i&=1/2+h+2c_x+3r+2d_0+(1-i)z_L+U_x/P_{ip}-(U_x-L_x),\\
\beta_i&=1/2+h+c_y+3r+2d_0+(1-i)z_L+U_y/P_{ip}-(U_y-L_y).
\end{align*}
The final $A'$ mean is at least $\ell_i$; both before and after selection
the $A',B'$ means are at most $\alpha_i,\beta_i<1$.
Using endpoint budgets in these differences is justified by rewriting
the lower correction as $-U_x/P_{ip}-L_x(1/P_i-1/P_{ip})$ and the upper
correction as $L_x+U_x(1/P_{ip}-1)$, and likewise for $y$.
All coefficients then have the required signs. There is no bound on a
difference of unknown expectations inferred from two upper bounds.

Delete $B'$ and set
\begin{align*}
Q_i&=P_i(1-\beta_i),& \psi_x&=L_x/Q_i,&\psi_y&=L_y/Q_i,\\
g&=L_x/P_i,& w_x&=U_x/P_{ip},& w_y&=U_y/P_{ip},\\
v_x&=w_x/(1-\beta_i),& j&=2-i.
\end{align*}
The counted blocks are $(A',U\setminus Z,W\setminus Z)$, with target
$(1,1-i,j)$. Their complete subset table is
\begin{center}
\begin{tabular}{c c c}
\toprule
Union&lower mean&upper mean\\
\midrule
$A'$&$\ell_i$&$\alpha_i+\beta_i$\\
$U\setminus Z$&$(1-i)(2-g-\alpha_i)$&$j+v_x-a_{ip}$\\
$W\setminus Z$&$j-\psi_y$&$j+w_y$\\
$A'\cup(U\setminus Z)$&$j-\psi_x$&$j+v_x$\\
$A'\cup(W\setminus Z)$&$\ell_i+j-\psi_y$&$j+w_y+\alpha_i$\\
$(U\cup W)\setminus Z$&$L_{UW}$&$2j+w_x+w_y-a_{ip}$\\
$A'\cup((U\cup W)\setminus Z)$&$2j-\psi_x-\psi_y$&$2j+w_x+w_y$\\
\bottomrule
\end{tabular}
\end{center}
Here $L_{UW}=4-\psi_x-\psi_y-\alpha_i-\beta_i$ for $i=0$, and
$L_{UW}=1-\psi_y$ for $i=1$.
For the lower columns, use the raw lower-excess estimates on the final
event, and the pre-deletion estimate for $U$ in the low branch.
For the upper columns, commute the last two conditions: delete $B'$
first, then select $\mathbf e(A)$. Deletion increases $A'$ by at most
$\beta_i$ and leaves its lower mean at least $a_{ip}$; the containing
$X$ mean is at most $j+v_x$. After deleting $B'$, its containing
$Y$ becomes $W\setminus Z$, of mean at most $j+w_y$.
The other unions are contained in $A'+Y$, $X+Y-A'$, and $X+Y$ before
deletion, giving the remaining entries. The final disjoint selection
only decreases upper means. Its positive probability was already proved;
commuting the two conditions defines the same law, not independent events.

{\bf Capacity and the two disjoint orientations.}
For every rectangle and branch, all table entries satisfy
$\kappa_S-1<L_S\leq U_S<\kappa_S+1$ for their target sums.
Use the four-product capacity bound, or the stronger vertex bound of
Lemma~\ref{lem:coupled_subset_capacity}, and denote the resulting lower
bound by $\mathcal C_i$. The low branch has the compulsory proper-cut
successes $W\setminus Z\geq1$ and $A'+U\geq1$. Removing the former
gives coefficient factor $e^{-2}/2$. In the high branch, the selected
$\mathbf e,Z$ edges connect the proper three-atom union. At the zero
target in $U\setminus Z$, specialize that coordinate to zero in capacity.
Positive capacity excludes the zero polynomial; the remaining bivariate
polynomial has zero constant term, giving factor $e^{-1}/2$.
Thus the oriented event has probability at least
\[
B_e^i=Q_i(e^{-j}/2)\mathcal C_i.
\]
Exchange $\mathbf e,\mathbf f$, $u,w$, and $A,B$ in the entire
construction, including both interval quadruples, to obtain $B_f^i$.
The two events have $(\mathbf e_T,\mathbf f_T)=(1,0)$ and $(0,1)$,
so they are disjoint. Each gives all three boundary degrees two.
The common separator branch is unchanged by the exchange. Hence in
that branch the happiness probability is at least $B_e^i+B_f^i>p$.
The last strict comparison and every preceding premise are checked in
both branches on every closed rectangle. Shared endpoints cover the
whole triangle. Only the same-branch probabilities are added.
The resulting happiness event retains one shared thinning coin; the
two oriented subevents are not separately charged in the payment proof.
\end{proof}

The preceding capacity estimate can be sharpened by retaining the
additive relations between the seven subset means.

\begin{lemma}[Capacity bound with coupled subset means]
\label{lem:coupled_subset_capacity}
Let $f$ be a real stable polynomial in three variables with nonnegative
coefficients and $f(\mathbf1)=1$, and let
$\boldsymbol\kappa\in\mathbb Z_{\geq0}^3$.
For each nonempty $S\subseteq\{1,2,3\}$, suppose that
\[
L_S\leq\sum_{i\in S}\partial_i f(\mathbf1)\leq U_S,
\quad \kappa_S-1<L_S\leq U_S<\kappa_S+1,
\]
where $\kappa_S:=\sum_{i\in S}\kappa_i$ and the singleton lower
bounds are nonnegative. Define
\[
\mathcal P:=\{\boldsymbol\mu\in\mathbb R^3:
L_S\leq\sum_{i\in S}\mu_i\leq U_S
\text{ for every nonempty }S\subseteq\{1,2,3\}\}.
\]
For $k\in\{1,2,3\}$ and $\boldsymbol\mu\in\mathcal P$, put
\begin{align*}
d_k(\boldsymbol\mu)&:=\min_{|S|=k}\{1+\sum_{i\in S}(\mu_i-\kappa_i)\},\\
e_k(\boldsymbol\mu)&:=\min_{|S|=k}\{1-\sum_{i\in S}(\mu_i-\kappa_i)\},\\
C(\boldsymbol\mu)&:=\min\{d_1d_2d_3,e_1d_1d_2,e_1e_2d_1,e_1e_2e_3\},
\end{align*}
where every factor in the last line is evaluated at $\boldsymbol\mu$.
Then
\[
\operatorname{Cap}_{\boldsymbol\kappa}(f)
\geq\min_{\boldsymbol v\in\operatorname{vert}(\mathcal P)}\{C(\boldsymbol v)\}.
\]
This bound is at least the four-product bound obtained by minimizing
each subset mean separately over its interval.
\end{lemma}

\begin{proof}
The actual mean vector belongs to $\mathcal P$, so this polytope is
nonempty. The singleton constraints make it compact. Every $d_k,e_k$
is a positive concave function on $\mathcal P$, being the minimum of
finitely many affine functions. Hence $\log d_k$ and $\log e_k$ are
concave. The logarithm of each of the four products is concave, and
their pointwise minimum $\log C$ is concave as well.
Writing any $\boldsymbol\mu\in\mathcal P$ as a convex combination
$\sum_v\lambda_v\boldsymbol v$ of vertices gives
\[
\log C(\boldsymbol\mu)
\geq\sum_v\lambda_v\log C(\boldsymbol v)
\geq\min_v\{\log C(\boldsymbol v)\}.
\]
Apply \cite[Theorem~5.11 and Corollary~5.13]{gkl24} at the actual
mean vector to obtain the capacity bound. Finally, every factor at
every vertex is at least its separately minimized interval bound.
Multiplying positive factors and taking the four-product minimum
proves the last assertion.

For rational endpoints, the bound is a finite rational calculation.
Every vertex has three linearly independent active constraint normals;
otherwise a nonzero direction orthogonal to all active normals gives
a short feasible segment through that point. Conversely, a feasible
intersection of three independent active planes is a vertex.
Thus one enumerates triples of the fourteen boundary planes, retains
their feasible intersections, and evaluates $C$ there. This remains
valid when $\mathcal P$ has dimension less than three.
\end{proof}

Apply the lemma separately to the low- and high-branch tables above,
with targets $(1,1,2)$ and $(1,0,1)$. In each application all seven
intervals must refer to the same conditional law. The compulsory-cut
coefficient factors and the conditioning probabilities are unchanged.
The final parameter choice uses this stronger bound in the paired-event
comparisons when the original four-product estimate does not suffice.

With the probability estimates in place, we construct the fractional matching used
to distribute payment increases. The matching discount is imposed only at eligible
atoms incident to a bad bundle.

\begin{lemma}[Localized matching]
\label{lem:candidate_matching}
For a degree cut $S$ and its child $u$, put
$U_u:=x(\delta^\uparrow(u))$, and let $b(u)$ be the indicator that $u$
is incident to an internal bad bundle of $S$. Define
\[
\begin{aligned}
F_u&:=1-\epsilon_B\mathbf1_{\{\epsilon_F\leq U_u\leq1-\epsilon_F\}}b(u),\\
Z_u&:=1+\mathbf1_{\{|\mathcal A(S)|\geq4,\ U_u\leq\epsilon_F\}},
\quad\alpha:=2d_0.
\end{aligned}
\]
For every $0\leq d\leq d_0$, there are nonnegative matching values
$m_{\mathbf e,u}$ such that, for every good internal bundle
$\mathbf e=(u,v)$,
\[
m_{\mathbf e,u}F_u+m_{\mathbf e,v}F_v\leq(1+\alpha)x_{\mathbf e},
\quad
\sum_{\substack{\mathbf e\text{ good}\\\mathbf e\ni u}}
m_{\mathbf e,u}=U_uZ_u.
\]
Thus the fractional discount is required only at atoms incident to a bad
bundle.
\end{lemma}

\begin{proof}
Lemma~\ref{lem:candidate_good_threshold} replaces the published $k=9$
endpoint threshold by $k_{\rm good}$, while
Lemma~\ref{lem:final24_competing} supplies the source fact that an
atom is incident to at most one bad bundle.  We check every cut family in
the max-flow proof of \cite[Lemma~6.2]{kko21} with the localized $F_u$.
The source network gives each good bundle supply $(1+\alpha)x_{\mathbf e}$
and each child sink demand $U_uF_uZ_u$. Saturating these demands and
dividing the incoming flow at $u$ by $F_u>0$ gives the stated identities.
The extra scalar comparisons needed below are
\[
0<\epsilon_B<1,\quad 1/2+k_{\rm good}h<1-\epsilon_F,\quad
2\epsilon_F<(1/2+k_{\rm good}h)(1-\epsilon_B).
\]

For a three-atom degree cut, a hypothetical bad bundle forces the third
atom to have degree at least $3-(4k_{\rm good}+2)h$.  Thus the contradiction
in \cite[Lemma~6.4]{kko21} remains valid because
\[
3-(4k_{\rm good}+2)h>2+d_0.
\]
There are therefore no bad bundles in this case, and $F_u=Z_u=1$.
The corresponding source-only max-flow cut has capacity margin
\[
(1+\alpha)(2-d_0/2)-(2+d_0)>0.
\]
For the source cuts with at least five atoms, Eqs.~(28)--(29) of the source
give, with $n:=|\mathcal A(S)|$,
\[
M_{\geq5}(n):=(1+\alpha)
(n-1-\frac{d_0}{2}-\frac n2(\frac12+h))
-(2+d_0+\epsilon_F n).
\]
Its coefficient in $n$ is positive, so the worst case is $n=5$, where
$M_{\geq5}(5)>0$. These bounds use only $F_u\leq1$ and
$U_uZ_u\leq U_u+\epsilon_F$, both unchanged by localization.
For a four-atom cut with at most one bad bundle,
the corresponding margin is
\[
M_{4,\leq1}:=(1+\alpha)(\frac52-\frac{d_0}{2}-h)
-(2+d_0+4\epsilon_F)>0.
\]
These are the two places where the source used only $h\leq0.01$,
$\epsilon_F\leq0.1$, and $\alpha\geq2d$.  In the four-atom case with two bad bundles, the
argument forcing every upward mass to be $\epsilon_F$-fractional remains
valid because
\[
\frac{2-\epsilon_F}{3}>1/2+k_{\rm good}h.
\]
The two bad bundles cover all four atoms. Thus every upward mass is
fractional, every $b(u)=1$, and all four sink factors are $1-\epsilon_B$.
The remaining capacity comparison is exactly
\[
M_{4,2}:=(1+\alpha)(2-2h-d_0/2)
-(2+d_0)(1-\epsilon_B)>0.
\]

It remains to check the nontrivial max-flow cut in
\cite[Eq.~(30)]{kko21}. Let $W$ be a nonempty proper subset of the atoms,
put $w:=|W|$, and let $k\leq w$ be the number
of distinct bad bundles adjacent to $W$. Since bad bundles form a matching,
choose one distinct endpoint in $W$ for each such bundle. For a chosen
endpoint with $U_u\leq\epsilon_F$, its demand is at most $2\epsilon_F$,
and hence at most $(1/2+k_{\rm good}h)(1-\epsilon_B)$.
Otherwise the bad-endpoint bound gives
$\epsilon_F<U_u\leq1/2+k_{\rm good}h<1-\epsilon_F$;
its factor is $F_u=1-\epsilon_B$ and the same demand bound holds.
For every unchosen atom use $U_uF_uZ_u\leq1+d$.
The total sink demand in $W$ is therefore at most
\[
k(1/2+k_{\rm good}h)(1-\epsilon_B)+(w-k)(1+d).
\]
By \cite[Lemma~6.3]{kko21}, the internal-bundle incidence at $W$ is at
least $w-d/2$. Removing the $k$ bad bundles leaves good supply at least
\[
(1+\alpha)(w-d/2-k(1/2+h)).
\]
Subtracting the preceding demand gives the same scalar inequality as
Eqs.~(30)--(31) of the source:
\[
w(\alpha-d)\geq kB(d)+\frac d2(1+\alpha),
\quad
B(d):=\frac\alpha2+(k_{\rm good}+1)h+\alpha h
-\frac{\epsilon_B}{2}-k_{\rm good}\epsilon_Bh-d.
\tag{*}
\]
For $k=0$, the slack in (*) is minimized at $w=1$ and $d=d_0$, where
\[
\alpha-d_0-\frac{d_0}{2}(1+\alpha)
=\frac{d_0}{2}(1-\alpha)>0.
\]
Now suppose $k\geq1$ and define
\[
S(w,k,d):=w(\alpha-d)-kB(d)-\frac d2(1+\alpha).
\]
Since
\[
\frac{\partial S}{\partial d}
=-w+k-\frac{1+\alpha}{2}<0,
\]
the worst value of $d$ is $d_0$. At that value,
\[
B(d_0)=-\frac{1+2k_{\rm good}h}{2}\,10^{-11}<-5\cdot10^{-12},
\quad S(1,1,d_0)>0.
\]
Increasing $w$ by one increases $S$ by
$\alpha-d_0=d_0>0$, while increasing $k$ by one increases $S$ by
$-B(d_0)>0$. Hence $S(w,k,d)>0$ for every $1\leq k\leq w$.
Thus every source--sink cut has the required capacity, and the matching
exists with the required capacity inequality and conservation identity.
Only the selected bad endpoints used the fractional discount; no atom
with $b(u)=0$ was discounted in the cut bounds.
This complete cut verification replaces the source's sufficient condition
$\epsilon_B\geq21h$; that condition is not an additional premise here.
\end{proof}

The payment analysis needs a structural alternative at each non-root atom. The
following trichotomy supplies either substantial bad mass, substantial
individually good mass, or a pair of bundles with a sufficiently likely joint
happiness event.

\begin{lemma}[Structural trichotomy]
\label{lem:candidate_structural_trichotomy}
Let $v$ be a non-root hierarchy atom with parent $S$, let $A,B,C$ be the
degree partition of $\delta(v)$, and suppose $0\leq d\leq d_0$. At least one
of the following holds:
\begin{enumerate}[label=(\alph*)]
\item \label{item:candidate_structural_bad}
the $x$-mass of bad top edges in $\delta^\rightarrow(v)$ is at least
$1/2-h$;
\item \label{item:candidate_structural_good}
the $x$-mass of top edges in $\delta^\rightarrow(v)$ that are
$2$--$1$--$1$ good with respect to $v$ is at least $1/2-h-d$;
\item \label{item:candidate_structural_paired}
there are two top half bundles $\mathbf e,\mathbf f$ in
$\delta^\rightarrow(v)$ such that
$x_{\mathbf e}(B)\leq\omega$, $x_{\mathbf f}(A)\leq\omega$, and their joint
$2$--$2$--$2$ happiness event has probability at least $p$.
\end{enumerate}
\end{lemma}

\begin{proof}
The small- and large-bundle estimates in
Lemma~\ref{lem:candidate_common_probability} make every non-half top bundle
$2$--$1$--$1$ good. Inspecting the proof of \cite[Lemma~5.25]{kko21}, its
only probabilistic inputs are the four individual happiness bounds supplied
by the same lemma; the rest of its argument only partitions incident bundle
mass. For completeness, at most four half bundles are incident to an atom,
since $5(1/2-h)>2+d_0$. If good but not $2$--$1$--$1$-good mass on
one degree side exceeds $1/2+4\omega$, two half bundles each contribute at
least $\omega$ there: otherwise one contributes at most $1/2+h$ and the other
three contribute less than $3\omega$, totaling at most
$1/2+h+3\omega\leq1/2+4\omega$. For each of these two bundles, the
direct mixed test forces opposite-side mass below $\omega$. The same-side
test then makes one $2$--$1$--$1$ good, a contradiction. Thus the
$1/2+4\omega$ bound holds for the replacement good class.
An ordinary good half bundle has its endpoint-tree $2$--$2$ event with
probability at least $(1-d_0)\gamma_{\rm good}h>p$.

We now repeat the proof of \cite[Theorem~5.28]{kko21}. If the first
alternative fails, there is no bad half bundle, because every half bundle has
mass at least $1/2-h$. If there is at most one half bundle, all remaining
mass is non-half; since $x(\delta^\rightarrow(v))\geq1-d$, that mass is at
least $1-d-(1/2+h)=1/2-h-d$, and the second alternative follows. Otherwise choose two
good half bundles $\mathbf e,\mathbf f$. If either is $2$--$1$--$1$ good,
its mass alone proves the second alternative. If neither is, the two-sided
mixed-bundle estimate included in
Lemma~\ref{lem:candidate_common_probability} implies that each bundle has at
least one of its $A$- and $B$-masses at most $\omega$, while each bundle has mass exceeding $2\omega$ after removing
the remainder, since $1/2-h-2\omega-4r-2d_0>0$.
Lemmas~\ref{lem:candidate_same_side_tail} and
\ref{lem:candidate_gkl_window} rule out the two small masses lying in the same
degree part. After exchanging $A,B$ or $\mathbf e,\mathbf f$ if necessary, this
leaves $x_{\mathbf e}(B),x_{\mathbf f}(A)\leq\omega$.
Lemma~\ref{lem:candidate_reparam_527} then gives the third alternative.
These are exactly the uses of \cite[Lemmas~5.21--5.24 and~5.27]{kko21} in
the source proof; after their replacement, no step uses the old inequality
$12r\leq h$.
\end{proof}

\subsection{Parity estimates and synchronized reductions}
\label{sec:reparam_payment_estimates}

At an endpoint incident to a bad bundle, we need a sharper parity estimate than
the universal polygon bound. We specify the finite, outward-rounded bootstrap used
to obtain it.

\begin{definition}[Parity bootstrap at bad endpoints]
\label{def:final24_parity}
Put
\[
\rho:=20d_0,\quad q_-:=1/2-h-\rho,\quad q_+:=1/2+h+\rho,
\quad \nu:=3/2-h-3\rho,
\]
\[
b_0:=1/2-h-2\rho,\quad
\overline b:=\gamma_{\rm good}h/(1-2d_0),\quad M_0:=2h+4\rho.
\]

For a real number $x$, let $\lceil x\rceil_{20}$ and
$\lfloor x\rfloor_{20}$ denote upward and downward rounding to
the $10^{-20}$ grid. For fixed inputs
$a_*,b_*>0$, $D_*\geq0$, $s_*\in\mathbb R$, and a starting bound
$B_0\geq0$, define eight stages,
indexed by $j=0,\ldots,7$:
\[
\begin{aligned}
Q_j&:=\lceil b_*/(a_*\psi(s_*-B_j))\rceil_{20},\\
C_j&:=\lceil\min\{B_j,Q_j/\nu\}\rceil_{20},\\
\ell_j&:=\lfloor1+D_*-3C_j-2Q_j\rfloor_{20}.
\end{aligned}
\]
For $j<7$, set
\[
B_{j+1}:=\lceil\min\{B_j,Q_j/\nu,
Q_j^2/(\ell_j k(\min\{\ell_j,1\}))\}\rceil_{20}.
\]
Each use below verifies $0<s_*-B_j<1$ and $\ell_j>0$.
The initial instance uses
\[
(a_*,s_*,b_*,D_*,B_0)=(q_-,b_0,\overline b,0,0.226).
\]
Put
\[
\Delta_0:=\lceil\min\{1,(M_0+2B_7+Q_7)/(1-q_+)\}\rceil_{20}.
\]
\[
q_{\rm bad}:=q_++(1-q_+)\Delta_*+2d_0<q_0.
\]
\end{definition}

The bootstrap controls how much the residual count changes between the two bundle
states. We use this control to bound the endpoint's odd-degree probability under a
polygon event.

\begin{lemma}[Polygon-conditioned parity at bad endpoints]
\label{lem:candidate_bad_bundle_parity}
\label{lem:final24_bad_parity}
Use Definitions~\ref{def:final24_parameters} and
\ref{def:final24_parity}. Let $\mathbf e=E(u,v)$ be a bad
half bundle internal to a degree parent $S$, and let $\mathcal P$ be the
polygon parent of an upward bottom edge out of $u$. For the polygon
event $E_{\mathcal P}$ from Lemma~\ref{lem:candidate_polygon_flow},
including its independent thinning to probability $P$, one has
\[
\Pr[\delta(u)_T\text{ odd}\mid E_{\mathcal P}]\leq q_{\rm bad}.
\]
The assertion is uniform for all hierarchy errors $0\leq d\leq d_0$.
\end{lemma}

\begin{proof}
Laminarity gives $u,v\subseteq S\subsetneq\mathcal P$: a bottom parent
above $u$ cannot be its degree parent, and therefore contains all the
children of $S$. All these cuts avoid the distinguished vertices
$u_0,v_0$. Write $\Omega:=\{u,v,\mathcal P\text{ are trees}\}$.

We first control means under this conditioning. Conditioning an internal
edge set to its maximum rank increases its internal marginals and
decreases the external marginals, by
\cite[Fact~2.12 and Lemma~2.23]{kko21}. If the current rank deficit is
$\xi$, the total increase is $\xi$; homogeneity makes the total decrease
$\xi$ as well. Thus the mean of any edge-set count changes by at most
$\xi$, even when that set meets both sides. Initially the rank deficits
of $u,v,\mathcal P$ are at most $d/2$. Conditioning successively on
$u$, then $v$, then $\mathcal P$ being trees has deficits at most
$d/2,d,2d$, respectively. Every edge-set mean therefore changes by at
most $7d/2$. These maximum-rank conditionings preserve homogeneous
real stability. In particular, conditional on $\Omega$,
\[
\E[\delta(u)_T],\E[\delta(v)_T]\in[2-\rho,2+\rho],
\quad q:=\Pr[\mathbf e_T=1]\in[q_-,q_+].
\]
Here and until the transfer to $E_{\mathcal P}$, probabilities and
expectations are conditional on $\Omega$ unless specified otherwise.

Since $u$ and $v$ are trees, a spanning tree contains at most one edge
of $\mathbf e$. Identifying the variables of this parallel bundle
therefore gives a binary coordinate $E:=\mathbf e_T$ in a homogeneous
multiaffine stable polynomial. Project this law off the coordinate $E$.
Its two conditional laws have adjacent total ranks, and
\cite[Theorem~2.10]{kko21} gives a monotone coupling in which changing
$E=1$ to $E=0$ adds exactly one remaining edge. In particular it adds
at most one edge to any prescribed subset. This is the same projection
argument used in \cite[proof of Lemma~2.27]{kko21}.

Each of the three tree events fails with probability at most $d/2$.
Consequently
\[
\Pr[\mathcal P\text{ is not a tree}\mid u,v\text{ are trees}]
\leq\frac{d}{2(1-d)}\leq2d_0.
\]
Badness is defined conditional on $u,v$ being trees. Adding the
$\mathcal P$-tree requirement can only decrease its numerator, so
\[
b_\Omega:=\Pr[\delta(u)_T=\delta(v)_T=2]\leq\overline b.
\]
This step does not assert independence of badness and the additional
tree event.

Let $X:=|(\delta(u)\setminus\mathbf e)\cap T|$ and
$Y:=|(\delta(v)\setminus\mathbf e)\cap T|$. They count disjoint edge
sets, and $X+Y=|\delta(u\cup v)\cap T|$. Define
\[
\Delta:=\E[X+Y\mid E=0]-\E[X+Y\mid E=1].
\]
The preceding coupling gives $0\leq\Delta\leq1$. For any subcount
$J$ of $X$, the same coupling satisfies
\[
0\leq\Delta_J:=\E[J\mid E=0]-\E[J\mid E=1]\leq\Delta,
\quad \Pr[J_0\ne J_1]\leq\Delta_J.
\]
Write $\mu_X:=\E[X\mid E=1]$, $\mu_Y:=\E[Y\mid E=1]$, and
$m:=\mu_X+\mu_Y$. The cut-mean bounds and the coupling imply
\[
1-\rho\leq\mu_X,\mu_Y\leq2+\rho-q\leq3/2+h+2\rho,
\]
\[
3-2\rho-q\leq m\leq4+2\rho-2q-(1-q)\Delta.
\]
For example, if $\Delta_X:=\E[X\mid E=0]-\E[X\mid E=1]\leq1$, then
\[
\mu_X=\E[X]-(1-q)\Delta_X
\geq2-\rho-q-(1-q)=1-\rho.
\]
The total-mean identity subtracts exactly $(1-q)\Delta$ and gives the
second display.

Conditional on $E=1$, the union $u\cup v$ is a connected proper set,
so $Z:=X+Y\geq1$. Put
$D:=\max\{(1-q)\Delta-M_0,0\}$.
The displayed identities give $m\geq1+\nu>2$ and, if $D>0$,
$m\leq3-D<3$. If $D=0$, then
$\Delta\leq M_0/(1-q)\leq M_0/(1-q_+)<\Delta_*$ directly.
Suppose $D>0$ and write $p_i:=\Pr[Z=i\mid E=1,\Omega]$.
Here $m>2$, so gapless rank support and $m<3$ imply $p_2>0$.
Since $\nu>1.49$,
\[
p_1\leq e^{-\nu}<0.226=B_0;
\]
the last inequality follows by taking the exponential Taylor sum
through degree $16$ at $1.49$.
The invariant for the recurrence in Definition~\ref{def:final24_parity}
is the following. Suppose $p_1\leq B_0$,
$m\in[1+\nu,3-D_*]$, and for every applicable bound $p_1\leq B$,
\[
b_*\geq a_*p_2\psi(s_*-B).
\]
Then $p_1\leq B_j$ implies $p_2\leq Q_j$.
The odds bound gives $p_1\leq C_j$; the rank-three estimate gives
$p_3\geq\ell_j$; and Newton gives
$p_1\leq Q_j^2/(\ell_j k(\min\{\ell_j,1\}))$ by
Lemma~\ref{lem:atom_newton} applied to $Z-1$.
The other estimates are those in
Lemma~\ref{lem:final24_rank_split}, and outward rounding preserves
each one. Induction therefore proves $p_1\leq B_j$ and
$p_2\leq Q_j$ at all eight stages.

For the initial instance, $m<3$ and
$\overline b\geq qp_2\psi(b_0-B_j)
\geq q_-p_2\psi(b_0-B_j)$ by the same rank-split lemma.
All split arguments belong to $(0,1)$ and all $\ell_j$ are positive.
Its last step gives
\[
B_7>0,\quad Q_7>0,
\]
\[
D\leq3-m\leq2p_1+p_2\leq2B_7+Q_7,\quad
\Delta\leq\Delta_0<0.17.
\]
No fixed-point convergence is used.

We now retain the actual bundle probability in both edge states and
charge their two contributions to the same badness budget:
\[
\overline b\geq q\Pr[X=Y=1\mid E=1]
+(1-q)\Pr[X=Y=2\mid E=0].
\]
Suppose for contradiction that $\Delta\geq\Delta_*$.
Cover $[q_-,q_+]$ by the $128$ closed cells
\[
a_i:=q_-+(q_+-q_-)i/128,\quad
b_i:=q_-+(q_+-q_-)(i+1)/128,\quad 0\leq i<128.
\]
Fix a cell $[a_i,b_i]$ containing $q$. Under $E=1$, the exact mean
identity above gives
\[
m\leq4+2\rho-2q-(1-q)\Delta\leq3-D_i,
\quad D_i:=2a_i-1-2\rho+(1-b_i)\Delta_*.
\]
Under $E=0$, the two proper cuts have compulsory successes. Thus
$U:=X-1$, $V:=Y-1$ have a stable joint generating polynomial.
Writing $\Delta_X:=\E[X\mid E=0]-\E[X\mid E=1]\in[0,1]$ gives
\[
\E[X\mid E=0]=\E[X]+q\Delta_X\in[2-\rho-q,2+\rho].
\]
The symmetric identity and the sum identity imply
\[
\E[U],\E[V]\in[L_i,1+\rho],\quad L_i:=1-\rho-b_i,
\]
\[
1+\delta_i\leq\E[U+V]\leq2-2a_i+2\rho+b_i\Delta_0<1.2,
\quad\delta_i:=1-2b_i-2\rho+a_i\Delta_*.
\]
Every $D_i,\delta_i$ is positive. Lemma~\ref{lem:bernoulli_estimates}
gives a rank-one atom greater than $0.36$. By \cite[Lemma~2.21]{kko21},
the rank-two atom $q':=\Pr[U+V=2\mid E=0]$ is at least
\[
q_i:=\delta_i(1-\delta_i+\delta_i^2/2-\delta_i^3/6).
\]
Indeed, the one-forced branch is at least $\delta_i e^{-\delta_i}$,
and the zero-forced branch is at least $e^{-1.2}/2>\delta_i$.
The individual nonzero tails exceed $0.39$ and their lower tails at
one exceed $0.70$. If $q'\geq0.1$, the up/down inequality and
Lemma~\ref{lem:quadratic_bounds} give
$\Pr[U=V=1]>0.025$.
Otherwise use $T_{0,0}$ from Definition~\ref{def:tail_profiles} on
the two intervals $[\min\{q_i,0.04\},0.04]$ and $[0.04,0.1]$.
For each right endpoint $u$, choose its fixed coupling $\lambda_u$
from the final downward atom enclosure, as in the goodness proof, and put
\[
F_{i,u}(q)=\frac q2-
\frac{(1+q-2L_i+(2-\lambda_u)T_{0,0}(q))^2}
{8(1-\lambda_u T_{0,0}(q))}.
\]
Every degree-one coefficient and denominator is positive on these
intervals, and $W\geq0$ follows from the lower tail bound $3q^2/2$.
Both intervals are therefore covered by concavity in
Lemma~\ref{lem:normalized_rank_tail}. Let $z_i$ be the minimum of
$0.025$ and the four endpoint values of $F_{i,u}$. Then
\[
\Pr[X=Y=2\mid E=0]=\Pr[U=V=1\mid E=0]\geq z_i.
\]
The possible extra interval when $q_i>0.04$ only weakens this lower bound.

The present-edge contribution to badness is therefore at most
\[
\widetilde b_i:=\overline b-(1-b_i)z_i>0.
\]
Apply the same eight-stage recurrence with inputs
\[
(a_*,s_*,b_*,D_*,B_0)
=(a_i,a_i-\rho,\widetilde b_i,D_i,B_7),
\]
and denote its outputs by
$\widetilde B_j,\widetilde Q_j,\widetilde C_j,\widetilde\ell_j$.
The rank-split argument now uses the sharper individual bound
$\mu_X,\mu_Y\leq2+\rho-q$ and the lower bound $q\geq a_i$.
Thus its split parameter is at least $a_i-\rho-\widetilde B_j$.
The rank-three lower bound uses $m\leq3-D_i$, while the odds bound
still uses $m\geq1+\nu$. The recurrence invariant therefore gives
$p_1\leq\widetilde B_j$ and $p_2\leq\widetilde Q_j$. Every displayed premise, including positivity
of $\widetilde\ell_j$ and membership of the split parameter in $(0,1)$,
is verified by exact rational substitution in every cell.
The finite certificate gives
\[
\min_{0\leq i<128}\{D_i-2\widetilde B_7-\widetilde Q_7\}
>1.53878\cdot10^{-8}>0.
\]
This contradicts $D_i\leq3-m\leq2p_1+p_2$.
The cells share their endpoints exactly and cover every possible $q$;
this is an outward interval certificate, not point sampling.
We conclude $\Delta<\Delta_*$.

To control parity inside $\mathcal P$, take
\[
J:=|((\delta(u)\cap E(\mathcal P))\setminus\mathbf e)\cap T|,
\quad I:=E+J.
\]
Couple $J_0,J_1$ as above. On $J_0=J_1$, the two possible values of
$E+J$ have opposite parity, so either parity atom of their $q$-mixture
is at most $\max\{q,1-q\}$. On a mismatch use the bound one. Hence,
for either parity $\varsigma$,
\[
\begin{aligned}
\Pr[I\text{ has parity }\varsigma]
&\leq\max\{q,1-q\}+\min\{q,1-q\}\Pr[J_0\ne J_1]\\
&\leq q_++(1-q_+)\Delta_*.
\end{aligned}
\]
For the last inequality, write $w:=\max\{q,1-q\}\leq q_+$.
The bound $w+(1-w)\Delta_*$ increases with $w$, since $\Delta_*<1$.
Bounding both parity atoms is necessary because the external boundary
count can flip parity.

Finally, conditional on $\mathcal P$ being a tree, its internal tree
and the contracted tree in $G/\mathcal P$ are independent by
\cite[Fact~2.8]{kko21}. The two child-tree events are internal.
All further requirements in $E_{\mathcal P}$, including the deleted
$C$ boundary, the accepted $A/B$ pair, and its acceptance and thinning
coins, belong to the contracted factor. Therefore
\[
\mathcal L(T\cap E(\mathcal P)\mid E_{\mathcal P},u,v\text{ trees})
=\mathcal L(T\cap E(\mathcal P)\mid\Omega).
\]
Under the left conditioning, $I$ remains independent of
$K_{\rm ext}:=|\delta(u)\cap\delta(\mathcal P)\cap T|$.
The equality $\delta(u)_T=I+K_{\rm ext}$ and arbitrary convolution
with $K_{\rm ext}$ preserve the bound on the larger internal parity
atom. The same product factorization gives
\[
\Pr[u\text{ or }v\text{ not a tree}\mid E_{\mathcal P}]
=\Pr[u\text{ or }v\text{ not a tree}\mid\mathcal P\text{ tree}]
\leq\frac{d}{1-d/2}\leq2d_0.
\]
Restoring these exceptional outcomes gives
\[
q_++(1-q_+)\Delta_*+2d_0=q_{\rm bad},
\]
without dividing by the rare polygon-event probability.
Lemma~\ref{lem:polygon_half_budget} gives the universal envelope $q_0$,
and exact substitution gives $q_{\rm bad}<q_0$.
\end{proof}

The payment formulas require events with prescribed probabilities and paired
reductions that occur together. Independent thinning achieves both requirements
without changing the relevant conditional tree laws.

\begin{lemma}[Synchronized reduction events]
\label{lem:payment_matching_setup}
An event of probability at least $\pi>0$ can be thinned to
probability $\pi$ without changing its conditional tree law.
Thus top events may be thinned to $p$ and polygon events to $P$.
The two reductions arising from one paired $2$--$2$--$2$ event
may moreover be chosen pointwise identical.
\end{lemma}

\begin{proof}
For an event $\mathcal H$ of probability $q_{\mathcal H}\geq\pi$,
use a tree-independent Bernoulli coin $C_{\mathcal H}$ with
probability $\pi/q_{\mathcal H}$. Then
$\mathcal H\cap C_{\mathcal H}$ has probability $\pi$ and the same
conditional tree law as $\mathcal H$, as in \cite[Section~7]{kko21}.

In Case~3 of \cite[Theorem~5.28]{kko21}, let $\mathcal J_{\mathbf e,\mathbf f,u}$ be the joint $2$--$2$--$2$ happiness event for the fixed paired bundles $\mathbf e,\mathbf f$ at $u$. Define $q_{\mathcal J}:=\Pr[\mathcal J_{\mathbf e,\mathbf f,u}]\geq p$, use one Bernoulli event $C_{\mathcal J}$ of probability $p/q_{\mathcal J}$ independent of the tree, and set
\[
\mathcal R_{\mathbf e,u}
=
\mathcal R_{\mathbf f,u}
:=
\mathcal J_{\mathbf e,\mathbf f,u}\cap C_{\mathcal J}
\]
pointwise. Thus the simultaneous-reduction identity required in the proof of \cite[Lemma~7.8]{kko21} is retained without biasing the conditional tree law. All reduction events below use this synchronized convention.

These auxiliary coins define only the payment vector used in the analysis.
The passage from $\mu$ to $\mu_\lambda$ is made directly at the level of
$T\mapsto c(T)+\mathsf J(T)$ by
Lemma~\ref{lem:prelim_slack_to_tour}; it does not require the reduction
events or layered slack vectors to be reconstructed under $\mu_\lambda$.

\end{proof}

Bottom-edge payments can be reduced further when a descendant shares enough
boundary mass with a polygon child. A low-degree representation of the child's
boundary count yields the needed conditional parity bound.

\begin{lemma}[Finite-rank polygon-child parity]
\label{lem:polygon_child_parity}
Let $P$ be a polygon cut, $S$ a child of $P$, and $u$ a proper hierarchy
descendant of $S$. Put
\[
t_u:=x(\delta(u)\cap\delta(S)),\quad
 e_{\rm ch}:=\epsilon_M/2+40d_0.
\]
For $y\in\{0.0025,0.10,0.15,0.20,0.25\}$ and $t_u\geq y$, the accepted and thinned event
$E_P$ satisfies
\[
\Pr[\delta(u)_T\text{ odd}\mid E_P]
\leq Q_2(y):=
\frac{1+(1-y+e_{\rm ch})^2e^{-2(1-y)}}2+14d_0.
\]
The assertion is uniform for $0\leq d\leq d_0$.
\end{lemma}

\begin{proof}
The cut identity and $x(\delta(S\setminus u))\geq2$ give $t_u\leq1+d$. Conditional on $P$ being a tree, its
internal tree factors from the contracted tree
by \cite[Fact~2.8]{kko21}. Acceptance and thinning
use only the contracted boundary pair and independent coins. Therefore
these conditions do not change the internal tree law. Also
\[
\Pr[S\text{ not a tree}\mid E_P]
\leq \frac{d}{2(1-d/2)}\leq d.
\]
Condition on $S$ being a tree. Its internal tree factors again and is
independent of every remaining condition imposed below. The count
\[
Y:=|(\delta(u)\cap E(S))\cap T|-1
\]
is Poisson-binomial, since $u$ is a proper nonempty subset of the connected
set $S$ and its internal count has a compulsory success. The $S$-tree condition may be imposed first. Its internal factor is then
independent of the remaining condition that $P$ is a tree, and of polygon
acceptance. Thus the internal law is exactly the original $S$-tree internal
law. Maximum-rank conditioning increases that mean by at most $d/2$, so
\[
1-t_u\leq\mu:=\E[Y]\leq1+3d/2-t_u.
\]
In particular, $0\leq\mu\leq1+3d_0/2<1.2$.

We construct, outside an event of conditional probability at most $14d$,
a boundary count $B$ that is a sum of at most two independent Bernoulli
variables, is independent of $Y$, and has mean $\theta$ satisfying
$|\theta-t_u|\leq e_{\rm ch}$.

Suppose first that $S$ is a boundary child, including either child of a
triangle. Its exterior boundary is one accepted parent side, with the
parent remainder deleted, and contains one selected edge. The subset
version of Lemma~\ref{lem:polygon_half_budget} bounds the change in its contribution to $B$ by
$\epsilon_M/2+3d$.
Let $H_S:=|\delta(S)\cap E(P)\cap T|$. Connectivity gives $H_S\geq1$.
The hierarchy gives $x(\delta(S)\cap E(P))\leq1+2d$; conditioning $P$ to be
a tree adds at most $d/2$, and conditioning $S$ to be a tree cannot
increase this exterior count. Hence $\E[H_S]\leq1+5d/2$.
Condition further on $H_S=1$. Its failure probability is at most $5d/2$,
and
\[
\E[H_S1_{\{H_S\geq2\}}]\leq2(\E[H_S]-1)\leq5d.
\]
For a subset $F$ of this boundary, conditioning changes its mean by at
most $5d/(1-5d/2)\leq6d$. Indeed, the numerator of the mean difference is
$\E[F_T]\Pr[H_S\ne1]-\E[F_T1_{\{H_S\ne1\}}]$, and each nonnegative
term is at most $5d$. The preceding $P$- and $S$-tree changes cost at most
$d$. Thus the internal boundary mean changes by at most $7d$.
The two boundary contributions are independent Bernoulli variables,
because one belongs to the internal parent factor and the other to the
contracted factor. They are also independent of the internal tree of $S$.
Their mean error is at most $\epsilon_M/2+10d\leq e_{\rm ch}$,
and the total exceptional probability is at most $14d$.

Now suppose that $S$ is an interior child. Its exterior part lies in the
parent remainder and is deleted; its original mass is at most $d$.
Let $S_-,S_+$ be its neighboring non-root children and put
$F_-=E(S_-,S)$, $F_+=E(S,S_+)$. Each link has mass at least $1-d$.
The adjacent union sides have cut value at most $2+4d$, by the cut identity
and \cite[Definition~4.31]{kko21}.
Applying the forest bound \cite[Lemma~2.23]{kko21}
to the two endpoint children and their union
gives $\Pr[(F_\pm)_T\ne1]\leq3d$ under the original law. Conditional on
$P$ being a tree, each failure probability is at most $4d$.
The remaining internal boundary has original mass at most $3d$ and gains
at most $d/2$. Therefore
\[
\Pr[H_S\ne2\mid E_P]\leq 4d+4d+7d/2\leq12d.
\]
After the $S$-tree condition this failure is at most $12d/(1-d)\leq13d$.
The mean of $H_S$ is at most $2+2d$, and hence
\[
\E[H_S1_{\{H_S\ne2\}}]
=\E[H_S]-2+2\Pr[H_S\ne2]\leq28d.
\]
For every boundary subset, the same mean-difference identity now bounds
its change upon conditioning $H_S=2$ by $28d/(1-13d)\leq29d$.
The preceding tree conditions cost at most $d$, and the deleted exterior
part has original mass at most $d$. The total error is at most
$31d\leq e_{\rm ch}$.

After the $P$- and $S$-tree conditions, project the contracted internal law
to the full boundary of $S$, then condition this full projected rank to
two. This preserves stability. Any subcount of that rank-two law is
Poisson-binomial of degree at most two. Take $B$ to be its subcount on
$\delta(u)$. Its Bernoulli factorization is independent of $Y$, by the
$S$-internal factorization; no independence of the two actual boundary
edges is asserted. The union of the exceptional events has probability
at most $14d$. All adjacent sets used in the forest bounds are proper
and avoid the distinguished vertices; no forest bound is applied to a
root atom.

On either good conditioning, write $B=B_1+B_2$ with independent Bernoulli
parameters $b_1,b_2$, padding by a zero parameter if needed. The degree
parity is that of $1+Y+B$. Put
\[
M_Y:=\E[(-1)^Y],\quad M_B:=(1-2b_1)(1-2b_2),\quad \theta=b_1+b_2.
\]
Then its odd probability is $(1+M_YM_B)/2$.
For a Poisson-binomial variable of mean at most $1.2$, a nonnegative
parity moment is at most $e^{-2\mu}$. If no parameter exceeds $1/2$,
this follows factor by factor. If two exceed $1/2$, their product is at
most $(\mu-1)^2\leq0.04<e^{-2.4}\leq e^{-2\mu}$; a larger even number is
impossible. A zero parity factor is immediate. A negative parity moment
has absolute value at most $2\mu-1$, because exactly one parameter exceeds
$1/2$.

If $M_YM_B\leq0$, there is nothing to prove. If both factors are nonnegative,
then $M_B\leq(1-\theta)^2$. For $t_u\leq1$ this gives
\[
M_YM_B\leq(1-t_u+e_{\rm ch})^2e^{-2(1-t_u)}
\leq(1-y+e_{\rm ch})^2e^{-2(1-y)}.
\]
The final inequality holds because $(z+e_{\rm ch})^2e^{-2z}$
increases for $0\leq z\leq1-y$ when $y\geq e_{\rm ch}$.
For $t_u>1$, use $M_YM_B\leq(e_{\rm ch}+d_0)^2$ instead.
If both factors are negative, their means imply
$1/2-e_{\rm ch}<t_u<1/2+3d_0/2$, and
\[
|M_YM_B|\leq(2e_{\rm ch}+3d_0)^2.
\]
For each of the five stated values of $y$, the exact exponential
certificate verifies
\[
e_{\rm ch}<y<1,\quad
\max\{(e_{\rm ch}+d_0)^2,(2e_{\rm ch}+3d_0)^2\}
<(1-y+e_{\rm ch})^2e^{-2(1-y)}.
\]
Restoring the exceptional probability at most $14d_0$ proves the lemma.
Every conditioning probability above is bounded explicitly away from zero,
and the internal law is never divided by the rare probability of $E_P$.
\end{proof}

\subsection{Ancestor credit and polygon charges}
\label{sec:ancestor_polygon_charges}

We now sum the available savings over the successive ancestors of an atom. The
argument keeps the top and bottom event rates distinct and applies the matching
discount only where it is required.

\begin{lemma}[Ancestor estimate]
\label{lem:candidate_ancestor}
Let $h,r,d_0,p,K,\omega$ be as in
Definition~\ref{def:final24_parameters}, and use the parameters in
Definition~\ref{def:final24_parity}.
Fix a KKO hierarchy with error $0\leq d\leq d_0$ and $\beta>0$, and define
$\tau:=t\beta$. Use the synchronized top reduction events with probability $p$,
the polygon events with probability $P=sp$, and the matching from
Lemma~\ref{lem:candidate_matching}.
For a non-root hierarchy atom $u$, define $U:=x(\delta^\uparrow(u))$ and partition
its upward edges into good top, bad, and bottom sets
$\mathcal G_u,\mathcal D_u,\mathcal L_u$. For a good top edge $g\in f=(u',v')$, define
\[
h_u(g):=\frac{
\Pr[\delta(u)_T\text{ odd}\mid\mathcal R_{f,u'}]
+\Pr[\delta(u)_T\text{ odd}\mid\mathcal R_{f,v'}]}2.
\]
For a bottom edge $g$ with polygon parent $Q$, put
$q_u(g):=\Pr[\delta(u)_T\text{ odd}\mid E_Q]$. The bad set includes the edges incident to
$u_0$ or $v_0$, which are never reduced. Set
\[
B_u:=
\tau\sum_{g\in\mathcal G_u}x_gh_u(g)
+s\beta\sum_{g\in\mathcal L_u}x_gq_u(g).
\]
For $U\geq\sigma$,
\[
B_u\leq\tau(1-\chi)F_uU,
\]
where $F_u$ is the localized matching factor of
Lemma~\ref{lem:candidate_matching}: it equals $1-\epsilon_B$ only
when $U$ is $\epsilon_F$-fractional and $u$ is incident to an internal
bad bundle of its degree parent, and equals one otherwise.
For an atom without a degree parent, set $b(u)=0$ and $F_u=1$.
\end{lemma}

\begin{proof}
We use the conditional parity estimates underlying
\cite[Lemma~7.3]{kko21} with the present parameters.
By Lemma~\ref{lem:payment_matching_setup}, thinning to the separate
rates $p$ and $P$ preserves every conditional tree law.
The event $E_S$ from the proof of
Lemma~\ref{lem:candidate_common_probability}, Lemma~\ref{lem:polygon_half_budget}, and
\cite[Corollary~2.17]{kko21} give, for every bottom edge,
\[
q_u(g)
\leq\frac{1+\exp(-2(1-\epsilon_M/2-3d_0))}{2}
<q_0.
\]
The last inequality follows from the exact rational comparison
\[
\sum_{j=0}^{18}\frac{(2(1-\epsilon_M/2-3d_0))^j}{j!}
-\frac1{2q_0-1}>0.
\]
Define $c:=1-sq_0/t$ and $c_{\rm bad}:=1-sq_{\rm bad}/t$.
Let $b(u)$ be the bad-incidence indicator in
Lemma~\ref{lem:candidate_matching}, and put $c_u=c_{\rm bad}$ if
$b(u)=1$, and $c_u=c$ otherwise. Every bottom term is bounded by the
corresponding parity estimate: universally by $q_0$, and at a bad
endpoint by Lemma~\ref{lem:candidate_bad_bundle_parity}.
The normalization in $B_u$ divides the actual
expected upward burden by $p$: its top terms have factor $p\tau$, while
its bottom terms have factor $P\beta=sp\beta$. Subtracting $B_u$ from
$\tau U$ shows that it
is sufficient to prove
\[
x(\mathcal D_u)
+\sum_{g\in\mathcal G_u}x_g(1-h_u(g))
+\sum_{g\in\mathcal L_u}x_g(1-sq_u(g)/t)
\geq(1-(1-\chi)F_u)U.
\tag{16}
\label{eq:candidate_ancestor_master}
\]
Call the left side the available saving. It is additive over disjoint
sets of upward edges, and every contribution is nonnegative.
Each bottom coefficient is at least $c_u\geq c$, so the universal coefficient $c$ remains available
in every case below. We use $c_{\rm bad}$ only in the discounted
fractional case.
Also $c>0$ gives the crude bound $B_u\leq\tau U$ for every $U$, without
the lower restriction $U\geq\sigma$.

We first record the level estimates used in the two large-mass cases.
Let $S_1=\mathsf p(u),S_2,\ldots,S_m=R$ be its successive ancestors,
where $R=V\setminus\{u_0,v_0\}$ is the hierarchy root. Use the level sets
from \cite[proof of Lemma~7.3]{kko21}, with an explicit terminal convention:
\[
\delta^{\geq i}:=\delta(u)\cap\delta(S_i),
\quad
\delta^{\geq m+1}:=\varnothing,
\]
\[
\delta^i:=\delta^{\geq i}\setminus\delta^{\geq i+1},
\quad
t_i:=x(\delta^{\geq i}).
\]
Thus $t_1=U\leq1+d$, $t_{m+1}=0$, and
$x(\delta^i)=t_i-t_{i+1}$. For $i<m$, each level consists entirely of top
edges or entirely of bottom edges. The terminal level
$\delta^m=\delta(u)\cap\delta(R)$ consists of bad edges incident to
$u_0$ or $v_0$ and saves one full unit per unit mass.

Recall $s_0,w_0$ from Definition~\ref{def:final24_parameters}, and put
\[
w_1:=r(1-r),
\]
\[
b_d:=\frac12+4\omega+2r+d,\quad k_d:=\frac{1-r-2d}{2}.
\]
For any $\upsilon\in[r,s_0]$, retain a level only when
$\upsilon+2d\leq t_i\leq1-\upsilon$.
The conditioning in \cite[Claim~7.5]{kko21} applies to its good top
edges and gives
\[
h_u(g)\leq1-\upsilon+\max\{2d,\upsilon^2\}=1-\upsilon+\upsilon^2.
\]
Indeed, each endpoint reduction event makes the relevant ancestor a tree;
the internal and contracted trees factor by \cite[Fact~2.8]{kko21}.
The claim is uniform in the boundary count revealed by the event,
and its intersection-mass interval includes both endpoints.
The retained levels satisfy that interval, and
$\upsilon^2\geq r^2>2d_0\geq2d$.
A bad edge saves one unit, while a bottom edge saves $c$.
The exact parameter comparisons give
\[
r<s_0,\quad c>w_0,\quad k_d>w_0.
\]
Since $\upsilon(1-\upsilon)$ is increasing for $\upsilon\in[r,s_0]$, every unit of retained
mass, of any of the three edge types, saves at least $\upsilon(1-\upsilon)$.

An atom incident to a bad bundle has
$U\leq1/2+k_{\rm good}h<1-\epsilon_F$. Thus neither of the following
large-mass cases is discounted.
Suppose first that $U\geq1-r$, so $F_u=1$.
Let $j$ be the last level with $t_j\geq1-r$, and put
\[
z:=t_{j+1}<1-r,\quad v:=x(\delta^j)=t_j-z\geq1-r-z.
\]
For a real number $a$, write $(a)_+:=\max\{a,0\}$.
If $z\geq1-s_0$, choose $\upsilon:=1-z$, so $r<\upsilon\leq s_0$.
Retain all levels after $j$ whose tail mass is at least $\upsilon+2d$.
Their total mass is at least
\[
z-\upsilon-2d=1-2\upsilon-2d.
\]
Their tail masses are at most $z=1-\upsilon$, so the preceding level estimate
applies. Hence the available saving is at least
\[
F_d(\upsilon):=\upsilon(1-\upsilon)(1-2\upsilon-2d).
\]
Its derivative satisfies
\[
F_d'(\upsilon)=1-2d-(6-4d)\upsilon+6\upsilon^2
\geq1-2d_0-6s_0>0.
\]
It follows that $F_d(\upsilon)\geq F_d(r)$. Since $U\leq1+d$, the exact margin
\[
r(1-r)(1-2r-2d_0)-\chi(1+d_0)>0
\]
proves Eq.~\eqref{eq:candidate_ancestor_master} in this subcase.

For a bottom level $i<m$, its polygon parent is $S_{i+1}$ and
$S_i$ is a child of that polygon. Lemma~\ref{lem:polygon_child_parity}
applies with the starting tail $t_i$, not $t_{i+1}$. For
$y\in\{0.0025,0.10,0.15,0.20,0.25\}$, let $\widehat Q_2(y)$ be the rational upper
enclosure obtained from the degree-$40$ exponential bounds, and put
$c_y:=1-s\widehat Q_2(y)$; recall $t=1$. Exact comparisons give
\[
c_{0.0025}>w_0,\quad
c_{0.25}>0.0099.
\]
For the last top level, the degree partition gives $x(C)\leq2r+d$,
$B\cap\delta(u)=\varnothing$, and $x(A\cap\delta^j)\geq v-2r-d$.
At most $1/2+4\omega$ of this mass is good but not $2$--$1$--$1$ good,
by Lemma~\ref{lem:candidate_structural_trichotomy}. The remaining mass,
at least $(v-b_d)_+$, is bad or $2$--$1$--$1$ good. For each such good
edge, \cite[Claim~7.4]{kko21} bounds the odd probability under one
endpoint event by $r+2d$; the other probability is at most one.
Its average saving is at least $k_d$, and a bad edge saves one.

For $z<1-s_0$, use the following three closed ranges.

If $0\leq z\leq0.35$, a top last level has saving at least
\[
\frac{1-r-2d_0}{2}(1/2-4\omega-3r-0.35-d_0)>\chi(1+d_0).
\]
This is the same $1/2+4\omega$ mass bound and one-endpoint happiness estimate
from Lemma~\ref{lem:candidate_structural_trichotomy}. If the last level is bottom, its starting tail is at least
$1-r>0.25$, so it saves at least
$c_{0.25}(1-r-0.35)>\chi(1+d_0)$.
If it is terminal, its root-bad mass saves $1-r-0.35$, which also suffices.

If $0.35\leq z\leq0.99$, use the five thresholds
\[
(y_1,y_2,y_3,y_4,y_5):=(0.0025,0.10,0.15,0.20,0.25).
\]
Set $\kappa_0:=0$ and
\[
\varepsilon_l:=\min\{y_l,0.01\},\quad
\kappa_l:=\min\{\varepsilon_l(1-\varepsilon_l),c_{y_l}\},\quad 1\leq l\leq5.
\]
These coefficients satisfy $0<\kappa_1\leq\cdots\leq\kappa_5$.
A later level with $t_i\geq y_l+2d$ saves at least $\kappa_l$ per unit:
top edges satisfy Claim~7.5 with $\varepsilon_l$, since
$\varepsilon_l+2d\leq t_i\leq0.99\leq1-\varepsilon_l$ and
$\varepsilon_l^2>2d_0$; bottom edges save at least $c_{y_l}$;
bad edges save one.

For each $l$, these levels form a prefix of the later levels.
Summing $t_i-t_{i+1}$, including the entire level that crosses the
threshold, gives total mass at least $(z-y_l-2d)_+$.
At each level, sum the increments $\kappa_l-\kappa_{l-1}$ for which
$t_i\geq y_l+2d$; their sum is at most that level's saving coefficient.
Interchanging the two finite sums proves the bound.
The available saving is therefore at least
\[
\sum_{l=1}^5(\kappa_l-\kappa_{l-1})(z-y_l-2d)_+.
\]
This counts each level with one coefficient, not with overlapping full
credits. The bound increases with $z$ and decreases with $d$.
The exact endpoint check gives
\[
\sum_{l=1}^5(\kappa_l-\kappa_{l-1})(0.35-y_l-2d_0)
>\chi(1+d_0).
\]

If $0.99\leq z\leq1-s_0$, retain later levels with starting tail at least
$0.0025+2d$. Their mass is at least $0.9875-2d$.
Good top edges satisfy the source claim with parameter $s_0<0.0025$,
and save at least $w_0$. Bottom edges save at least $c_{0.0025}>w_0$;
bad edges save one. The exact final comparison is
\[
w_0(0.9875-2d_0)>\chi(1+d_0).
\]
These ranges share their endpoints and exhaust $z<1-s_0$.

Suppose next that $1-\epsilon_F<U<1-r$, again with $F_u=1$.
Retain all levels with $t_i\geq r+2d$. Their total mass is at least
$U-r-2d$, and every unit saves at least $w_1$.
Since $U>1-\epsilon_F=0.9$, division by $U$ gives saving per unit at least
\[
w_1(1-(r+2d)/0.9).
\]
The exact inequality
\[
w_1(1-(r+2d_0)/0.9)-\chi>0
\]
proves Eq.~\eqref{eq:candidate_ancestor_master} in this case.

For $\epsilon_F\leq U\leq1-\epsilon_F$, first bound the good top
terms without using bad incidence. Let $u':=\mathsf p(u)$, fix a good top edge
$g\in f=(u'',v'')$, and condition on either endpoint event
$\mathcal R_{f,z}$, $z\in\{u'',v''\}$.
The auxiliary thinning coin is independent of the tree. After the
endpoint $u''$ is made a tree, \cite[Fact~2.8]{kko21} factors its internal
tree from the contracted tree. Relative to $E(u')$, the other tree and
boundary-count conditions in the happiness event depend only on the
contracted tree and on $(\delta(u)\cap\delta(u'))_T$; they reveal no other
internal edges of $u'$. Conditioning further on $u'$ being a tree and
factoring once more leaves the tree-conditioned law inside $u'$ used in
\cite[Claim~7.5]{kko21}, possibly conditioned on that boundary count.
The claim is uniform in its value. Thus this is a factorization after
contraction, not independence of overlapping tree events. Together with
\cite[Lemma~2.23]{kko21}, it gives
\[
\Pr[u'\text{ is not a tree}\mid\mathcal R_{f,z}]\leq d/2.
\]
Since $x(\delta(u)\cap\delta(u'))=U\in[\epsilon_F,1-\epsilon_F]$,
\cite[Claim~7.5]{kko21}, uniformly in the boundary count, gives
\[
\Pr[\delta(u)_T\text{ odd}\mid\mathcal R_{f,z}]
\leq1-\epsilon_F+\epsilon_F^2+d/2
\leq1-\epsilon_F+2\epsilon_F^2.
\]
Thus each good top edge saves at least
$\epsilon_F-2\epsilon_F^2=0.08$ per unit, and every bad edge saves
one. If $b(u)=0$, then $F_u=1$ and the universal bottom coefficient
suffices, because $0<c<0.08$ and
\[
c-\chi>0.
\]
If $b(u)=1$, use Lemma~\ref{lem:candidate_bad_bundle_parity} for every
upward bottom edge. Its polygon parent contains both endpoints of the
internal bad bundle, as required by that lemma. The exact comparisons are
\[
c_{\rm bad}-(\chi+\epsilon_B-\chi\epsilon_B)>0,
\]
\[
0.08-(\chi+\epsilon_B-\chi\epsilon_B)>0.
\]
Thus every unit supplies at least $\chi+\epsilon_B-\chi\epsilon_B$,
which proves the discounted bound. This also suffices at either boundary
of the fractional interval. Only this case uses the improved parity
estimate; all retained-level and small-mass arguments use $c$.

Finally, let $\sigma\leq U<\epsilon_F$, so $F_u=1$.
Condition on the parent of $u$ as above and apply
\cite[Claim~7.5]{kko21} with its actual parameter $U$.
The conditional failure of the parent-tree event contributes at most
$d_0/2$, so a good top edge saves at least $U-U^2-d_0/2$.
This expression is increasing on $[\sigma,\epsilon_F]$, and
\[
\sigma^2-2d_0>0,
\quad
\sigma-\sigma^2-d_0/2-\chi>0.
\]
A bad edge saves one unit and a bottom edge saves $c>\chi$.
Every unit therefore supplies the required saving $\chi$, completing all
cases of Eq.~\eqref{eq:candidate_ancestor_master}.
\end{proof}

To complete the bottom-edge accounting, we also need bounds on horizontal charges
within a polygon parent. Boundary and interior atoms have different conditional
structures, so we estimate them separately.

\begin{lemma}[Polygon charges on the enlarged domain]
\label{lem:polygon_charge_extended}
Use Definition~\ref{def:final24_parameters} and $0\leq d\leq d_0$.
For a polygon cut $S$ with polygon parent $\widehat S$, the expected
horizontal charge is at most $0.31P\beta$ if $S$ is a boundary atom,
and at most $0.85P\beta$ if it is an interior atom.
\end{lemma}

\begin{proof}
First suppose $S$ is a boundary atom. Work conditional on the accepted
parent event, including its independent thinning. Let $A,B,C$ be the
polygon partition of $\delta(S)$; the arrows refer to $\widehat S$. Put
\[
a:=x(A^\rightarrow),\quad b:=x(B^\rightarrow),\quad
a_\uparrow:=x(A^\uparrow),\quad b_\uparrow:=x(B^\uparrow).
\]
The internal parent tree is independent of the accepted external
boundary pair. The external part of $\delta(S)$ meets exactly one
parent boundary side, so it contains one selected edge. The separate
boundary marginal budget in Lemma~\ref{lem:candidate_polygon_flow},
the parent-tree change, and remainder deletion show that its
$A^\uparrow$ and $B^\uparrow$ probabilities differ from the original
masses by at most $e_M:=\epsilon_M+2d$.
This is the internal/external factorization used in
\cite[Lemma~7.10]{kko21}; no strong-Rayleigh property is asserted
for the entire accepted law.

Let $H_S$ count the internal parent-tree edges leaving $S$.
Connectivity gives $H_S\geq1$. The hierarchy bounds and the
parent-tree rank deficit give
\[
x(\delta^\rightarrow(S))\leq1+2d,\quad
\E[H_S]\leq1+5d/2.
\]
Consequently
\[
\E[H_S\mathbf1_{\{H_S\geq2\}}]
\leq2(\E[H_S]-1)\leq5d.
\]
Internal marginals increase on conditioning the parent to be a tree.
Thus the probabilities that the unique internal edge belongs to
$A^\rightarrow$ and $B^\rightarrow$ are at least
$(a-5d)_+$ and $(b-5d)_+$, respectively, where $z_+:=\max\{z,0\}$.
Opposite internal and external single edges make $S$ happy.
Independence and $(u-e)_+(v-f)_+\geq uv-ev-fu$ for nonnegative
$u,v,e,f$ therefore give
\[
\Pr[S\text{ happy}]
\geq a_\uparrow b+b_\uparrow a
-e_M(a+b)-5d(a_\uparrow+b_\uparrow).
\]
The polygon-side lower bounds imply
\[
a_\uparrow b+b_\uparrow a\geq(1-d)(a+b)-2ab.
\]
Also $1-2d\leq a+b\leq1+2d$, $0\leq a\leq1+2d$, and
$a_\uparrow+b_\uparrow\leq1+d$. Writing $a+b=1+t$, $|t|\leq2d$,
gives
\[
(1-d)(a+b)-2ab
\geq a^2+(1-a)^2-3d-10d^2.
\]
Substitution yields
\[
\begin{aligned}
\Pr[S\text{ happy}]
&\geq a^2+(1-a)^2-\epsilon_M-10d-2d\epsilon_M-19d^2\\
&\geq a^2+(1-a)^2-2\epsilon_M-20d.
\end{aligned}
\]
The last step uses
$\epsilon_M(1-2d)+d(10-19d)\geq0$, valid throughout
$0\leq d\leq d_0=2\cdot10^{-8}$ and $\epsilon_M<0.005$.
Assume $a=\max\{a,b\}$. Taking expectations in the pointwise payment formula bounds the
horizontal charge by
\[
(1+d)P\beta(a\Pr[S\text{ not happy}]+x(C^\rightarrow)).
\]
Since $x(C^\rightarrow)\leq d$ and $2a^2(1-a)\leq8/27$ for $a\geq0$,
this is at most $J_\partial P\beta$, where
\[
J_\partial:=(1+d_0)(8/27+(2\epsilon_M+20d_0)(1+2d_0)+d_0)
<0.31.
\]
This proves the boundary estimate directly on the new domain, without
invoking the parameter-restricted conclusion of \cite[Lemma~7.9]{kko21}.

For an interior atom, its upward boundary lies in the parent remainder
and is deleted by the accepted parent event. Thus only the internal
parent-tree factor remains; the size of $\epsilon_M$ is irrelevant.
Condition also on $C_T=0$, whose failure probability is at most $2d$.
The unrounded calculation in \cite[Lemma~7.11]{kko21} gives both
remaining side means in $[1-2d,1+3d]$. Their nonzero tails are at least
$1-e^{-1+2d}$ and their tails at most one exceed $0.495$ by Markov.
The up/down split therefore gives central probability at least
$0.155$ at total rank two. The three link failures cost at most
$12d$, so the internal happiness probability is at least
$0.155(1-12d)>0.153$. Restoring deletion gives
$\Pr[S\text{ happy}]>0.153(1-2d)>0.152$.
The charge is consequently at most
\[
(1+d)P\beta(1-0.152)(1+2d)<0.85P\beta.
\]
All scalar comparisons hold at $d=d_0$ and hence throughout the interval.
\end{proof}

\subsection{The payment theorem}
\label{sec:payment_theorem}

The preceding probability, matching, and charge estimates can now be assembled
into the payment vector. We retain separate expected decreases for good top edges
and bottom edges because the repair step uses their different amounts of available
credit.

\begin{lemma}[Payment theorem]
\label{lem:candidate_payment_replay}
\label{lem:final24_payment}
Use Definitions~\ref{def:final24_parameters} and~\ref{def:final24_parity}.
Fix a hierarchy with error $0\leq d\leq d_0$ and let $\beta>0$.
Set $\tau:=t\beta$, take the common events from
Lemma~\ref{lem:candidate_common_probability}, and use the matching from
Lemma~\ref{lem:candidate_matching}.
Use the KKO reduction and increase formulas in \cite[Section~7]{kko21}. Their
expected-decrease coefficient is
\[
a=\zeta rpt>1.37147\cdot10^{-8}.
\]
The bottom coefficient is
\[
a_{\rm bot}>3.05\cdot10^{-8}.
\]
More precisely, the construction produces a partition
$E=E_{\rm g}\mathbin{\dot\cup}E_{\rm b}$ and a payment vector
$s^{\rm pay}$ satisfying the following:
\begin{enumerate}[label=(\alph*)]
\item \label{item:candidate_payment_bottom} every bottom edge is in
$E_{\rm g}$;
\item \label{item:candidate_payment_coordinate}
$s^{\rm pay}_e\geq-\beta x_e$ on $E_{\rm g}$ and
$s^{\rm pay}_e=0$ on $E_{\rm b}$;
\item \label{item:candidate_payment_expectation}
$\E[s^{\rm pay}_e]\leq-a\beta x_e$ on $E_{\rm g}$, while every bottom edge
satisfies $\E[s^{\rm pay}_e]\leq-a_{\rm bot}\beta x_e$;
\item \label{item:candidate_payment_cut} the deterministic cut inequalities
in Lemma~\ref{lem:payment_deterministic} hold.
\end{enumerate}
\end{lemma}

\begin{proof}
Use Lemma~\ref{lem:payment_matching_setup} to thin the top events to
$p$ and the polygon events to $P$, preserving the synchronized pairs.
The raw probability bounds come from
Lemma~\ref{lem:candidate_common_probability}.
Use the matching from Lemma~\ref{lem:candidate_matching}.

The reductions themselves retain their pointwise formulas: a bottom
edge has reduction $\beta x_e$ times its polygon indicator, and a top
edge in bundle $\mathbf e=(u,v)$ has reduction
$\tau x_e(\mathcal R_{\mathbf e,u}+\mathcal R_{\mathbf e,v})/2$.
Their expected reductions are therefore $P\beta x_e$ and $p\tau x_e$,
respectively. The payment increases depend on these realized reductions
and the matching, not on the indicators' common frequency. Thus the
change from $p$ to $P$ for polygon events changes neither the support
nor the coordinate bound, and leaves the pointwise proof of
Lemma~\ref{lem:payment_deterministic} applicable.
Since $0\leq d\leq d_0<1$, the deterministic polygon calculation also has
$(1+d)(1-d/2)\geq1$. Thus
Items~\ref{item:candidate_payment_bottom},
\ref{item:candidate_payment_coordinate}, and
\ref{item:candidate_payment_cut} follow from the pointwise construction
and Lemma~\ref{lem:payment_deterministic}.
We verify the expectation inequalities below.

The present parameters satisfy
\[
h<0.0055,\quad d_0<h^2,
\]
\[
\epsilon_M<0.005,\quad d_0<\epsilon_M^2,
\]
as well as $q_0<t=1$ and $\tau=\beta$.
Lemma~\ref{lem:polygon_half_budget} gives the conditional polygon
estimate
\[
\Pr[u\text{ not left happy}\mid E_S]\leq q_{\rm unh}<t/s,
\]
with the same bound on the right. Its proof uses the half marginal
budget for the side count and coordinate domination for the small
remainder. The exact strict margin is $t-sq_{\rm unh}>0$.
An upward bottom edge therefore has expected unhappy reduction at most
$P\beta q_{\rm unh}x_e<p\tau x_e$.

For completeness, the dependencies in the payment argument are as follows.
The separate polygon and top probabilities and the synchronized top-event
identities come from
Lemmas~\ref{lem:candidate_common_probability}
and~\ref{lem:payment_matching_setup}. The weighted matching capacity and
conservation identities come from Lemma~\ref{lem:candidate_matching},
with fixed $\alpha=2d_0$ and the present $\epsilon_B$.
The large-endpoint burden is supplied by
Lemma~\ref{lem:candidate_ancestor}; the smaller endpoints are treated
explicitly below. The three bottom-parent configurations use
Lemma~\ref{lem:candidate_structural_trichotomy} and the conditional
polygon estimates just proved. The final matching combination and
deterministic cut inequalities use only these identities and the
pointwise support and coordinate bounds. In particular, no old
parameter-dependent payment conclusion or hypothesis $12r\leq h$ is
invoked.

Suppose first that both endpoints of a good top bundle have upward mass
at least $\sigma$. The ancestor bound, matching conservation, and
weighted capacity give expected total burden at most
$(1-\chi)(1+2d_0)p\tau$ times the bundle mass.
Subtracting from its expected reduction, the decrease normalized by
$rp\tau$ times its mass is at least
\[
\frac{1-(1-\chi)(1+2d_0)}r
>\zeta.
\]
If an endpoint has upward mass below $\sigma$ and its degree parent has
at least four atoms, the matching factor $Z_u=2$ halves its crude
burden. Since $1/2<1-\chi$, the same bound applies.
This also covers the case when both endpoints have the factor two.
Zero upward mass causes zero endpoint burden, according to the
zero-mass convention in the construction.

It remains to check a three-atom degree parent, with atoms $u,v,w$,
where $U:=x(\delta^\uparrow(u))<\sigma$.
All three internal bundles are good, as proved in
Lemma~\ref{lem:candidate_matching}. Write
\[
e:=x(E(u,v)),\quad f:=x(E(u,w)),\quad g:=x(E(v,w)),
\]
\[
V:=x(\delta^\uparrow(v)),\quad
W:=x(\delta^\uparrow(w)).
\]
The degree and parent-cut equations give
\begin{equation}
|f-V|\leq d,\quad |e-W|\leq d,\quad |g-U|\leq d.
\label{eq:triangle_identities}
\end{equation}
The nested-cut bound gives $V,W\leq1+d$, so
$e,f\leq1+2d$ and $g\leq\sigma+d$.
Since $e+f+U=x(\delta(u))\geq2$, we also have
\[
e,f\geq1-\sigma-2d,\quad
V,W\geq1-\sigma-3d>1-\epsilon_F.
\]
Thus $F_u=F_v=F_w=1$ and $Z_v=Z_w=1$.
Use the fixed matching constant $D:=1+2d_0$.
Conservation at $v,w$ followed by capacity for the $f$ and $g$ bundles
gives
\[
V+W
=m_{e,v}+m_{g,v}+m_{f,w}+m_{g,w}
\leq m_{e,v}+D(f+g).
\]
Also $V+W\geq e+f-2d$. Therefore
\[
\begin{aligned}
m_{e,v}
&\geq e-(\sigma+3d+2d_0+2d_0\sigma+6d_0d)\\
&\geq e-(\sigma+5d_0+2d_0\sigma+6d_0^2).
\end{aligned}
\]
Recall
\[
\rho_{\rm sm}:=
1-\frac{\sigma+5d_0+2d_0\sigma+6d_0^2}
{1-\sigma-2d_0}.
\]
Then $m_{e,v}/e\geq\rho_{\rm sm}$.
Use the crude burden at $u$ and
Lemma~\ref{lem:candidate_ancestor} at $v$. Weighted matching capacity
bounds the sum of the crude endpoint burdens by $Dp\tau e$,
and the ancestor saving at $v$ is at least
$\chi p\tau m_{e,v}$. The expected net decrease is consequently at least
\[
p\tau(\chi m_{e,v}-2d_0e)
\geq rp\tau e(\vartheta\rho_{\rm sm}-2d_0/r).
\]
The exact margin
\[
\vartheta\rho_{\rm sm}-2d_0/r-\zeta=10^{-9}>0
\]
proves the required bound. Interchanging $v$ and $w$ proves it for the
other bundle incident to the small endpoint. Both endpoints of the
remaining bundle exceed $\sigma$, so the ordinary case already applies.
Thus every good top edge has expected decrease at least
$\zeta rp\tau x_e=a\beta x_e$.

We now bound bottom-edge payments, starting with a polygon cut $S$
whose parent is a degree cut. Let $A,B,C$ be its polygon partition and
$A',B',C'$ its degree partition. The partition bounds give
\[
x(\delta(S))\leq2+d,\quad x(A),x(B)\geq1-d,\quad
x(A'),x(B')\geq1-r,
\]
and $A'\subseteq A$, $B'\subseteq B$ by
\cite[Remark~5.20]{kko21}. The polygon-side lower bounds also hold for a
triangle: if its two non-root atoms are $a,b$, the cut identity gives
\[
x(E(a,\overline S))
=\frac{x(\delta(a))+x(\delta(S))-x(\delta(b))}{2}
\geq1-d/2\geq1-d,
\]
and the other side is symmetric.
Here $\overline S$ denotes the complement of $S$ in the vertex set.

The polygon partition also satisfies $x(C)\leq d$, and the degree
partition satisfies $x(C')\leq2r+d$.
In the paired alternative of
Lemma~\ref{lem:candidate_structural_trichotomy},
Item~\ref{item:candidate_structural_paired}, write
$\mathbf e=(S,v)$ and $\mathbf f=(S,w)$.
They have masses at least $1/2-h$ and cross bounds
$x_{\mathbf e}(B')\leq\omega$, $x_{\mathbf f}(A')\leq\omega$.
For $D\in\{A,B,C\}$, put
\[
e_D:=x_{\mathbf e}(D),\quad f_D:=x_{\mathbf f}(D).
\]
Since $A'\subseteq A$, $B'\subseteq B$, and $A',B',C'$ partition the
boundary, we have $B\setminus B'\subseteq C'$ and
$A\setminus A'\subseteq C'$. Thus
\[
e_A\geq\frac12-h-\omega-2r-d,\quad e_B\leq\omega+2r+d.
\]
The first inequality follows already by retaining the part of
$\mathbf e$ in $A'$. Since
\[
\frac12-h-2\omega-4r-2d_0>0,
\]
we have $e_A>e_B$. Symmetrically, $f_B>f_A$.

Let $R_0$ be the common synchronized indicator
$\mathcal R_{\mathbf e,S}=\mathcal R_{\mathbf f,S}$, and let
$R_1:=\mathcal R_{\mathbf e,v}$ and
$R_2:=\mathcal R_{\mathbf f,w}$. Each has expectation $p$;
no independence is required.
Use the entire paired set
$D:=\mathbf e\mathbin{\dot\cup}\mathbf f$, and write
$\boldsymbol e:=(e_A,e_B,e_C)$ and
$\boldsymbol f:=(f_A,f_B,f_C)$.
For a nonnegative vector $\boldsymbol z=(z_A,z_B,z_C)$, define
\[
N(\boldsymbol z):=\max\{z_A,z_B\}+z_C.
\]
This function is positively homogeneous and subadditive.
The boundary totals of the reduction vector on $D$ are
\[
\boldsymbol r_D
=\frac{\tau}{2}((\boldsymbol e+\boldsymbol f)R_0
+\boldsymbol eR_1+\boldsymbol fR_2).
\]
Here $\boldsymbol r_D$ denotes reductions, not the scalar partition
parameter $r$.
Let $I_S(D)$ be the polygon charge attributable to reductions in $D$,
as defined in \cite[Eq.~(48)]{kko21}. This charge is subadditive over
$D$ and its complement. On $D$, discarding the unhappiness
indicators only increases that charge. Subadditivity of $N$ therefore
gives
\[
\E[I_S(D)]
\leq\frac{(1+d)p\tau}{2}
(N(\boldsymbol e+\boldsymbol f)+N(\boldsymbol e)+N(\boldsymbol f)).
\]
The complementary top edges keep their trivial bound. Compare this display
with the trivial contribution
$(1+d)p\tau(x(\mathbf e)+x(\mathbf f))$.
The opposite-side inequalities imply
\[
N(\boldsymbol e)=x(\mathbf e)-e_B,\quad
N(\boldsymbol f)=x(\mathbf f)-f_A,
\]
\[
N(\boldsymbol e+\boldsymbol f)
=x(\mathbf e)+x(\mathbf f)-\min\{e_A+f_A,e_B+f_B\}.
\]
Consequently the normalized saving on the entire paired set is at least
\[
\begin{aligned}
\frac{\min\{e_A+f_A,e_B+f_B\}+e_B+f_A}{2}
&=\frac{\min\{e_A+e_B+2f_A,f_A+f_B+2e_B\}}2\\
&\geq\frac{\min\{x(\mathbf e)-e_C,x(\mathbf f)-f_C\}}2\\
&\geq\frac14-\frac h2-\frac d2.
\end{aligned}
\]
The last line uses $e_C,f_C\leq x(C)\leq d$.

Define
\[
\Delta_{\rm bot}(d):=\frac14-\frac h2-\frac d2.
\]
In the first alternative of \cite[Lemma~7.8]{kko21}, the bad mass is at
least $1/2-h$ and is never reduced. Its saving is therefore at least
$1/2-h>\Delta_{\rm bot}(d)$.
In the second alternative, the $2$--$1$--$1$-good mass is at least
$1/2-h-d$. Its reduction event at $S$ makes $S$ happy, so at most
half the trivial charge remains. Its saving is at least
$(1/2-h-d)/2=\Delta_{\rm bot}(d)$.
Thus the same coefficient applies in all three alternatives.

For the upward boundary, a top edge has expected reduction $p\tau x_e$.
An upward bottom edge in $A\cup B$ has expected charge at most
$(1+d)P\beta q_{\rm unh}x_e<(1+d)p\tau x_e$
by the conditional polygon estimate above.
For the remaining upward edges in $C$, use the crude bound
$(1+d)P\beta x_e$. This also dominates a top edge's crude reduction,
since $P\geq p$ and $\tau\leq\beta$.
Since $x(C)\leq d$ and $d\leq1$, the excess over the $p\tau$ envelope
is at most $2dP\beta=2dsp\beta$. The synchronized saving just proved
involves only top events and retains the factor $p\tau$.
Hence the degree-parent case
satisfies
\[
\E[I_S]\leq(1+d)p\tau(2+d-\Delta_{\rm bot}(d))+2dP\beta.
\]
The three normalized increases for a degree parent, a boundary atom
with polygon parent, and an interior atom are therefore bounded by
\[
J_1(d):=(1+d)t(2+d-\Delta_{\rm bot}(d))+2ds,
\]
\[
J_2(d):=(1+d)(t(1+d)+0.31s)+2ds,
\quad
J_3(d):=(1+d)(td+0.85s)+2ds.
\]

For the last two bounds, Lemma~\ref{lem:polygon_charge_extended}
supplies the horizontal bottom contributions $0.31P\beta$ and
$0.85P\beta$. The upward contributions retain their pointwise
bounds $p\tau(1+d)$ for a boundary atom and $p\tau d$ for an
interior atom. The factor $1+d$ and remainder $2dP\beta$ give
exactly $J_2(d)$ and $J_3(d)$ after division by $p\beta$.
Every horizontal bottom term and $C$-side remainder is therefore
scaled by $s$; no equality $P=p$ is used.
Each $J_i(d)$ is increasing on $[0,d_0]$. Define
\[
\Delta_{\rm bot}:=\Delta_{\rm bot}(d_0)=1/4-h/2-d_0/2,
\quad
J_i:=J_i(d_0)\quad(1\leq i\leq3).
\]
Direct rational substitution gives
\[
J_1-J_2>0,\quad
J_1-J_3>0.
\]
Every normalized increase is therefore at most $J_1$.
Since each bottom edge has expected reduction $P\beta x_e=sp\beta x_e$, define
\[
a_{\rm bot}:=p(s-J_1)
=p(s-(1+d_0)t(2+d_0-\Delta_{\rm bot})-2d_0s).
\]
The exact rational evaluation gives
\[
a_{\rm bot}>3.05\cdot10^{-8},
\]
which also shows $a_{\rm bot}>a$.
This proves Item~\ref{item:candidate_payment_expectation} on every edge.
All numerical comparisons in the proof use rational quantities with
positive denominators; in particular
$1-\sigma-2d_0>0$, $1-r>0$, and $2q_0-1>0$.
They therefore follow by clearing denominators and comparing integers.
\end{proof}

The expectation bounds are only one part of the payment guarantee. We separately
record the pointwise cut inequalities used when adding the repair vectors.

\begin{lemma}[Deterministic payment inequalities]
\label{lem:payment_deterministic}
Fix a KKO hierarchy with uniform error envelope $0\leq d\leq1$ and a parameter
$\beta>0$, and let $s^{\mathrm{pay}}$ be the payment vector from the
synchronized construction for these data.
For a polygon cut $P$, including a triangle cut viewed as a degenerate
polygon with $C=\varnothing$, let $A,B,C$ be its polygon partition and call a set $F$ admissible
for $P$ if $p(e)=P$ for every $e\in F$ and $x(F)\geq1-d/2$.
\begin{enumerate}[label=(\alph*)]
\item \label{item:deterministic_payment_polygon} if $P$ is not left happy, then for every admissible edge set $F$ for $P$,
$s^{\mathrm{pay}}(A)+s^{\mathrm{pay}}(F)+(s^{\mathrm{pay}})^-(C)\geq0$;
the symmetric inequality holds when $P$ is not right happy;
\item \label{item:deterministic_payment_degree} if $S$ has a degree-cut parent and $\delta(S)_T$ is odd, then
$s^{\mathrm{pay}}(\delta(S))\geq0$.
\end{enumerate}
\end{lemma}

\begin{proof}
The proof uses no numerical payment estimate.  Use the zero-mass convention
fixed after Definition~\ref{def:final24_parameters}. If a
polygon cut $P$ (possibly a triangle) with partition $A,B,C$ is not left happy, then its reduction
event does not occur and the reduction is zero on every admissible set $F$.
Writing $r$ and $I_P$ for the unchanged KKO reduction and increase vectors,
the definition of $I_P$ gives
\[
s^{\mathrm{pay}}(A)+s^{\mathrm{pay}}(F)+(s^{\mathrm{pay}})^-(C)
\geq-r(A)+(1+d)(r(A)+r(C))(1-d/2)-r(C)\geq0.
\]
The right-happy inequality is symmetric. If $S$ has a degree-cut parent and
$\delta(S)_T$ is odd, the horizontal reductions vanish.  If
$x(\delta^\uparrow(S))=0$, there is no upward reduction to compensate and
all remaining increases are nonnegative, so the conclusion is immediate.
Otherwise, matching conservation gives
$\sum_{f\in\delta^\rightarrow(S)}m_{f,S}
=Z_Sx(\delta^\uparrow(S))>0$, and hence
\[
s^{\mathrm{pay}}(\delta(S))
\geq
-\sum_{g\in\delta^\uparrow(S)}r_g
+
\sum_{e\in\delta^\rightarrow(S)}
\sum_{g\in\delta^\uparrow(S)}
r_g
\frac{m_{e,S}}{\sum_{f\in\delta^\rightarrow(S)}m_{f,S}}
=0.
\]
These are precisely \cite[Theorem~4.33(iii)--(iv)]{kko21}. This proves both conclusions.
\end{proof}

\subsection{Hierarchy and separated repairs}
\label{sec:repair_combined_estimates}

We next check that the hierarchy and repair constructions apply throughout the
threshold range used in the final layering. The statement below summarizes the
common domain and the two repair cost coefficients.

\begin{lemma}[Hierarchy and repair on the required domain]
\label{lem:final24_hierarchy}
For $0<\eta\leq1.4\cdot10^{-9}$, the hierarchy construction in the
preliminaries remains valid, with added cuts $7\eta$-near-minimum and
the same exhaustive one-sided and two-sided classification.
Under the deterministic payment hypotheses of
Lemma~\ref{lem:prelim_hat_pair}, its separated repair conclusions
remain valid on this domain, with
\[
R_1(\eta):=12\frac{2+\eta}{1-7\eta},\quad
R_2(\eta):=5\frac{2+\eta}{1-\eta}.
\]
\end{lemma}

\begin{proof}
We verify the hierarchy domain explicitly.
The construction in \cite[Appendix~B]{kko22} retains singleton
one-sided components and, for each non-singleton component, inserts
its non-root atoms and their union. The one-sided component has no
inside atoms. Its polygon structure gives adjacent-bundle mass at least
$1-7\eta$, atom-boundary mass at most $2+7\eta$, and remainder
mass at most $7\eta$ \cite[Theorem~6.2]{kko22}.
The inserted union is the complement of the root atom and has the
same boundary. The geometric facts giving laminarity and the
valid-hierarchy conclusion of \cite[Fact~B.4]{kko22} do not use
$\eta\leq10^{-12}$: their polygon domain is $\eta\leq1/10$.
The two-child convention gives the same triangle parents. Thus the
construction and its cut classification apply on the asserted domain.

Separate this geometric construction from the source's payment theorem.
At the present larger threshold, $7\eta$ can exceed its printed
payment-error cutoff $10^{-10}$, so that theorem is not invoked.
Instead use the probability, matching, ancestor, and payment estimates
proved above, with conservative hierarchy error $14\eta$. Their domain
holds because
\[
14(1.4\cdot10^{-9})=1.96\cdot10^{-8}<d_0.
\]
The geometric hierarchy and deterministic cut classification are
unchanged. In particular, no extension of a parameter-specialized
source payment conclusion is assumed.

The component repair lemmas allow $\eta\leq1/10$.
Lemma~\ref{lem:prelim_hat_pair} combines them on the required domain;
their amplitudes give exactly the displayed coefficients $R_1,R_2$.
\end{proof}

The polygon representation may contain inside atoms, whereas the congestion
argument needs a boundary order. We first show how to represent the memberships of
finitely many cuts by outside atoms.

\begin{lemma}[Outside witnesses for finitely many cuts]
\label{lem:repair_outside_witness}
Let $P$ represent a non-singleton connected component of
$\eta$-near-minimum cuts, where $0<\eta\leq2/5$. If
$C_1,\ldots,C_k$ are cuts in this component and $1\leq k<1/\eta$,
then every atom $a$ has an outside atom $a'$ such that
\[
a\subseteq C_i\quad\Longleftrightarrow\quad a'\subseteq C_i
\quad(1\leq i\leq k).
\]
This includes the root atom, which need not be an outside atom.
\end{lemma}

\begin{proof}
Apply the polygon representation to $x^0$, including $e_0$.
This vector is fractionally $2$-edge-connected. Every component cut
avoids the root atom containing both endpoints of $e_0$, so its
capacity under $x^0$ equals its capacity under $x$.

For each representing diagonal choose the half-polygon containing the
cell of $a$. Their intersection contains the whole cell and therefore
has positive area. If $a$ is outside, take $a'=a$. Otherwise, if the
intersection contained no outside atom, \cite[Lemma~4.30]{kko22} would
imply that it contained no inside atom either, a contradiction.
An outside atom in the intersection has exactly the chosen membership
pattern. Atom cells have positive area and lie wholly on the prescribed
sides of the representing diagonals; no witness is placed at a diagonal
intersection. The same argument applies to the root atom.
\end{proof}

We use these outside witnesses to bound how many directional repair events can
charge a single edge. The key comparison puts the relevant crossing cuts in one
boundary order.

\begin{lemma}[Two directional increase sets per edge]
\label{lem:repair_directional_congestion}
For $0<\eta\leq1/10$, every edge belongs to at most one right-event
increase set and at most one left-event increase set in the KKO
two-sided repair construction.
\end{lemma}

\begin{proof}
By \cite[Fact~4.9]{kko22}, an edge can receive a two-sided increase
in at most one polygon. Fix that polygon, and let $a,b$ be the atoms
containing the endpoints of an edge $e$. Suppose $e$ belongs to the
right-event increase sets at two distinct polygon points $p,q$.
Using the notation of \cite[Section~5]{kko22}, put
\[
L_1=L(p),\quad R_1=L(p)_R,\quad
L_2=L(q),\quad R_2=L(q)_R.
\]
The increase set at $p$ is contained in
$E(L_1\cap R_1,R_1\setminus L_1)$, and the analogous containment
holds at $q$. If an endpoint belongs to both $L_1,L_2$, then
\cite[Lemma~5.4]{kko22} gives a contradiction. Otherwise, after
interchanging $a,b$, the memberships are
\[
\begin{aligned}
a&\subseteq L_1\cap R_1,& b&\subseteq R_1\setminus L_1,\\
b&\subseteq L_2\cap R_2,& a&\subseteq R_2\setminus L_2.
\end{aligned}
\]
All four sets are component cuts, and $4<1/\eta$.
Lemma~\ref{lem:repair_outside_witness} supplies outside atoms
$a',b',r'$ matching, on all four cuts, the respective membership
patterns of $a,b$ and the actual root. In particular, $r'$ belongs
to none of the four cuts. The three witnesses are distinct because
their membership patterns differ.

Start the counterclockwise order of outside atoms immediately after
$r'$. All four cuts are nonwrapping intervals in this single linear
order. This change of origin does not change a crossing orientation:
the union of a crossing pair avoids $r'$, and rotating the origin
within its complement preserves the order of its ends. Since $R_1$
crosses $L_1$ on the right, every outside atom of $L_1\cap R_1$
precedes every outside atom of $R_1\setminus L_1$, giving $a'<b'$.
The right crossing of $R_2$ and $L_2$ gives $b'<a'$, a contradiction.
Thus at most one right-event increase set contains $e$. Reversing the
polygon orientation proves the left-event assertion. Only these four
membership patterns have been transferred to outside atoms; the polygon
may still have inside atoms.
\end{proof}

We retain the KKO repair vector and its pointwise guarantee. Exact
cut-excess accounting bounds each directional bad event, and the
preceding lemma bounds how many such events can increase an edge.

\begin{lemma}[Sharpened two-sided KKO repair]
\label{lem:sharpened_two_sided_repair}
Let $x^0$ be a feasible subtour-LP solution with support
$E\cup\{e_0\}$, let $x$ be its restriction to $E$, and let $\mu$ be any
spanning-tree distribution with marginals $x$. Let $0<\eta\leq1/10$.
For every $\alpha>0$, there is a random vector
$r^{(2)}:E\to\mathbb R_{\geq0}$ such that:
\begin{enumerate}[label=(\alph*)]
\item \label{item:sharpened_two_sided_cut} if an $\eta$-near-minimum cut $S$
is crossed on both sides and $\delta(S)_T$ is odd, then
$r^{(2)}(\delta(S))\geq\alpha(1-\eta)$;
\item \label{item:sharpened_two_sided_expectation} for every edge $e$,
$\E[r^{(2)}_e]\leq5\alpha\eta x_e$.
\end{enumerate}
\end{lemma}

\begin{proof}
We use the KKO construction from \cite[Section~5]{kko22}. Fix a cut $L$
crossed on both sides, and let $Q=L^L$ and $R=L^R$ be the left and right
crossing cuts chosen by that construction. Define
\[
A:=L\cap R,
\quad
B:=R\setminus L,
\quad
C:=L\cap Q,
\quad
D:=Q\setminus L,
\]
and, for every proper cut side $X$, define
\[
\epsilon_X:=x(\delta(X))-2.
\]
The cuts $L,Q,R$ are $\eta$-near-minimum cuts. The four cut sides
$A,B,C,D$ are proper and avoid the endpoints of the added edge; in
particular,
\[
\epsilon_A,\epsilon_B,\epsilon_C,\epsilon_D\geq0.
\]
They are also $2\eta$-near-minimum cuts by
\cite[Lemma~2.7]{kko22}.

The KKO edge sets partition $\delta(L)$ as
\[
\delta(L)=E^{\leftarrow}(L)\mathbin{\dot\cup}
E^{\rightarrow}(L)\mathbin{\dot\cup}E^{\circ}(L),
\]
where
\[
E^{\rightarrow}(L)=E(A,B),
\quad
E^{\leftarrow}(L)=E(C,D).
\]
The disjointness is also recorded in
\cite[Corollary~5.8]{kko22}. For disjoint cut sides $X,Y$, the cut identity
gives
\[
x(E(X,Y))
=1+\frac{\epsilon_X+\epsilon_Y-\epsilon_{X\cup Y}}{2}.
\]
Consequently,
\[
x(E^{\rightarrow}(L))
=1+\frac{\epsilon_A+\epsilon_B-\epsilon_R}{2},
\quad
x(E^{\leftarrow}(L))
=1+\frac{\epsilon_C+\epsilon_D-\epsilon_Q}{2},
\]
and hence
\[
x(E^{\circ}(L))
=\epsilon_L+
\frac{\epsilon_Q+\epsilon_R-\epsilon_A-\epsilon_B-\epsilon_C-\epsilon_D}{2}.
\]

The right bad event is
\[
\mathcal B^{\rightarrow}(L)
:=
\{|E^{\rightarrow}(L)\cap T|\neq1
\text{ or }|E^{\circ}(L)\cap T|\neq0\}.
\]
By \cite[Corollary~2.12]{kko22},
\[
\Pr[|E^{\rightarrow}(L)\cap T|\neq1]
\leq\frac{\epsilon_A+\epsilon_B+\epsilon_R}{2}.
\]
By \cite[Fact~2.13]{kko22},
\[
\Pr[|E^{\circ}(L)\cap T|\neq0]
\leq x(E^{\circ}(L)).
\]
Adding these bounds and substituting the exact expression for
$x(E^{\circ}(L))$ gives
\[
\Pr[\mathcal B^{\rightarrow}(L)]
\leq
\epsilon_L+\epsilon_R+
\frac{\epsilon_Q-\epsilon_C-\epsilon_D}{2}
\leq\frac52\eta.
\]
Here we used
$\epsilon_L,\epsilon_Q,\epsilon_R<\eta$ and
$\epsilon_C,\epsilon_D\geq0$. The symmetric calculation gives
\[
\Pr[\mathcal B^{\leftarrow}(L)]\leq\frac52\eta.
\]

Keep the increase sets of \cite[Theorem~5.2]{kko22}, and set
$r^{(2)}_e=\alpha x_e$ if any directional bad event whose increase set
contains $e$ occurs, and $r^{(2)}_e=0$ otherwise.
Lemma~5.3 there triggers the relevant event, Lemma~5.1 places its
increase set inside the cut boundary, and Lemma~5.6 supplies mass at
least $1-\eta$. These facts give \ref{item:sharpened_two_sided_cut}.

By Lemma~\ref{lem:repair_directional_congestion}, at most two
directional bad events can increase any fixed edge. Thus, without
any independence assumption, the union bound gives
\[
\E[r^{(2)}_e]
\leq2\frac52\eta\alpha x_e
=5\alpha\eta x_e.
\]
This proves \ref{item:sharpened_two_sided_expectation}.
\end{proof}

For one-sided cuts, exact cut-excess accounting and a fractional allocation
to the two boundary bundles give a stronger repair estimate.

\begin{lemma}[Fractional endpoint repair for one-sided cuts]
\label{lem:sharpened_one_sided_repair}
\label{lem:fractional_endpoint_one_sided_repair}
In the setting of \cite[Theorem~A.12]{kko22}, for
$0<\eta\leq1/10$ and $\alpha>0$, there is a random repair vector
$s^*:E\to\mathbb R_{\geq0}$ that has all
the pointwise cut guarantees of that theorem and satisfies
\[
\E[s^*_e]\leq12\alpha\eta x_e
\]
for every edge $e$. This also covers the triangle construction of
\cite[Lemma~A.13]{kko22}. Both constructions remain
supported only on internal bundles between consecutive non-root atoms,
and hence only on bottom edges.
\end{lemma}

\begin{proof}
For every nonempty $X\subseteq V\setminus\{u_0,v_0\}$, put
$\epsilon_X:=x(\delta(X))-2\geq0$. The degree equations and the
forest identity give
\[
\Pr[T[X]\text{ is not a tree}]
\leq |X|-1-\E[|T\cap E(X)|]
=\epsilon_X/2.
\]
Let $S$ be such a cut side and let $E_j=E(X_j,Y_j)\subseteq\delta(S)$,
$j\in\{1,2\}$, be disjoint edge sets. Suppose that $X_j,Y_j$ are
disjoint nonempty subsets of $V\setminus\{u_0,v_0\}$, and put
$P_j=X_j\cup Y_j$. If $T[X_j],T[Y_j],T[P_j]$ are all trees, then
$|(E_j)_T|=1$. Consequently,
\[
\Pr[|(E_j)_T|\ne1]
\leq(\epsilon_{X_j}+\epsilon_{Y_j}+\epsilon_{P_j})/2,
\]
while the cut identity gives
\[
x(E_j)=1+(\epsilon_{X_j}+\epsilon_{Y_j}-\epsilon_{P_j})/2.
\]
Writing $D=\delta(S)\setminus(E_1\cup E_2)$, the event
$|(E_1)_T|=|(E_2)_T|=1$ and $D_T=0$ implies $\delta(S)_T=2$.
The union bound and Markov's inequality therefore give the cancellation
\begin{align}
\Pr[\delta(S)_T\ne2]
&\leq\sum_{j=1}^2
(\epsilon_{X_j}+\epsilon_{Y_j}+\epsilon_{P_j})/2+x(D)\notag\\*
&=\epsilon_S+\epsilon_{P_1}+\epsilon_{P_2}.
\label{eq:one_sided_two_link_accounting}
\end{align}
No independence assumption is used.

Fix a polygon from a non-singleton connected component of one-sided
crossing cuts, with atoms $a_0,\ldots,a_{m-1}$ and root atom $a_0$.
Such a polygon has no inside atoms, as in
\cite[Appendix~A]{kko22}. Write an interval as $[\ell,r)$ when it
contains precisely $a_\ell,\ldots,a_{r-1}$. The component cuts
partition into the laminar families $\mathcal L,\mathcal R$ of
\cite[Definition~4.10 and Fact~4.11]{kko21}: members of $\mathcal L$
are crossed only on the right, and members of $\mathcal R$ only on
the left. Hence an $\mathcal L$-cut has $r<m$ and an
$\mathcal R$-cut has $\ell>1$. Every arc is an endpoint of a
component cut, by the coarsest-atom definition. No component cut is
an atom or the union of all non-root atoms, since it crosses another
member of the component. In particular, $m\geq4$.
The relevant sets are its component cuts
and those non-root atoms that are themselves $\eta$-near-minimum cuts.
An interior relevant set triggers when its tree boundary is odd;
a boundary relevant set triggers when it is not happy. We first bound
these trigger probabilities.

Consider an interior non-atom component cut $A$. After exchanging left
and right if necessary, the geometry in
\cite[proof of Lemma~4.26]{kko21} supplies its strict parent $B$ in
the right laminar family and a component cut $C$ crossing both $A,B$
on the left. The two boundary links are
\[
E_1=E(A\cap C,C\setminus A),
\quad
E_2=E(A\setminus C,B\setminus(A\cup C)).
\]
Their sides are nonempty and avoid the distinguished endpoints, and
the links are disjoint subsets of $\delta(A)$. Their union sides
are $P_1=C$ and $P_2=B\setminus C$. Since $A,B,C$ are
$\eta$-near-minimum cuts, uncrossing $B,C$ gives
$\epsilon_{B\setminus C}\leq\epsilon_B+\epsilon_C\leq2\eta$.
Thus Eq.~\eqref{eq:one_sided_two_link_accounting} yields
\[
\Pr[\delta(A)_T\ne2]\leq4\eta.
\]

We need the intermediate intervals in the proof of
\cite[Lemma~4.18]{kko21}, rather than their intersection. Here the
input cut excess is $\eta$, half the normalization in that source.
For adjacent non-root atoms $a_j,a_{j+1}$, there is an interval
ending at $a_{j+1}$ and containing $a_j$, with excess at most
$3\eta$. Reflection gives an interval starting at $a_j$ and
containing $a_{j+1}$ with the same bound. We record the construction
and its sharper endpoint-dependent costs.

Choose a component cut $A$ with endpoint $j+2$. If it ends there,
it already contains the pair and has excess at most $\eta$.
Otherwise it starts there. If $A\in\mathcal L$, let $L$ be its
strict parent. Such parents exist away from the boundary, and a
cut in the opposite family crosses both child and strict parent,
by \cite[Lemmas~4.12 and~4.14]{kko21}. If $\ell(L)\leq j$, choose
$R\in\mathcal R$ crossing both $A,L$; then $L\setminus(A\cup R)$
is the desired interval. If $\ell(L)=j+1$, take the strict parent
$L'$ of $L$, so $\ell(L')\leq j$. Let $R$ cross $A,L$, and let
$R'$ cross $L,L'$, on the right. The intervals $R,R'$ overlap and
are nested. Their larger member $Q$ crosses $L'$ on the right,
intersects $A$, and starts at or after $j+2$. Thus
$L'\setminus(A\cup Q)$ is the desired interval. In both cases,
uncrossing first bounds the union excess by $2\eta$, and then
the difference excess by $3\eta$.

If $A\in\mathcal R$, choose $L\in\mathcal L$ crossing it on the
left. When $\ell(L)\leq j$, use $L\setminus A$. Otherwise
$\ell(L)=j+1$; let $L'$ be its strict parent. If $L'$ crosses
$A$, use $L'\setminus A$. If $L'$ contains $A$, choose
$R\in\mathcal R$ crossing $L,L'$ on the right. Its start is at
least $j+2$ and cannot be greater, since then it would cross $A$
on the right. Thus $L'\setminus R$ is the desired interval.
These three differences have excess at most $2\eta$. All sets
used avoid the root, and every uncrossing has the required nonempty
regions. This also proves the reflected construction.

Put $I:=\{2,\ldots,m-1\}$. For $i\in I$, define $r_i,\ell_i$
by the following exhaustive table of component-cut endpoints.
\begin{center}
\begin{tabular}{l c c}
\toprule
Endpoints at $i$ & $r_i$ & $\ell_i$\\
\midrule
Both starting and ending cuts & 1 & 1\\
Only ending cuts, including an $\mathcal L$-cut & 2 & 1\\
Only ending cuts, all in $\mathcal R$ & 3 & 1\\
Only starting cuts, including an $\mathcal R$-cut & 1 & 2\\
Only starting cuts, all in $\mathcal L$ & 1 & 3\\
\bottomrule
\end{tabular}
\end{center}
The endpoint property makes these definitions exhaustive and gives
$r_i+\ell_i\leq4$. Also $r_2=\ell_{m-1}=1$, since a cut ending
at $2$ or starting at $m-1$ would be an atom. The constructions
above give an interval starting at $a_i$ and containing $a_{i+1}$
of excess at most $r_i\eta$ for $2\leq i\leq m-2$, and an
interval ending at $a_{i-1}$ and containing $a_{i-2}$ of excess
at most $\ell_i\eta$ for $3\leq i\leq m-1$.
The values $r_{m-1},\ell_2$ are only bookkeeping values.

For a relevant interior atom $S=a_i$, choose these intervals $L,R$
ending and starting at $S$. They intersect exactly in $S$; put
$X=L\setminus S$ and $Y=R\setminus S$. These are nonempty and
disjoint. The two links $E(S,X),E(S,Y)$ contain one tree edge each
whenever the five induced subgraphs on $S,X,Y,L,R$ are trees.
Count the event for $S$ only once. With
$D:=\delta(S)\setminus(E(S,X)\cup E(S,Y))$, the forest bound and
the cut identities yield
\begin{align*}
\Pr[\delta(S)_T\ne2]
&\leq(\epsilon_S+\epsilon_X+\epsilon_Y+\epsilon_L+\epsilon_R)/2+x(D)\\
&=\epsilon_S/2+\epsilon_L+\epsilon_R
\leq(1/2+r_i+\ell_{i+1})\eta.
\end{align*}
In particular, the difference-set excesses cancel; no bound on them
or independence of the forest events is needed.

For a leftmost cut or atom $S$, define
$F=E(S,V\setminus(S\cup a_0))$ and $G=E(S,a_0)$.
These sets partition $\delta(S)$, and $S$ is happy exactly when
$F_T=1$. For a link $E(X,Y)\subseteq F$ with disjoint nonempty sides
avoiding $u_0,v_0$, put $P=X\cup Y$. The same forest bound and cut
identity give
\begin{align*}
\Pr[F_T\ne1]
&\leq(\epsilon_X+\epsilon_Y+\epsilon_P)/2+x(F)-x(E(X,Y))\\
&=1+\epsilon_S+\epsilon_P-x(G).
\end{align*}
For a leftmost component cut $L$, let $L_R$ be the crossing cut used
in \cite[proof of Lemma~4.30]{kko21} and use the link
$E(L_R\cap L,L_R\setminus L)$. Here
$\epsilon_L,\epsilon_{L_R}\leq\eta$, while
\[
x(G)\geq x(E(a_1,a_0))\geq1-\eta
\]
by \cite[Lemma~4.17]{kko21}, with the same parameter halving.
Hence $\Pr[L\text{ is not happy}]\leq3\eta$.
For a relevant boundary atom $a_1$, use $X=a_1$, $Y=R\setminus a_1$,
where $R$ starts at $a_1$, contains $a_2$, and has excess at most
$3\eta$. The bounds $\epsilon_{a_1}\leq\eta$ and the same
root-neighbor estimate give
$\Pr[a_1\text{ is not happy}]\leq5\eta$.
The rightmost cases are symmetric. No forest bound is applied to
the root atom $a_0$.

Let $\mathcal H_S$ denote the trigger for a relevant set $S$.
Assign demand $c_S=4$ to every component cut,
$c_{a_i}=1/2+r_i+\ell_{i+1}$ to every relevant interior atom,
and $c_S=5$ to every relevant boundary atom.
The preceding estimates give $\Pr[\mathcal H_S]\leq c_S\eta$;
using demand $4$ only weakens the boundary component-cut bound $3$.

The internal arc $i\in I$ joins $a_{i-1}$ to $a_i$.
For a relevant interval $S=a_{\ell(S)}\cup\cdots\cup a_{r(S)-1}$,
define its allowed internal endpoints by
\[
D(S):=\{\ell(S),r(S)\}\cap I.
\]
Every $D(S)$ is nonempty. Indeed, a component cut crosses another
member of its connected component, so it is neither an atom nor the
union of all non-root atoms. The absence of inside atoms also makes
each component cut uniquely determined by its endpoints.
We claim the weighted Hall inequality
\begin{equation}
\sum_{S:D(S)\subseteq W}c_S\leq12|W|
\label{eq:fractional_endpoint_hall}
\end{equation}
for every $W\subseteq I$. Select the component cuts with
$D(S)\subseteq W$. First remove the set $J$ of vertices of $W$
incident to no selected component cut. This leaves the selected cuts
unchanged. At an interior vertex $j$, the sum of the two possible
incident atom demands is at most
\[
1+r_{j-1}+\ell_j+r_j+\ell_{j+1}\leq11.
\]
At $j=2$ it is at most $5+(1/2+r_2+\ell_3)\leq19/2$,
and the other boundary is symmetric. These two boundary vertices
are distinct because $m\geq4$. Thus the removed atom demand is
at most $11|J|$, even when an atom is counted twice. If the remaining
set is empty, the claim follows. Otherwise rename it $W$, put
$n=|W|$, let $q$ count pairs $i,i+1\in W$, and set
\[
b_L=\mathbf1_{\{2\in W\}},\quad
b_R=\mathbf1_{\{m-1\in W\}},\quad b=b_L+b_R.
\]
Every remaining vertex is incident to a selected component cut.
Call a vertex $B$-type if $\ell_i=3$, so only $\mathcal L$-cuts
start there and no cut ends there. Call it $C$-type if $r_i=3$,
so only $\mathcal R$-cuts end there and no cut starts there.
These types are disjoint. Let $p_B,p_C$ be their counts in $W$,
and put $p=p_B+p_C$.

Draw the selected $\mathcal L$-cuts as noncrossing chords on
$W\cup\{1\}$ in convex position, and the selected
$\mathcal R$-cuts on $W\cup\{m\}$. Delete the $C$-type vertices
from the first drawing and the $B$-type vertices from the second.
No component edge is deleted. Both remaining internal vertex sets
are nonempty: if all vertices were $B$-type, no selected
$\mathcal R$-cut could be incident to them, and every selected
$\mathcal L$-cut would need an internal ending endpoint, which a
$B$-type vertex cannot supply. The all-$C$ case is symmetric.

Let $t_L,t_R$ count selected component edges joining consecutive
points in the original compressed order of $W$. These counts do
not change after the respective deletions. To see this for the
$\mathcal R$ drawing, suppose a new adjacency edge $(i,k)$
appears across a nonempty block of deleted $B$-type vertices.
For such a vertex $j$, any component cut starting there belongs
to $\mathcal L$. It must end before $k$: an end beyond $k$
would cross $(i,k)$ on the right, and an end at $k$ would force
its right-crossing $\mathcal R$-cut also to cross $(i,k)$,
contradicting laminarity. Its ending endpoint cannot be a deleted
$B$-type vertex. Therefore no such cut is selected, and no
$\mathcal R$-cut is incident to $j$, contrary to the remaining
vertex property. Reflection proves the other case. Deleting an
initial or final block creates no new internal pair on its missing
side. Since endpoints determine a cut and no component cut is an
atom, the two families together use at most $n-1-q$ original
compressed-adjacency pairs. Thus
\[
t_L+t_R\leq n-1-q.
\]
Also $2$ is not $C$-type and $m-1$ is not $B$-type, so the
boundary indicators survive the respective deletions.

We use the following edge deficit. Let a simple noncrossing graph
on $v\geq3$ circularly ordered vertices, one designated as root,
contain the full boundary path through its internal vertices.
If $t$ marked path edges form a matching and none belongs to a
triangle, then
\[
|E|\leq2v-3-t/2.
\]
Indeed, add the $z\leq2$ missing boundary edges incident to the
root. The resulting graph contains the full outer cycle. If its
bounded face lengths are $k_F$, Euler's formula gives
$2v-3-|E|=z+\sum_F(k_F-3)$. A triangular face containing a
marked edge must also contain an added edge. Each added edge is
on only one bounded face, and a triangle contains at most one
matching edge. Thus at most $z$ marked edges lie on triangular
faces. A face of length $k\geq4$ contains at most
$\lfloor k/2\rfloor\leq2(k-3)$ marked edges. Each marked path
edge is on one bounded face, proving
$t\leq z+2\sum_F(k_F-3)\leq2(2v-3-|E|)$.
For $v=2$, there is no internal path edge to mark and the ordinary
single-edge bound applies.

Mark each pair $(i,i+1)$ in $W$ with $r_i=\ell_{i+1}=2$, and
let $t$ count them. These pairs form a matching, since a shared
vertex would have neither a starting nor an ending component cut.
Their endpoints are neither $B$- nor $C$-type, so they survive
both deletions. In each reduced drawing add every missing edge
of its compressed internal path. Also add $(1,2)$ to the first
drawing when $b_L=1$, and $(m-1,m)$ to the second when $b_R=1$.
All additions are boundary edges and preserve outerplanarity.

No marked edge $(i,j)$ lies in a triangle of either augmented
drawing. A third vertex $k<i$ would require a component edge
ending at $j$, or a path edge skipping $i$; both are impossible.
A third vertex $k>j$ similarly requires a component edge starting
at $i$, or a path edge skipping $j$. The special atomic boundary
edge cannot provide such a diagonal, since $r_i=2$ excludes
$i=2$, and $\ell_j=2$ excludes $j=m-1$. There is no vertex
strictly between $i,j=i+1$. The edge deficit therefore gives,
with $N_L,N_R$ the selected component counts,
\begin{align*}
N_L+(n-p_C-1-t_L)+b_L&\leq2(n-p_C)-1-t/2,\\
N_R+(n-p_B-1-t_R)+b_R&\leq2(n-p_B)-1-t/2.
\end{align*}
The one-internal-vertex cases use the $v=2$ clause. Adding yields
\[
N_L+N_R\leq3n-1-q-b-p-t.
\]

For an interior atom, compare $1/2+r_i+\ell_{i+1}$ with $7/2$.
If neither endpoint cost is $3$, the only positive excess is $1$
for a marked $(2,2)$ pair. Otherwise the excess is at most twice
the number of costs equal to $3$. Each $C$-type vertex appears
in at most one $r_i$ position and each $B$-type vertex in at most
one $\ell_{i+1}$ position. The total interior-atom demand is
therefore at most $(7/2)q+2p+t$. Including irrelevant atoms only
increases this upper bound. Since $b\leq2$, the core demand is
at most
\begin{align*}
\sum_{S:D(S)\subseteq W}c_S
&\leq4(3n-1-q-b-p-t)+(7/2)q+2p+t+5b\\
&=12n-4-q/2+b-2p-3t\leq12n-2.
\end{align*}
Adding back $J$ costs at most $11|J|\leq12|J|$.
The empty set has demand zero, proving
Eq.~\eqref{eq:fractional_endpoint_hall} for the original $W$.

Use source-to-cut capacities $c_S$, cut-to-endpoint edges for
$i\in D(S)$ of capacity $\sum_Sc_S$, and endpoint-to-sink capacity
$12$. For any set of cut nodes, let $W$ be its endpoint neighborhood.
Its demand is at most the left side of
Eq.~\eqref{eq:fractional_endpoint_hall}, so max-flow/min-cut supplies
nonnegative allocations $f_{S,i}$ with
\[
\sum_i f_{S,i}=c_S,\quad
\sum_S f_{S,i}\leq12,\quad
f_{S,i}=0\quad\text{if }i\notin D(S).
\]
Set $w_{S,i}=f_{S,i}/c_S$. For
$e\in E_i:=E(a_{i-1},a_i)$, define
\[
s^*_e:=\alpha x_e\max\{
w_{S,i}\mathbf1_{\mathcal H_S}:i\in D(S)
\},
\]
where an empty maximum is zero. Apply this definition to every one-sided
polygon and set $s^*_e=0$ on edges outside their internal bundles, except
for the separate triangle construction below. For a triggered set $S$, the bundles
$E_i$, $i\in D(S)$, are disjoint subsets of $\delta(S)$, and each
has mass at least $1-7\eta$ by the near-cycle estimates used in
\cite[Theorem~A.12]{kko22}. Since $\sum_iw_{S,i}=1$,
\[
s^*(\delta(S))\geq
\alpha\sum_{i\in D(S)}w_{S,i}x(E_i)
\geq\alpha(1-7\eta).
\]
This gives the interior odd-cut guarantee. If a boundary cut is odd
and the corresponding polygon side is happy, that cut is not happy
and therefore triggers, giving the boundary happy-case guarantee.
Also $0\leq s^*_e\leq\alpha x_e$ on these internal bundles.
Bounding the maximum by the sum gives, without independence,
\[
\E[s^*_e]
\leq\alpha x_e\sum_Sw_{S,i}\Pr[\mathcal H_S]
\leq\alpha\eta x_e\sum_S f_{S,i}
\leq12\alpha\eta x_e.
\]
The separate triangle construction has expectation at most
$(3/2)(7\eta)\alpha x_e=(21/2)\alpha\eta x_e$ by
\cite[Lemma~A.13]{kko22}, so it satisfies the same bound.
An edge belongs to the internal bundles of at most one one-sided
polygon, and these bundles join consecutive non-root atoms.
Thus there is no further multiplicity, and the bottom-edge support
is unchanged.
\end{proof}

We keep the two repair components separate so that the one-sided vector
remains supported on bottom edges and can be absorbed by their stronger
payment decrease.

\begin{lemma}[Separated KKO repair vectors]
\label{lem:prelim_hat_pair}
Fix the exact max-entropy tree law with marginals $x$ and the KKO hierarchy
for the $\eta$-near-minimum cuts, where $0<\eta\leq1.4\cdot10^{-9}$, and let
$\beta>0$. Let
$E=E_{\mathrm g}\mathbin{\dot\cup}E_{\mathrm b}$ be a partition in which
every bottom edge belongs to $E_{\mathrm g}$, and let $s^{\mathrm{pay}}$
satisfy the deterministic inequalities in
Lemma~\ref{lem:payment_deterministic}. Assume in addition that
\[
s^{\mathrm{pay}}_e\geq-\beta x_e\quad(e\in E_{\mathrm g}),
\quad
s^{\mathrm{pay}}_e=0\quad(e\in E_{\mathrm b}).
\]
Then there are nonnegative random
vectors $r^{(2)},r^{(1)}:E\to\mathbb R_{\geq0}$ such that
\begin{enumerate}[label=(\alph*)]
\item \label{item:separated_repair_two_sided} if an $\eta$-near-minimum cut $S$ is crossed on both sides and $\delta(S)_T$ is odd, then
\[
r^{(2)}(\delta(S))\geq(2+\eta)\beta;
\]
\item \label{item:separated_repair_one_sided} if an $\eta$-near-minimum cut $S$ in the KKO one-sided classification is not the root and $\delta(S)_T$ is odd, then
\[
s^{\mathrm{pay}}(\delta(S))+r^{(1)}(\delta(S))\geq0;
\]
\item \label{item:separated_repair_expectation} for every edge,
\[
\E[r^{(2)}_e]
\leq5\frac{2+\eta}{1-\eta}\eta\beta x_e,
\quad
\E[r^{(1)}_e]
\leq12\frac{2+\eta}{1-7\eta}\eta\beta x_e;
\]
\item \label{item:separated_repair_support} pointwise, $r^{(1)}$ is supported only on bottom edges, and hence only on $E_{\mathrm g}$.
\end{enumerate}
\end{lemma}

\begin{proof}
Apply Lemma~\ref{lem:sharpened_two_sided_repair} with parameter
\[
\alpha_2:=\frac{2+\eta}{1-\eta}\beta
\]
and apply Lemma~\ref{lem:sharpened_one_sided_repair}, with its fractional
endpoint allocation and separate triangle construction, with parameter
\[
\alpha_1:=\frac{2+\eta}{1-7\eta}\beta
\]
to obtain $r^{(2)}$ and $r^{(1)}$.
Lemma~\ref{lem:sharpened_two_sided_repair} gives
\ref{item:separated_repair_two_sided} and the first expectation bound in
\ref{item:separated_repair_expectation}.
Lemma~\ref{lem:sharpened_one_sided_repair} retains the pointwise
one-sided guarantees and gives
\[
\E[r^{(1)}_e]\leq12\alpha_1\eta x_e,
\]
which is the second expectation bound in \ref{item:separated_repair_expectation}.

We verify \ref{item:separated_repair_one_sided} without invoking \cite[Theorem~B.3]{kko22} as a black box. If $S$ is a hierarchy cut whose parent is a degree cut, Lemma~\ref{lem:payment_deterministic}-\ref{item:deterministic_payment_degree} gives
\[
s^{\mathrm{pay}}(\delta(S))\geq0
\]
whenever $\delta(S)_T$ is odd.

For an interior cut or atom of a one-sided near-cycle component,
Lemma~\ref{lem:sharpened_one_sided_repair} gives
\[
r^{(1)}(\delta(S))
\geq
\alpha_1(1-7\eta)
=(2+\eta)\beta
\geq-s^{\mathrm{pay}}(\delta(S)).
\]
The last inequality is exactly where the coordinate hypothesis is used:
\[
s^{\mathrm{pay}}(\delta(S))
\geq-\beta x(\delta(S))
>-(2+\eta)\beta.
\]
The same conclusion holds for a leftmost or rightmost cut when the corresponding side is happy. Suppose instead, without loss of generality, that $S$ is leftmost and its polygon $P$ is not left happy. Let $S'$ be the union of the non-root atoms of $P$ and define
\[
F:=\delta(S)\setminus\delta(S').
\]
By \cite[Lemma~2.10]{kko22}, $x(F)\geq1-d/2$. Every edge $e\in F$ joins distinct children of $P$, and hence satisfies $p(e)=P$. Thus Lemma~\ref{lem:payment_deterministic}-\ref{item:deterministic_payment_polygon} applies. Since $A\cup F\subseteq\delta(S)$ and the worst possible contribution from the remaining $C$-edges is $(s^{\mathrm{pay}})^-(C)$,
\[
s^{\mathrm{pay}}(\delta(S))+r^{(1)}(\delta(S))
\geq
s^{\mathrm{pay}}(A)+s^{\mathrm{pay}}(F)
+(s^{\mathrm{pay}})^-(C)\geq0.
\]
The rightmost case is symmetric. If the parent is a triangle cut $P$ with
children $a_1,a_2$, use $F=E(a_1,a_2)$. Lemma~A.13 of \cite{kko22} gives the
happy-case repair. For the non-happy case, the identity
\[
x(\delta(a_1))+x(\delta(a_2))
=2x(F)+x(\delta(P))
\]
and the subtour constraints for $a_1,a_2$, together with
$x(\delta(P))\leq2+d$, give $x(F)\geq1-d/2$. Therefore
Lemma~\ref{lem:payment_deterministic}-\ref{item:deterministic_payment_polygon} applies by the same displayed
calculation. These are Types~4 and~5 in the proof of
\cite[Theorem~B.3]{kko22}, and together with the degree-parent and interior
cases they exhaust the one-sided classification.

For \ref{item:separated_repair_support}, the fractional allocation in
Lemma~\ref{lem:sharpened_one_sided_repair} assigns repair only to sets
$E(a_{i-1},a_i)$ between consecutive non-root atoms; the triangle
construction assigns repair only to $E(a_1,a_2)$. In each case the atoms
are distinct children of the same near-cycle or triangle cut. Their
smallest common hierarchy ancestor is therefore their parent, so these
edges are bottom edges. The hypothesis of the lemma makes every such
edge good.
\end{proof}

{\bf Root parity.}
The root
$R:=V\setminus\{u_0,v_0\}$ is handled by parity.  Indeed,
$x^0(\delta(u_0))=x^0(\delta(v_0))=2$ and $x^0_{e_0}=1$, so
$x(\delta(u_0))=x(\delta(v_0))=1$.  Every spanning tree has positive degree
at both vertices, so the two degree random variables, each having expectation
one, are equal to one almost surely. Connectivity prevents their two unique
tree edges from joining $u_0$ directly to $v_0$, since that would isolate
$\{u_0,v_0\}$ from the remaining vertices. Hence the two unique incident
edges both cross $\delta(R)$, and $\delta(R)_T=2$ almost surely. Thus the root
never creates an odd-cut constraint.

\subsection{Cutwise payment credit and centering}
\label{sec:final_cutwise_credit}

After adding the one-sided repair, we retain the expected decrease available on
each cut rather than immediately replacing it by a uniform edgewise bound. This
yields separate credit estimates for the two relevant cut classes.

\begin{lemma}[Aggregate expected credit on one-sided cuts]
\label{lem:final24_credit}
Fix one row in Definition~\ref{def:final24_parameters}, a threshold
$0<u\leq H$, its hierarchy with error $d\leq d_0$, and amplitude
$\beta>0$. Put
\[
\begin{aligned}
\alpha&:=2d_0,&c&:=1-\chi,&
a^{(4)}&:=p(1-c(1+\alpha)),\\
b(u)&:=a_{\rm bot}-R_1(u)u,&m(u)&:=\min\{a,b(u)\},&
V&:=1/2+k_{\rm good}h,
\end{aligned}
\]
and suppose $b(u)>0$. Define
\[
\begin{aligned}
A^{(0)}(u)&:=\frac{m(u)(2-\alpha(2+d_0))}{1+m(u)/p},\\
A^{(1)}(u)&:=a^{(4)}(3/2-h-V)+m(u)V,\\
A(u)&:=\min\{A^{(0)}(u),A^{(1)}(u)\},\quad
B(u):=(1-d_0)b(u).
\end{aligned}
\]
Let $W:=s^{\rm pay}+r^{(1)}$. On good edges,
$m(u)\beta x_e\leq-\E[W_e]\leq\beta x_e$, while $W_e=0$ on bad
edges. Every degree-parent cut has expected credit at least $\beta A(u)$;
every other relevant non-root one-sided cut has expected credit at least
$\beta B(u)$.
\end{lemma}

\begin{proof}
Write $m=m(u)$. The payment bound and bottom support of $r^{(1)}$
give credit at least $a$ on good top edges and $b(u)$ on bottom edges,
in units of $\beta x_e$. The coordinate bound gives the upper bound one.
Bad edges have neither payment nor one-sided repair. Zero-mass edges
may be omitted.

For a child $v$ of a degree parent, write
\[
U_v:=x(\delta^\uparrow(v)),\quad
D_v:=x(\delta^\uparrow(v)\cap E_{\rm b}),\quad
I_v:=x(\delta^\rightarrow(v)\cap E_{\rm g}),\quad
\ell_v:=\frac{B_v}{\beta U_vZ_v}.
\]
Here $B_v$ is the normalized upward burden in
Lemma~\ref{lem:candidate_ancestor}; its actual expectation is $pB_v$.
When $U_v=0$, set $\ell_v=0$. Matching conservation then makes every
corresponding $m_{\mathbf e,v}$ zero. Bad upward mass includes
distinguished-endpoint edges, which are never reduced.

For every child, $\ell_v\leq F_v$. If $U_v\geq\sigma$, this follows
from the ancestor estimate. Otherwise $F_v=1$ and
$B_v\leq\beta U_v$, since $sq_0<1$.
If the parent has at least four children, the stronger bound
$\ell_v\leq cF_v$ holds: the ancestor estimate handles
$U_v\geq\sigma$, and $U_v<\sigma<\epsilon_F$ gives $Z_v=2$ and
$\ell_v\leq1/2<c$.

Suppose that $v$ is incident to an internal bad bundle. Its parent
has at least four children by Lemma~\ref{lem:candidate_matching}, and
$U_v\leq V<1-\epsilon_F$. We claim, for $U_v>0$, that
\[
\ell_v\leq cF_v(1-D_v/U_v).
\]
If $U_v<\epsilon_F$, then $F_v=1$, $Z_v=2$, and
$B_v\leq\beta(U_v-D_v)$. Otherwise $F_v=1-\epsilon_B$.
The fractional case in the ancestor proof bounds each top conditional
odd probability by $0.92$, and Lemma~\ref{lem:candidate_bad_bundle_parity}
bounds every upward bottom probability by $q_{\rm bad}$.
The exact row comparisons
\[
0.92\leq c(1-\epsilon_B),\quad
sq_{\rm bad}\leq c(1-\epsilon_B),\quad
pc(1-\epsilon_B)\geq a,\quad a^{(4)}\geq a
\]
give $B_v\leq\beta cF_v(U_v-D_v)$. Division by $U_vZ_v$, with
$Z_v\geq1$, proves the claim. This includes $U_v=\epsilon_F$.

For a good horizontal bundle $\mathbf e=(v,w)$, the exact payment
formulas in \cite[Eqs.~(33) and~(37)]{kko21} and localized matching
capacity give
\[
-\frac{\E[s^{\rm pay}(\mathbf e)]}{\beta}
=p(x_{\mathbf e}-m_{\mathbf e,v}\ell_v-m_{\mathbf e,w}\ell_w)
\geq a^{(4)}x_{\mathbf e}
+pcF_v\frac{D_v}{U_v}m_{\mathbf e,v}
\]
when $v$ is incident to a bad bundle. We retain both the baseline
credit and the additional bad-mass credit. Summing over its incident
good bundles and using $\sum_{\mathbf e\ni v}m_{\mathbf e,v}=U_vZ_v$
gives
\begin{equation}\label{eq:additive_endpoint_credit}
-\frac{\E[s^{\rm pay}(\delta^\rightarrow(v)\cap E_{\rm g})]}{\beta}
\geq a^{(4)}I_v+pcF_vZ_vD_v.
\end{equation}
For $U_v=0$, the second term is zero and the same bound follows from
the baseline bound. No independence between reduction and increase
events is used.

Let $T(v):=-\E[W(\delta(v))]/\beta$. In the bad-incidence case,
internal bad mass is at most $1/2+h$, so
$I_v\geq3/2-h-U_v$. The repair is zero on horizontal top bundles;
the upward good mass has credit at least $m(U_v-D_v)$. Hence
\[
T(v)\geq a^{(4)}I_v+m(U_v-D_v)+pcF_vZ_vD_v
\geq a^{(4)}(3/2-h-U_v)+mU_v\geq A^{(1)}(u).
\]
The second step uses $F_vZ_v\geq1-\epsilon_B$ and the displayed
row comparisons. The last uses $U_v\leq V$ and $a^{(4)}\geq a\geq m$.

If $v$ has no internal bad bundle, its good boundary mass is at least
$2-D_v$, giving $T(v)\geq m(2-D_v)$. The neighboring surpluses
$F_w-\ell_w$ are nonnegative even for a three-child parent. Thus the
same exact payment identity, now using $B_v/\beta\leq U_v-D_v$ and
$F_vZ_v\geq1$, gives
\[
T(v)\geq p(F_vU_vZ_v-B_v/\beta-\alpha I_v)
\geq p(D_v-\alpha(2+d_0)).
\]
All complementary edges have nonnegative expected credit. Multiply
the good-mass bound by $p$ and this bound by $m$, then add and divide
by $p+m$. The terms in $D_v$ cancel, giving $T(v)\geq A^{(0)}(u)$.

For a remaining one-sided near-cycle cut, let $S'$ be the union of
its component's non-root atoms. By \cite[Lemma~2.10]{kko22},
$x(\delta(S)\cap\delta(S'))\leq1+d_0$; every edge in the difference
is bottom, so the bottom mass is at least $1-d_0$.
For a triangle parent $Q$ with children $a_1,a_2$, the cut identity
\[
2x(E(a_1,a_2))=x(\delta(a_1))+x(\delta(a_2))-x(\delta(Q))
\]
gives bottom mass at least $1-d_0/2\geq1-d_0$.
The bottom credit therefore gives $T(S)\geq B(u)$.
These cases exhaust the one-sided classification. The root has
almost-surely even boundary and is excluded.
\end{proof}

Different parameter rows can favor different cut classes. We mix their expected
credits on a common hierarchy before centering the resulting vector, then add the
two-sided repair.

\begin{lemma}[Uniformization after mixing cutwise credit]
\label{lem:final24_uniform}
Fix $0<u\leq H$, one geometric hierarchy for this threshold, and
amplitude $\beta>0$.
For each row with $b_j(u)>0$, use the preceding lemma on that hierarchy.
Choose nonnegative weights $w_j$ summing to one, supported on such rows,
and put
\[
\begin{aligned}
\bar m&:=\sum_jw_jm_j(u),&
A_w&:=\sum_jw_jA_j(u),&B_w&:=\sum_jw_jB_j(u),\\
M_w&:=\min\{A_w,B_w\},&
C_w&:=\frac{M_w}{(1-\bar m)(2+u)+M_w},&
\kappa_w&:=C_w-R_2(u)u.
\end{aligned}
\]
There is a random vector $z^{(u)}$ such that
$z^{(u)}_e\geq-\beta x_e$ on every edge,
$z^{(u)}(\delta(S))\geq0$ on every relevant odd $u$-near-minimum
cut, and $\E[z^{(u)}_e]\leq-\kappa_w\beta x_e$.
\end{lemma}

\begin{proof}
The hierarchy construction in Lemma~\ref{lem:final24_hierarchy}
uses only $x$, the distinguished edge, and the near-minimum threshold.
Fix it, its polygon representations, and its cut classes before choosing
the parameter rows. In particular, a cut has the same degree-parent or
remaining one-sided classification in every row. The good sets,
matchings, and reduction events may differ. Realize their auxiliary
laws conditionally on the same exact-law tree, preserving each row's
synchronized pairs.

Put $W:=\sum_jw_jW_j$, $M:=2+u$, and
\[
D:=(1-\bar m)M+M_w,\quad C:=M_w/D,\quad\lambda:=M/D.
\]
Every selected row has $0<m_j<1$ and positive cut credits. Thus $D>0$,
$0<C<1$, and
\[
\lambda(1-\bar m)+C=1,\quad\lambda M_w=CM.
\]
On a good edge in row $j$,
$(W_j)_e-\E[(W_j)_e]\geq-(1-m_j)\beta x_e$; on a bad edge both
terms are zero, so this bound still holds. Averaging gives the bound
$W_e-\E[W_e]\geq-(1-\bar m)\beta x_e$ without requiring the good
sets to agree. Define
\[
q:=\lambda(W-\E[W])-C\beta x,\quad z^{(u)}:=q+r^{(2)}.
\]
Then $q_e\geq-\beta x_e$ and $\E[q_e]=-C\beta x_e$.

On a relevant odd one-sided cut, each $W_j(\delta(S))\geq0$.
The shared cut classification gives expected credit at least
$\beta A_w$ or $\beta B_w$, hence at least $\beta M_w$.
Therefore
\[
q(\delta(S))\geq\beta(\lambda M_w-Cx(\delta(S)))
\geq\beta(\lambda M_w-CM)=0.
\]
On a relevant odd two-sided cut, the coordinate bound loses at most
$M\beta$. Add one two-sided repair from
Lemma~\ref{lem:prelim_hat_pair}, with amplitude $\beta$, to supply
this amount. Its construction depends only on the common hierarchy
and tree, not on the chosen row; no separate copy is charged per row.
It is nonnegative and costs at most $R_2(u)u\beta x_e$ in expectation.
This proves all three assertions.

The expectations and weights are analytical quantities under the exact
tree and auxiliary laws. The algorithm only samples a tree and computes
its minimum join; it does not construct these centered witnesses.
The implementation comparison is applied once after all layers are combined.
\end{proof}

\subsection{Layering and the final certificate}
\label{sec:layered_certificate}

A certificate at one threshold controls only cuts sufficiently close to minimum.
We combine finitely many such certificates on the same tree, using each cut's
excess to pay for the remaining layers.

\begin{lemma}[Finite-layer combination]
\label{lem:finite_layer_combination}
Let $0=u_0<u_1<\cdots<u_N=H$, and define
\[
\beta(u):=\frac{u}{4+2u},\quad
\Delta_i:=\beta(u_i)-\beta(u_{i-1}).
\]
Suppose that, on the same sampled tree, $V_i$ satisfies
$(V_i)_e\geq-\Delta_i x_e$ and
$\E[(V_i)_e]\leq-\kappa_i\Delta_i x_e$.
Assume also that $V_i(\delta(S))\geq0$ for every odd tree cut
not containing $e_0$ with $x(\delta(S))<2+u_i$.
Then $Z:=\sum_{i=1}^N V_i$ satisfies
\[
Z_e\geq-\beta(H)x_e,\quad
\E[Z_e]\leq-\sum_{i=1}^N\kappa_i\Delta_i x_e,
\]
and $x(\delta(S))/2+Z(\delta(S))\geq1$ for every such odd cut,
without a restriction on its $x$-value.
\end{lemma}

\begin{proof}
The coordinate and expectation bounds follow by telescoping and linearity.
For an odd cut not containing $e_0$, put $\rho=x(\delta(S))-2\geq0$
and $j:=\max\{i\in\{0,\ldots,N\}:u_i\leq\rho\}$.
Layers above $j$ are nonnegative on this cut, while the other layers give
\[
\frac{x(\delta(S))}{2}+Z(\delta(S))
\geq1+\frac\rho2-\beta(u_j)(2+\rho)\geq1,
\]
because $\beta(u_j)\leq\beta(\rho)$.
This includes grid endpoints and $\rho>H$.
Auxiliary constructions can be jointly realized conditionally on the
common tree. Neither independence between layers nor persistence of an
edge's type is needed.
\end{proof}

We finish by applying the finite-layer lemma to the prescribed threshold grid and
row-mixing schedule. The resulting exact sum gives the slack vector used in the
proof of the main theorem.

\begin{lemma}[Mixed-row layered slack certificate]
\label{lem:final24_layered}
Use Definition~\ref{def:final24_parameters}. There is a random vector
$Z$ with $Z_e\geq-\beta(H)x_e$ and
\[
\frac{x(\delta(S))}{2}+Z(\delta(S))\geq1
\]
for every cut avoiding $e_0$ with odd tree boundary. Moreover,
$\E[Z_e]\leq-B_*x_e$, where
\[
B_*:=\frac{371494773516081479041}{133333333517622666730346463630000000000}
>2.78621079751\cdot10^{-18}.
\]
\end{lemma}

\begin{proof}
Set $u_i:=iH/N$ for $0\leq i\leq N$, with $H=138217/10^{14}$ and
$N=10^6$. The hierarchy domain holds because
$H<1.4\cdot10^{-9}$ and $14H<d_0$.
Every local estimate in the thirty-four rows holds throughout $d\leq d_0$;
no former row-specific absorption endpoint is required.

Use the following explicit support schedule. A singleton has weight one.
For a pair $\{j,k\}$, put $\xi_j:=A_j(u_i)-B_j(u_i)$ and use
\[
w_j:=\frac{\xi_k}{\xi_k-\xi_j},\quad
w_k:=\frac{-\xi_j}{\xi_k-\xi_j}.
\]
The exact checks give $\xi_j\xi_k<0$ in every paired layer, so both weights
are positive and their sum is one. They satisfy $A_w=B_w$.
In every selected row $b_j(u_i)>0$, and the coefficient
$\kappa_i:=\kappa_w(u_i)$ from Lemma~\ref{lem:final24_uniform} is
positive. Define $k_i:=\lfloor10^{18}\kappa_i\rfloor$.
The support schedule and exact block sums are
\begin{center}
\begin{tabular}{r r c r}
\toprule
First layer & Last layer & Selected rows & Sum of $k_i$\\
\midrule
1&14088&$\{10,30\}$&217260955401051\\
14089&14090&$\{30,34\}$&30637421414\\
14091&44479&$\{16,34\}$&458767788390354\\
44480&89026&$\{6,16\}$&648090572825924\\
89027&134640&$\{6,11\}$&633530002969706\\
134641&166368&$\{11,20\}$&422714025435708\\
166369&184478&$\{20,25\}$&234674686479670\\
184479&195204&$\{8,25\}$&136726465521807\\
195205&302814&$\{5,8\}$&1278458890762992\\
302815&388420&$\{5,12\}$&895156800224777\\
388421&411665&$\{12,26\}$&224275150546967\\
411666&411669&$\{26,31\}$&37901962082\\
411670&435255&$\{17,31\}$&219350764086864\\
435256&435256&$\{17,21\}$&9124613525\\
435257&474337&$\{1,21\}$&345229810035612\\
474338&563639&$\{1,22\}$&703465604769660\\
563640&563661&$\{18,22\}$&158659574817\\
563662&610946&$\{18,32\}$&324318491964897\\
610947&610947&$\{27,32\}$&6505963460\\
610948&658498&$\{13,27\}$&292477716623809\\
658499&675981&$\{3,13\}$&98902504761856\\
675982&783682&$\{3,7\}$&501708911379360\\
783683&794120&$\{7,23\}$&38770847097740\\
794121&802164&$\{23,28\}$&28683268975445\\
802165&802168&$\{28,33\}$&14004100977\\
802169&810104&$\{14,33\}$&27276565132550\\
810105&839259&$\{9,14\}$&91483445017895\\
839260&895836&$\{2,9\}$&138264810144107\\
895837&952903&$\{2,15\}$&82279356325765\\
952904&981887&$\{15,24\}$&18033363879912\\
981888&981894&$\{19,24\}$&2420909725\\
981895&996713&$\{19,29\}$&3029676230309\\
996714&1000000&$\{4,29\}$&104369952282\\
\bottomrule
\end{tabular}
\end{center}
All computations use rational numbers, including the weights and floors;
Section~\ref{sec:arithmetic_details} specifies their evaluation.
Adding the thirty-three blocks gives
\[
\sum_{i=1}^{10^6}k_i=8063294099483019.
\]
No optimality assertion for the support schedule is needed.

Since $\beta'(u)=1/(2+u)^2$, the layer amplitudes satisfy
\[
\Delta_i:=\beta(u_i)-\beta(u_{i-1})
\geq\frac{H}{N(2+H)^2}.
\]
Consequently
\begin{equation}\label{eq:four_band_B}
\sum_i\kappa_i\Delta_i
\geq\frac{10^{-18}H}{N(2+H)^2}\sum_i k_i=B_*.
\end{equation}
In particular, the exact comparison gives the implementation reserve
\[
B_*-10^{-28}-2.78621\cdot10^{-18}>7.97419\cdot10^{-25}.
\]

At layer $i$, use amplitude $\Delta_i$ and mix the selected rows on
one fixed hierarchy. Jointly realize all the finite auxiliary laws
conditionally on the same sampled tree. Lemma~\ref{lem:final24_uniform}
supplies the coordinate and odd near-minimum-cut bounds, so
Lemma~\ref{lem:finite_layer_combination} proves all-cut feasibility and
the asserted expectation bound. Different layers need not have nested
hierarchies or persistent edge types. Since $\beta(H)<1/2$, the join
coordinates are nonnegative. The tour algorithm does not compute the
mixtures, auxiliary events, or centered expectations.
\end{proof}

\subsection{Exact arithmetic details}
\label{sec:arithmetic_details}

We give the arithmetic rules, finite ranges, and intermediate margins
needed to reproduce the numerical comparisons in the proof. The inputs
are precisely the thirty-four rows of Definition~\ref{def:final24_parameters}.
All subsequent quantities are computed from those rows and the formulas
cited below; the displayed rounded margins are conclusions, not
replacement inputs. The probability and conditioning arguments establish
the inequalities to be checked. The arithmetic below verifies their
scalar premises, not the underlying probabilistic assertions.

{\bf Rational arithmetic and exponential bounds.}
Represent a terminating decimal with $k$ decimal places by its integer
numerator divided by $10^k$. Perform addition, multiplication, division,
and comparisons over the rationals. Before dividing by a quantity whose
sign is used, check its strict positivity. For integers $a$ and $b>0$,
the rounding operations used in Definition~\ref{def:final24_parity} are
exactly
\[
\lfloor a/b\rfloor_{20}
=\frac{\lfloor10^{20}a/b\rfloor}{10^{20}},\quad
\lceil a/b\rceil_{20}
=-\frac{\lfloor-10^{20}a/b\rfloor}{10^{20}}.
\]
Thus downward and upward rounding enclose the input with width less
than $10^{-20}$, including when the input is negative.

For $0\leq x<42$, define
\[
t_0=1,\quad t_j=xt_{j-1}/j\ (1\leq j\leq40),\quad
L(x)=\sum_{j=0}^{40}t_j,\quad
U(x)=L(x)+\frac{xt_{40}}{41(1-x/42)}.
\]
The tail terms after $t_{41}$ have successive ratios at most $x/42$.
Consequently
\[
L(x)\leq e^x\leq U(x),\quad
E_-(x):=1/U(x)\leq e^{-x}\leq E_+(x):=1/L(x).
\]
These are rational functions, not floating-point exponential evaluations.
Use $E_-$ in a positive lower bound and $E_+$ in a positive upper bound;
reverse the choices after a negative coefficient. For an interval
$x\in[a,b]\subset[0,42)$, use
$e^{-x}\in[E_-(b),E_+(a)]$.
For general rational interval operations, addition uses the endpoint
sums, multiplication uses the minimum and maximum of the four endpoint
products, and reciprocation sends $[a,b]$ to $[1/b,1/a]$ when $a>0$.
These rules preserve inclusion. Round intermediate tail bounds upward
and updated atom bounds downward only as prescribed in
Lemma~\ref{lem:tail_profiles}; retain the parity rounding and final floors.
For the paired-event comparison use the same exponential formulas with
degree $41$ in place of $40$, and replace the denominator $1-x/42$
by $1-x/43$ accordingly.
For a nonnegative rational $v=a/b$ and integer scale $M=10^{20}$ or
$10^{30}$, choose the unique nonnegative integer $k$ satisfying
\[
k^2b\leq aM^2<(k+1)^2b.
\]
Then $k/M\leq\sqrt v<(k+1)/M$ supplies the specified directed root
enclosures using only integer comparisons.
For an upper root enclosure use $k/M$ when $k^2b=aM^2$, and
$(k+1)/M$ otherwise.
The atom-dependent coefficients require only polynomial sign tests.
For $\delta>0$ and $0<x\leq1/3$, put $q=q_L(x)$, $w=w_L(x)$ and
\[
\mathcal D_{L,\delta}(x)
=q^2\delta^2+4w\delta^3-4q^3-27w^2-18qw\delta.
\]
This is the discriminant of $wS^3+qS^2-\delta S+1$.
Its negative critical point has positive value, and its positive
minimum is nonpositive exactly when $\delta\geq d_L(x)$.
Hence its sign is positive below $X_L(\delta)$ and negative above.
The polynomials $12\mathcal D_{1/2,\delta}(x)/x^2$ and
$16\mathcal D_{1/4,\delta}(x)/x^2$ are, respectively,
\begin{align*}
&-30x^4+(72-36\delta)x^3+(3\delta^2+126\delta-9)x^2\\
&\hspace{1em}-(16\delta^3+12\delta^2+108\delta+48)x
+24\delta^3+12\delta^2,\\
&-54(\delta+1)x^3+(9\delta^2+126\delta+117)x^2\\
&\hspace{1em}-(16\delta^3+24\delta^2+72\delta+64)x
+16\delta^3+16\delta^2.
\end{align*}
Integer bisection on multiples of $10^{-30}$, after clearing positive
denominators, gives the upper enclosure
$\bar X_L=10^{-30}\lceil10^{30}X_L\rceil$; at $\delta=0$ use zero.
Replace $x_c,y_c$ by $\bar X_{1/2},\bar X_{1/4}$ in the formulas
for $J,K$ to obtain lower bounds $\bar J,\bar K$, since those
ratios decrease in $x$.
Both $c\bar J(c)$ and $c\bar K(c)$ increase with $c$.
Use them at the downward-rounded atoms before rounding tails upward
as in Lemma~\ref{lem:tail_profiles}. Concavity is required only of
the unrounded profiles, not these staircase evaluations.

For the second-atom coefficient, define
\[
P_\delta(r)=r^3+(18\delta+8)r^2
+(9\delta^2-72\delta-48)r+36\delta^3.
\]
The discriminant of
$9t^3+9rt^2-9r\delta t+r^2(4-r)$ equals $729r^3P_\delta(r)$.
For $2\leq r<4$, its negative critical value is positive;
its positive minimum is $rF_\delta(r)$. Thus $P_\delta$ has
the sign of $-F_\delta$. If $P_\delta(2)\geq0$, use $r_*=2$.
Otherwise bisection on $[2,4]$ at scale $10^{20}$ gives
$10^{-20}\lfloor10^{20}r_*\rfloor$, retaining a root on the grid
exactly. At $\delta=0$ use $r_*=4$. Take the maximum with the
exact $k_0$ as in Lemma~\ref{lem:atom_newton}.
As an independent arithmetic check, the same signs follow from
\[
P_\delta(r)=36\delta^3+9r\delta^2-18r(4-r)\delta
-4r^2(4-r)-3r(4-r)^2.
\]

In the large-window bound,
if $s\leq\sqrt{1-b}<s+10^{-30}$, use
$2bs/(1+s)\leq\phi(b)$; its error is less than $10^{-30}$.
The definitions of $T_{18}$ and the degree-$16$ parity initialization
retain their stated shorter sums. Also,
$c(x)=1-x+x^2/2-x^3/6\leq e^{-x}$ on $[0,1]$.

{\bf Probability comparisons.}
For a binomial atom at an exact mean $M$, use
\[
\mathcal B_{n,k}(M):=\binom nk(M/n)^k(1-M/n)^{n-k},
\quad \mathcal F_{n,k}(M):=\sum_{j=0}^{k}\mathcal B_{n,j}(M).
\]
Terms outside $0\leq j\leq n$ are zero. The finite dimension ranges
and mean endpoints are specified in
Lemma~\ref{lem:bernoulli_estimates}; substitute them in these two
expressions. The logarithmic inequalities in its proof cover every
larger dimension. In particular the CDF-at-two comparisons use
$3\leq n\leq99$ at $M=2.01$ and $3\leq n\leq9$ at $M=2.5$,
followed by the respective bounds valid for all $n\geq100$ and $n\geq10$.
A finite dimension calculation is not used in place of those arguments.

Here are the remaining event evaluations, with all outer conditioning
factors retained.
\begin{enumerate}[label=(\roman*)]
\item In Lemma~\ref{lem:candidate_good_threshold}, use the intervals
$[q_*,0.1]$ and $[0.1,q_c]$, with the indicated two tail profiles.
Check the atom, coefficient, ratio, degree-one, denominator, and lower
$W$ conditions at their worst endpoints, then the four restored
endpoint probabilities. In Lemma~\ref{lem:final24_competing}, use the
ordinary profile with $c=0.33$ and endpoint $0.11$ for the first branch.
Its two restored margins above $0.01330$ exceed $0.00066817$ and
$0.01180996$. The second branch retains its displayed mean and derivative
checks. The fixed-mean estimates cover all finite dimensions.
\item For the small non-half event, compute $E_{\delta_3,0}$ and $T_3$ by the
two finite updates in Definition~\ref{def:tail_profiles}. Check
$\delta_3>0$ and $2q_s>\delta_3$. At $q=0.04$ verify the atom,
coefficient, and factorial-ratio conditions in the small-event proof;
also
check $0<e-E(q)<0.02$, $1-E_+(L_0-E(q))>0.38$, and every inner
degree-one, denominator, and $0.015$ comparison printed in its proof.
Check $r+4d_0-e^2/2>0$, $A>0$, $\varphi>0.979$, and $D<0.015$.
Then evaluate $o_sG(q_s)$ and $o_sG(0.04)$, using upward tail and
excess enclosures inside decreasing expressions. The two concavity
arguments prove coverage of both continuous intervals. The large-rank
branch $q\geq0.04$ is checked separately as displayed.
\item For the large non-half event, take the minimum of its two endpoint
expressions on $[q_l,0.025]$ and its bounds on $[0.025,0.12]$ and
$[0.12,1]$. Replace each positive exponential factor by $E_-$.
The proved derivative sign excludes an interior minimum in the first range.
\item In Lemma~\ref{lem:cutoff_central_rank}, perform four mean updates
at each of $L_*,0.04,0.1$. Check that all atoms increase in $(0,1)$,
their final values exceed $0.75$, and $D_{L_*}>D_s>D_m>0$.
Compute the upper root enclosure and verify
$0<m_*<L_*<0.04$ and $m_*+m_*^2/D_{L_*}<L_*$.
Evaluate $W_0,W_1,W_2$ and every derivative and domain comparison in
the oriented-window proof. Use $D_s$ only on its small interval and
$D_m$ on the middle interval. For the direct mixed event, check both
derivatives, both fallback comparisons, and its full product $P_0P_e$.
The actual lower mean $1-3r-4d_0$ and its nonzero-tail bound are used;
no lower mean $0.994$ is required.
\item In Lemma~\ref{lem:candidate_reparam_527}, form the $21$ crossing
intervals and all $231$ ordered rectangles. For each rectangle take both
orientations and both separator branches. Compute the two zero-atom
updates at each lower crossing endpoint, the three separator updates,
all six conditioning probabilities, and the seven subset rows.
Check all atom domains, positive denominators, decreasing iterations,
probability ordering, and strict target-neighbor inequalities.
For capacity, the four-product bound may be retained whenever its complete
oriented event probability already exceeds $p$. Otherwise enumerate
the vertices as prescribed in Lemma~\ref{lem:coupled_subset_capacity}:
choose three independent subset normals and each of their eight endpoint
choices, solve the three rational linear equations, retain exactly the
solutions satisfying all seven closed subset intervals, and minimize
the four products there. Nonemptiness and positivity are checked.
Multiply by $Q_iE_-(2-i)/2$, using the degree-$41$ version of the
exponential enclosure for this step. Add the two orientations only
within the same branch, then compare to $p$.
\item For the polygon event put $q_M=(\epsilon_M-2d_0)/2$ and evaluate
$\frac{27}{4}(0.94)(q_M-10d_0)^2(1-q_M)(1-3d_0/2)$.
Its target is $P=sp$, not $p$. Check also the good-half matching
condition $\gamma_{\rm good}h<0.01330$ and the polygon-charge premises.
\end{enumerate}

The table below gives strict downward-rounded margins. Columns $S,W,D$
are the small-event, oriented-window, and direct-mixed probabilities
minus $p$, in units of $10^{-12}$. Column $G$ is the good-threshold
probability minus $\gamma_{\rm good}h$, in units of $10^{-9}$.
Column $C$ is $10^6$ times the minimum paired-event ratio to $p$ minus
one; column $B$ is the parity contradiction margin in units of $10^{-8}$.
Only the exact parameter rows and displayed formulas are inputs.
The rounded margins are not fed back into any calculation.

\begin{center}
\begin{tabular}{c r r r r r r}
\toprule
Row&$S$&$G$&$W$&$D$&$C$&$B$\\
\midrule
1&116.55539&8.29929&8.99263&8.99375&12.10133&4.76494\\
2&150.32379&8.28640&9.35568&9.35687&11.63227&4.83037\\
3&125.97651&8.26930&9.04100&9.04238&11.75221&4.77036\\
4&366438.36257&8.28699&9.25092&9.25213&11.38801&4.83052\\
5&109.52046&8.36560&8.93930&16.65899&12.12832&4.78380\\
6&104.44739&8.43372&8.90712&8.90820&12.06586&4.80786\\
7&137.66638&8.28314&9.10323&9.10435&11.57181&4.80317\\
8&5.92940&8.40021&8.90875&8.90981&11.98125&4.79638\\
9&6.39964&8.27253&9.23932&9.24051&11.69931&4.81015\\
10&5.66903&3.68091&8.92105&8.92213&12.21288&2.04463\\
11&5.69439&3.66479&8.94746&8.94855&12.17801&2.03715\\
12&5.75915&3.63861&9.01217&9.01326&12.09090&2.02550\\
13&5.83712&3.61258&9.08647&9.08759&11.94138&2.01496\\
14&5.95145&3.60075&9.18945&9.19062&11.67475&2.01371\\
15&182151.22512&3.59453&9.30774&9.30896&11.50076&2.01312\\
16&5.62509&2.88904&8.92170&8.92278&12.16328&1.58350\\
17&5.70038&2.86327&9.00405&9.00515&12.02764&1.56929\\
18&5.78240&2.84415&9.08691&9.08791&12.06074&1.55879\\
19&274656.67445&2.81722&9.30797&9.30918&11.46289&1.54431\\
20&5.65804&2.86665&8.95784&8.95893&12.17137&1.57222\\
21&5.71902&2.84882&9.02278&9.02388&12.06365&1.56237\\
22&5.75543&2.83971&9.06039&9.06150&11.99214&1.55734\\
23&5.87313&2.81727&9.17709&9.17826&11.70115&1.54512\\
24&274635.69474&2.80615&9.33445&9.33564&11.48165&1.53908\\
25&5.65719&2.86664&8.95710&8.95818&12.15669&1.57221\\
26&5.72043&2.84882&9.02396&9.02508&12.09520&1.56237\\
27&5.75502&2.83970&9.06051&9.06162&11.93247&1.55734\\
28&5.87627&2.81727&9.18044&9.18162&11.67738&1.54512\\
29&321952.77326&2.80614&9.33403&9.33521&11.46233&1.53907\\
30&5.63328&2.87591&8.93017&8.93126&12.20501&1.57733\\
31&5.71011&2.85052&9.01363&9.01475&12.07570&1.56327\\
32&5.76196&2.83641&9.06728&9.06840&11.96475&1.55549\\
33&5.88168&2.81637&9.18548&9.18663&11.68664&1.54458\\
34&5.63242&2.87590&8.92928&8.93039&12.20480&1.57731\\
\bottomrule
\end{tabular}
\end{center}
{\bf Bad-endpoint parity.}
For each row, evaluate the eight initial stages in
Definition~\ref{def:final24_parity}. Then enumerate the $128$ closed
bundle-probability cells in Lemma~\ref{lem:final24_bad_parity}.
Compute $D_i,\delta_i,L_i,q_i$ and check
\[
D_i>0,\quad 0<q_i<0.1,\quad
0<\delta_i<E_-(1.2)/2,\quad
2-2a_i+2\rho+b_i\Delta_0<1.2,\quad
1-E_+(L_i)>0.39.
\]
On each of the two tail intervals, check the endpoint atom and $J$
domains and the positive degree-one and denominator conditions.
For $[l,u]$ the lower numerator check is
$1+l-2L_i+(2-\lambda_u)3l^2/2>0$.
Evaluate the four endpoint bounds, take their minimum with $0.025$
to obtain $z_i$, and check $z_i>0$ and
$\widetilde b_i=\overline b-(1-b_i)z_i>0$.
Run the eight cell-specific stages. At each stage check that the split
argument belongs to $(0,1)$, the rank-three lower bound is positive,
and the next upper bound does not increase. The Newton coefficient is
$k(\min\{\ell_j,1\})$, not the constant two. If $\ell_j>1$, the
assumed conditional law is already impossible; the capped coefficient
also keeps every stated numerical operation defined.
The initial $\Delta_0$ values are all below $0.17$, and the minimum of
$D_i-2\widetilde B_7-\widetilde Q_7$ exceeds
$1.53907\cdot10^{-8}$ in all thirty-four rows, as recorded in column $B$.
Both bootstraps stop at the prescribed eighth stage; no convergence
assumption is used. The $128$ cells and the $21$ crossing intervals
have matching adjacent endpoints and include the endpoints of their
full domains. Their mean and tail enclosures hold throughout each
closed cell, not merely at sampled points.

{\bf Child parity, ancestor credit, and payment.}
Use the five values $y_l$ in Lemma~\ref{lem:candidate_ancestor}.
With $e_{\rm ch}=\epsilon_M/2+40d_0$, the precise upper enclosure is
\[
\widehat Q_2(y)=
\frac{1+(1-y+e_{\rm ch})^2E_+(2(1-y))}{2}+14d_0.
\]
Check the child-parity comparison in
Lemma~\ref{lem:polygon_child_parity} using $E_-$ on its right side.
Then compute $c_y=1-s\widehat Q_2(y)$ and
$\kappa_l=\min\{\varepsilon_l(1-\varepsilon_l),c_{y_l}\}$.
Check $y_l>e_{\rm ch}$, $\varepsilon_l^2>2d_0$,
$0<\kappa_1\leq\cdots\leq\kappa_5$, and the three last-level
comparisons at $z=0.35$ in that proof. The exact values in every row give
\[
\sum_{l=1}^5(\kappa_l-\kappa_{l-1})(0.35-y_l-2d_0)
>\chi(1+d_0).
\]
Each $\kappa_l$ is evaluated before any rounding. Monotonicity in $z$ and
$d$ then covers $0.35\leq z\leq0.99$. The adjacent ranges use the
endpoint expressions and the derivative of $F_d$ already printed in
Lemma~\ref{lem:candidate_ancestor}. This is a finite sum of credit
increments, so overlapping thresholds do not count the same credit twice.

The remaining matching and payment checks are direct substitutions in
Lemmas~\ref{lem:candidate_matching}, \ref{lem:candidate_ancestor}, and
\ref{lem:candidate_payment_replay}. In particular retain
$q_{\rm bad}<q_0$, $sq_0<1$, $a_{\rm bot}>a>0$,
$1-\sigma-2d_0>0$, the comparisons $J_1>J_2,J_3$, and all
row interfaces in Lemma~\ref{lem:final24_credit}.
The proofs specify the worst endpoint whenever $d$ varies in
$[0,d_0]$; one must use those monotonicity arguments, rather than infer
uniformity from a single numerical substitution.

{\bf Layer sum.}
Use the exact parameters in Definition~\ref{def:final24_parameters},
put $d=d_0$, and evaluate $u_i=iH/N$ for $1\leq i\leq N$.
First compute each row's
$\chi_j,\sigma_j,\rho_{{\rm sm},j},a_j,a_{{\rm bot},j}$ from that
definition, and $a_j^{(4)},V_j$ from
Lemma~\ref{lem:final24_credit}. Check
\[
1-\sigma_j-2d>0,\quad
0<a_j\leq a_j^{(4)},\quad a_j<a_{{\rm bot},j},\quad V_j<0.9.
\]
Let $S_i$ be the selected row set in the thirty-three-block table of
Lemma~\ref{lem:final24_layered}. For each $j\in S_i$, evaluate
\[
b_j(u_i),\ m_j(u_i),\ A_j^{(0)}(u_i),\ A_j^{(1)}(u_i),\
A_j(u_i),\ B_j(u_i)
\]
in this order, using Lemma~\ref{lem:final24_credit}.
Check $b_j(u_i)>0$, $0<m_j(u_i)<1$, and $A_j(u_i),B_j(u_i)>0$.
The denominators $1-u_i$ and $1-7u_i$ are positive since $u_i\leq H$.

Next evaluate the weights in Lemma~\ref{lem:final24_layered}, checking
$\xi_j\xi_k<0$ in every paired layer. Their sum is one and
$A_w(u_i)=B_w(u_i)$ in that case. Compute
$\overline m_i:=\bar m$, $A_w(u_i)$, $B_w(u_i)$,
$M_i:=M_w(u_i)$, $C_w(u_i)$, and $\kappa_i:=\kappa_w(u_i)$
from Lemma~\ref{lem:final24_uniform}. With
\[
D_i:=(1-\overline m_i)(2+u_i)+M_i,\quad
R_i:=\frac{5(2+u_i)u_i}{1-u_i},
\]
one has $D_i>0$ and $\kappa_i=(M_i-R_iD_i)/D_i$.
Check positivity by the exact comparison $M_i>R_iD_i$.

Write $\kappa_i=n_i/d_i$ with positive integers $n_i,d_i$.
Euclidean division determines the unique integers $k_i,e_i$ satisfying
\[
10^{18}n_i=k_id_i+e_i,\quad 0\leq e_i<d_i.
\]
Consequently
\[
\frac{k_i}{10^{18}}\leq\kappa_i<\frac{k_i+1}{10^{18}},
\quad k_i=\lfloor10^{18}\kappa_i\rfloor.
\]
This rule includes exact integer endpoints and does not require a
reduced fraction. Starting from zero, add the $k_i$ in increasing
order of $i$ within each block. All positivity, weight, and floor
comparisons are checked before including a summand, with no rounding
before the Euclidean division. This gives the thirty-three block sums printed
in Lemma~\ref{lem:final24_layered}; no additional data are required.
Their sum, Eq.~\eqref{eq:four_band_B}, and clearing positive
denominators give the stated $B_*$ and implementation reserve.

%% file: main.bbl
\begin{thebibliography}{ABCC07}

\bibitem[ABCC07]{applegate2007}
David~L. Applegate, Robert~E. Bixby, Va\v{s}ek Chv\'{a}tal, and William~J.
  Cook.
\newblock {\em The Traveling Salesman Problem: A Computational Study}.
\newblock Princeton University Press, 2007.

\bibitem[Aro98]{arora98}
Sanjeev Arora.
\newblock Polynomial time approximation schemes for {Euclidean} traveling
  salesman and other geometric problems.
\newblock {\em Journal of the ACM}, 45(5):753--782, 1998.

\bibitem[BBL09]{bbl09}
Julius Borcea, Petter Br{\"a}nd{\"e}n, and Thomas~M. Liggett.
\newblock Negative dependence and the geometry of polynomials.
\newblock {\em Journal of the American Mathematical Society}, 22(2):521--567,
  2009.

\bibitem[Chr76]{christofides76}
Nicos Christofides.
\newblock Worst-case analysis of a new heuristic for the traveling salesman
  problem.
\newblock Technical Report 388, Graduate School of Industrial Administration,
  Carnegie-Mellon University, 1976.

\bibitem[DFJ54]{dfj54}
George~B. Dantzig, D.~Ray Fulkerson, and Selmer~M. Johnson.
\newblock Solution of a large-scale traveling-salesman problem.
\newblock {\em Operations Research}, 2(4):393--410, 1954.

\bibitem[GKL24]{gkl24}
Leonid Gurvits, Nathan Klein, and Jonathan Leake.
\newblock From trees to polynomials and back again: New capacity bounds with
  applications to {TSP}.
\newblock In {\em 51st International Colloquium on Automata, Languages, and
  Programming ({ICALP} 2024)}, volume 297 of {\em Leibniz International
  Proceedings in Informatics ({LIPIcs})}, pages 79:1--79:20. Schloss
  Dagstuhl--Leibniz-Zentrum fuer Informatik, 2024.
\newblock Full version: arXiv:2311.09072v2.

\bibitem[GSS11]{oss11}
Shayan~Oveis Gharan, Amin Saberi, and Mohit Singh.
\newblock A randomized rounding approach to the traveling salesman problem.
\newblock In {\em Proceedings of the 52nd Annual {IEEE} Symposium on
  Foundations of Computer Science ({FOCS} 2011)}, pages 550--559. {IEEE}, 2011.

\bibitem[HK70]{heldkarp70}
Michael Held and Richard~M. Karp.
\newblock The traveling-salesman problem and minimum spanning trees.
\newblock {\em Operations Research}, 18(6):1138--1162, 1970.

\bibitem[KKO20]{kko20}
Anna~R. Karlin, Nathan Klein, and Shayan {Oveis Gharan}.
\newblock An improved approximation algorithm for {TSP} in the half integral
  case.
\newblock In {\em Proceedings of the 52nd Annual {ACM} {SIGACT} Symposium on
  Theory of Computing ({STOC} 2020)}, pages 28--39. Association for Computing
  Machinery, 2020.

\bibitem[KKO21]{kko21}
Anna~R. Karlin, Nathan Klein, and Shayan {Oveis Gharan}.
\newblock A (slightly) improved approximation algorithm for metric {TSP}.
\newblock In {\em Proceedings of the 53rd Annual {ACM} {SIGACT} Symposium on
  Theory of Computing ({STOC} 2021)}, pages 32--45. Association for Computing
  Machinery, 2021.
\newblock Full version: arXiv:2007.01409v6.

\bibitem[KKO22]{kko22}
Anna~R. Karlin, Nathan Klein, and Shayan {Oveis Gharan}.
\newblock A (slightly) improved bound on the integrality gap of the subtour
  {LP} for {TSP}.
\newblock In {\em Proceedings of the 63rd {IEEE} Annual Symposium on
  Foundations of Computer Science ({FOCS} 2022)}, pages 832--843. {IEEE}, 2022.
\newblock Full version: arXiv:2105.10043v3.

\bibitem[KKO23]{kko23}
Anna~R. Karlin, Nathan Klein, and Shayan {Oveis Gharan}.
\newblock A deterministic better-than-$3/2$ approximation algorithm for metric
  {TSP}.
\newblock In {\em Integer Programming and Combinatorial Optimization ({IPCO}
  2023)}, volume 13904 of {\em Lecture Notes in Computer Science}, pages
  261--274. Springer, 2023.

\bibitem[KLS15]{karpinski15}
Marek Karpinski, Michael Lampis, and Richard Schmied.
\newblock New inapproximability bounds for {TSP}.
\newblock {\em Journal of Computer and System Sciences}, 81(8):1665--1677,
  2015.

\bibitem[Mit99]{mitchell99}
Joseph S.~B. Mitchell.
\newblock Guillotine subdivisions approximate polygonal subdivisions: A simple
  polynomial-time approximation scheme for geometric {TSP}, $k$-{MST}, and
  related problems.
\newblock {\em SIAM Journal on Computing}, 28(4):1298--1309, 1999.

\bibitem[MS11]{momke_svensson11}
Tobias M\"omke and Ola Svensson.
\newblock Approximating graphic {TSP} by matchings.
\newblock In {\em Proceedings of the 52nd Annual {IEEE} Symposium on
  Foundations of Computer Science ({FOCS} 2011)}, pages 560--569. {IEEE}, 2011.

\bibitem[Muc12]{mucha12}
Marcin Mucha.
\newblock $13/9$-approximation for graphic {TSP}.
\newblock In {\em 29th International Symposium on Theoretical Aspects of
  Computer Science ({STACS} 2012)}, volume~14 of {\em Leibniz International
  Proceedings in Informatics ({LIPIcs})}, pages 30--41, 2012.

\bibitem[Ser78]{serdyukov78}
Anatoliy~I. Serdyukov.
\newblock O nekotorykh ekstremal'nykh obkhodakh v grafakh.
\newblock {\em Upravlyaemye Sistemy}, 17:76--79, 1978.

\bibitem[SV14]{sebo_vygen14}
Andr\'{a}s Seb\H{o} and Jens Vygen.
\newblock Shorter tours by nicer ears: $7/5$-approximation for the graph-{TSP},
  $3/2$ for the path version, and $4/3$ for two-edge-connected subgraphs.
\newblock {\em Combinatorica}, 34(5):597--629, 2014.

\bibitem[SV19]{straszak2019maximum}
Damian Straszak and Nisheeth~K. Vishnoi.
\newblock Maximum entropy distributions: Bit complexity and stability.
\newblock In {\em Proceedings of the Thirty-Second Conference on Learning
  Theory}, volume~99 of {\em Proceedings of Machine Learning Research}, pages
  2861--2891, 2019.

\bibitem[SWZ14]{schalekamp14}
Frans Schalekamp, David~P. Williamson, and Anke~van Zuylen.
\newblock $2$-matchings, the traveling salesman problem, and the subtour {LP}:
  A proof of the boyd--carr conjecture.
\newblock {\em Mathematics of Operations Research}, 39(2):403--417, 2014.

\bibitem[Wil90]{williamson90}
David~P. Williamson.
\newblock Analysis of the {Held--Karp} heuristic for the traveling salesman
  problem.
\newblock Master's thesis, Massachusetts Institute of Technology, Cambridge,
  MA, June 1990.
\newblock Also available as MIT LCS Technical Report TR-479.

\bibitem[Wol80]{wolsey80}
Laurence~A. Wolsey.
\newblock Heuristic analysis, linear programming and branch and bound.
\newblock In {\em Mathematical Programming Studies}, volume~13, pages 121--134.
  Springer, 1980.

\end{thebibliography}
